%% file: REVISION_CLEAN_SMC_Context_Dependent.tex
\documentclass[aap]{imsart}
\RequirePackage[colorlinks,citecolor=blue,urlcolor=blue]{hyperref}

\usepackage[utf8]{inputenc}
\usepackage[english]{babel}
\usepackage{amsmath,amssymb,amsthm,verbatim,mathtools,bbm,mathrsfs,float,subcaption}
\usepackage[numbers,sort&compress]{natbib}
\usepackage{algorithm,algpseudocode}
\usepackage{enumitem}
\usepackage{url}

\usepackage{xfrac}

\usepackage{tikz} 
\usetikzlibrary{automata}
\usetikzlibrary{calc,decorations.pathreplacing, arrows.meta, positioning}

\usepackage{float}
\usepackage[md]{titlesec}
\usepackage{ulem}[normalem]
\usepackage{soul,cancel}

\newcommand{\me}{\mathrm{e}}
\newcommand{\rank}[1]{\operatorname{rank}(#1)}

\DeclareMathOperator{\A}{\mathcal{A}}
\DeclareMathOperator{\scrA}{\mathscr{A}}
\DeclareMathOperator{\scrD}{\mathscr{D}}
\DeclareMathOperator{\E}{\mathbb{E}}
\DeclareMathOperator{\V}{\mathbb{V}}
\DeclareMathOperator{\Prob}{\mathbb{P}}
\DeclareMathOperator{\tP}{\text{P}}
\DeclareMathOperator{\Pa}{\mathcal{P}}
\DeclareMathOperator{\scrP}{\mathscr{P}}

\DeclareMathOperator{\Gap}{\mathrm{SpecGap}}

\DeclareMathOperator{\tT}{\text{T}}
\DeclareMathOperator{\tK}{\text{K}}

\DeclareMathOperator{\tC}{\text{C}}
\DeclareMathOperator{\tG}{\text{G}}
\DeclareMathOperator{\tA}{\text{A}}

\newcommand{\norm}[1]{\left\lVert#1\right\rVert}

\newcommand{\abs}[1]{\left|#1\right|}
\newcommand{\iidsim}{\stackrel{\mbox{\tiny{iid}}}{\sim}}
\newcommand{\ind}{\mathbbm{1}}
\newcommand{\defeq}{\coloneqq}

\newtheorem{assump}{Assumption}

\newtheorem{theorem}{Theorem}
\newtheorem{proposition}{Proposition}
\newtheorem{lemma}{Lemma}
\newtheorem{definition}{Definition}
\newtheorem*{definition*}{Definition}
\newtheorem*{notation*}{Notation}
\newtheorem{corollary}{Corollary}
\newtheorem{remark}{Remark}
\theoremstyle{plain}

\newcommand{\bb}{\mathbf{b}}
\newcommand{\s}{\mathbf{s}}
\newcommand{\bt}{\mathbf{t}}
\newcommand{\x}{\mathbf{x}}
\newcommand{\tx}{\tilde{x}}

\newcommand{\tpsi}{\tilde{\psi}}
\newcommand{\y}{\mathbf{y}}

\newcommand{\calA}{\mathcal{A}}
\newcommand{\C}{\mathcal{C}}
\newcommand{\D}{\mathcal{D}}
\newcommand{\calE}{\mathcal{E}}
\newcommand{\G}{\mathcal{G}}
\newcommand{\I}{\mathcal{I}}

\newcommand{\Smut}{\mathcal{S}}
\newcommand{\calS}{\mathcal{S}}
\newcommand{\X}{\mathcal{X}}
\newcommand{\Db}{\Delta \beta}
\newcommand{\bigO}{\mathcal{O}}

\newcommand{\Q}{\mathbf{Q}}
\newcommand{\tQ}{\tilde{\mathbf{Q}}}

\newcommand{\tg}{\tilde{\gamma}}
\newcommand{\gmin}{\gamma_{\min}}
\newcommand{\gmax}{\gamma_{\max}}

\newcommand{\tgmin}{\tg_{\min}}
\newcommand{\tgmax}{\tg_{\max}}

\makeatletter
\newtheorem{repeatthm@}{Lemma}

\makeatother

\usepackage{thm-restate}

\makeatletter
\newcommand{\mytag}[2]{%
  \text{#1}%
  \@bsphack
  \begingroup
    \@onelevel@sanitize\@currentlabelname
    \edef\@currentlabelname{%
      \expandafter\strip@period\@currentlabelname\relax.\relax\@@@%
    }%
    \protected@write\@auxout{}{%
      \string\newlabel{#2}{%
        {#1}%
        {\thepage}%
        {\@currentlabelname}%
        {\@currentHref}{}%
      }%
    }%
  \endgroup
  \@esphack
}
\makeatother

\begin{document}

\begin{frontmatter}

\title{Improved Bounds for Context-Dependent Evolutionary Models Using Sequential Monte Carlo}
\runtitle{Improved SMC Bounds for Context-Dependent Evolutionary Models}

\begin{aug}
\author[A]{\fnms{Joseph}~\snm{Mathews}\ead[label=e1]{joseph.mathews@duke.edu}}
\and
\author[B]{\fnms{Scott C.}~\snm{Schmidler}\ead[label=e2]{scott.schmidler@duke.edu}\orcid{0000-0000-0000-0000}}
\address[A]{Department of Statistical Science, Duke University \printead[presep={,\ }]{e1}}

\address[B]{Department of Statistical Science,
Duke University \printead[presep={,\ }]{e2}}
\end{aug}

%





\begin{abstract} 
Statistical inference in evolutionary models with site-dependence is a long-standing challenge in phylogenetics and computational biology. We consider the problem of approximating marginal sequence likelihoods under dependent-site models of biological sequence evolution. We prove an upper bound on the mixing time for a Markov chain Monte Carlo algorithm that samples the conditional distribution over latent sample paths, when the chain is initialized with a warm start.  We then introduce a sequential Monte Carlo (SMC) algorithm for approximating the marginal likelihood, and show that our mixing time bound can be combined with recent importance sampling and finite-sample SMC results to obtain bounds on the finite sample approximation error of the resulting estimator. Our results show that the proposed SMC algorithm yields an efficient randomized approximation scheme for many practical problems of interest, and offers a significant improvement over a recently developed  importance sampler for this problem. Our approach combines recent innovations in obtaining bounds for MCMC and SMC samplers, and may prove applicable to  other problems of approximating marginal likelihoods and Bayes factors.
%
%
\end{abstract}



\end{frontmatter}

\section{Introduction}
\label{sec:Intro}
Let $\x = (x_1,x_2,\ldots,x_n)$ and $\y = (y_1,y_2,\ldots,y_n)$ denote two DNA sequences. A fundamental quantity in phylogenetics is the probability that $\x$ transitions to $\y$ under a given model of DNA evolution. Calculation of  sequence transition probabilities is required for evaluating (marginal) likelihoods in a wide variety of statistical inference problems: the reconstruction of phylogenetic tree topologies \cite{Felenstein:1985,MrBayes:2012,Felenstein:1973}; the estimation of divergence times (branch lengths) \cite{Sanderson:1997,Thorne:1998,Kishino:2001}, mutation model parameters \cite{Rodriguez:1990,Yang:1994}, and selection coefficients \cite{Halpern:1998,Yang:2008}; and the reconstruction of ancestral sequences \cite{Pagel:2004,Yang:1995}, to name just a few. Let $\x_t = (x_1(t),\ldots,x_n(t))$ denote the state of the sequence at time $t$. Standard evolutionary models assume each site $x_i(t)$ evolves according to a continuous-time Markov chain (CTMC) with rate matrix $\Q$ \cite{Felenstein:1973}. Typically, the processes $x_i(t)$ and $x_j(t)$
are assumed to evolve independently for $i \neq j$, for computational tractability. Under this \textit{independent site model} (ISM) assumption, transition probabilities are straightforward to compute since they factor into a product of transition probabilities at each site: 
\begin{align}
\label{eqn:typical MargLik}
\Pr(\x_T = \y \mid  \x_0 = \x) = \prod^n_{i=1} p_{(T,\Q)}(y_i \mid x_i) =  \prod^n_{i=1} (e^{T\Q})_{x_i,y_i}.
\end{align}
%
Some well-known choices for $\Q$ include the Jukes-Cantor (JC69) \cite{Jukes:1969} and generalized time reversible (GTR) \cite{Tavare:1986} models. This independence assumption is critical to the tractability of computations for reconstructing phylogenetic trees and parameter estimation more generally (e.g. Felsenstein's pruning algorithm \cite{Felenstein:1973}). 
While independent site models are appealing in their simplicity and computational convenience, they fail to capture known important features of biological evolution that create dependence among sites; examples include CpG di-nucleotide mutability \cite{Pederson:1998}, structural constraints in RNA and proteins \cite{Robinson:2003}, and enzyme-driven somatic hypermutation in B-cell affinity maturation \cite{Wiehe:2018,Mathews:2023b}.

A variety of \textit{dependent} site models (DSMs) have been proposed to relax this independence assumption \cite{Robinson:2003,Siepel:2004,Jensen:2000,Larson:2020,Hwang:2004,VonHaeseler:1998,Christensen:2005,Arndt:2005,Lunter:2004}, incorporating varying amounts of dependence. Codon models \cite{Goldman:1994} allow individual nucleotide substitution rates to depend on sites within the same codon, but still assume independence among the codons themselves. Jensen and Pedersen \cite{Jensen:2000,Jensen:2001} describe a Markov random field model where the substitution rate at a given site depends on its  
neighboring codons. \citet{Robinson:2003} give a model of protein evolution that incorporates dependencies among codons distant in the sequence, based on their spatial proximity in the protein tertiary structure. 
However, computing marginal sequence likelihoods exactly under these models of site dependence is difficult or intractable since the corresponding likelihood no longer factors.

As a result, Markov chain Monte Carlo (MCMC) algorithms have been proposed which address statistical inference in these problems by sampling unobserved sequence evolution paths from $\x$ to $\y$, with the desired transition probability obtained by marginalization over all such paths \cite{Jensen:2000,Robinson:2003, Hwang:2004,Rodrigue:2005,Rodrigue:2006,Hobolth:2009,HobolthThorne:2014,Li:2024}.
However, MCMC has important disadvantages for use in evaluating likelihoods within iterative sampling (Bayesian) or optimization (MLE) algorithms, due to its inherently serial nature \cite{Vandewerken:2013}, the need to assess convergence empirically \cite{Gelman:1992,Cowles:1996}
and the difficulty of doing so in high-dimensions \cite{Brooks:1998,Vandewerken:2017},
and the rarity of available quantitative mixing time and approximation error bounds \cite{Rosenthal:1995,Jones:2001}. Recently  \citet{Mathews3:2025} proposed an alternative approach based on importance sampling, using an ISM as an instrumental distribution. This approach is attractive in its ability to leverage the substantial body of existing phylogenetics software, which often provides the ability to sample evolutionary paths under the site independence assumption. However, while the sample complexity of this importance sampler (IS) grows much slower than the problem dimension (sequence length $n$), the complexity nonetheless grows exponentially in the number $r$ of observed mutations, rendering the importance sampling complexity prohibitively large for many applied problems of interest.

Our results are two-fold. First, we establish an 
upper bound on the mixing time for a component-wise Metropolis algorithm for this problem \cite{Robinson:2003,Li:2024} under a warm start. Although this algorithm has been used in applications, \cite{Robinson:2003,Li:2024}, to our knowledge this bound constitutes the first rigorous convergence rate analysis. The main technical difficulty addressed in doing so is the failure of the density ratio between the ISM and DSM models $\mu$ and $\pi$ to be uniformly bounded. This is because there is no limit on the number of possible unobserved jumps (mutations followed by subsequent reversion mutations) along any endpoint-conditioned path from $\x$ to $\y$. Thus to establish our result, we bound the \textit{approximate spectral gap} of \citet{Atchade:2021}, which enables us to consider the spectral gap of the MCMC chain restricted to a high probability subset of the state space, and combine this analysis with a bound on the moment generating function of the mutation count process.

An important practical limitation of this result is the reliance on a warm start, which can be difficult to guarantee in practice.  
Moreover, estimation of marginal likelihoods directly from MCMC output is notoriously difficult \cite{Wolpert:2012,Fourment:2020,Chib:1995}, making MCMC undesirable for many applications of DSMs (e.g. phylogenetic tree reconstruction \cite{Li:2025b}) which require repeated calculation of marginal likelihoods during iterative sampling or optimization. 
Both of these issues are addressed by our second result.

Our second contribution is a sequential Monte Carlo (SMC) algorithm for approximating marginal sequence likelihoods under DSMs, along with corresponding finite sample error bounds on the resulting estimator, establishing a randomized approximation scheme for this problem with a sample complexity that significantly improves on previous results \cite{Mathews3:2025}, in some cases providing an \textit{exponential} improvement in sample complexity. The SMC algorithm proceeds by sequentially sampling endpoint-conditioned paths from a sequence of DSMs $\pi_0, \pi_1, \ldots,\pi_{V} = \pi$ with the context-dependence  `tempered', such that the initial distribution $\pi_0$  is an ISM. We derive error bounds under mutation models exhibiting \textit{neighborhood context-dependence}, in which mutation rates at each site are allowed to depend on other sites in a local neighborhood; such models nevertheless lead to \textit{global} dependence among the marginal processes at all $n$ sites in the sequence.



This result combines our newly-obtained warm-start mixing time bound for the Metropolis algorithm with recent results of \citet{Marion:2018b}, who showed that a warm-start mixing time bound for the mutation (Markov) kernel in an SMC algorithm suffices to establish SMC error bounds, provided the $\chi^2$-divergence between any two intermediate distributions $\pi_v$ and $\pi_{v-1}$ is uniformly bounded. 
This enables us to obtain the bounds here by establishing a uniform bound on the $\chi^2$-divergence in addition to the warm-start mixing time bound.

The remainder of this paper is organized as follows: Section~\ref{sec: background} establishes notation and introduces the two algorithms studied here, a previously-developed component-wise  Metropolis kernel and our SMC algorithm; Section~\ref{sec:MainResults} presents the main results of this paper, including the MCMC mixing time and SMC complexity bounds; Sections~\ref{sec: mixing time bound}
and~\ref{sec: SMC} 
give the proofs of the MCMC mixing time bound and SMC complexity bounds, respectively; 
and Section~\ref{sec:Conclusion} summarizes our results and discusses future directions. Some technical results needed in the proofs are deferred to Appendix~\ref{Appdx:SMC Results}.

\section{Background and Notation}\label{sec: background}

\subsection{Models of Molecular Evolution}
\label{sec:Models of DNA Evolution}
Let $\x = (x_1,\ldots,x_n)$ denote a sequence where $x_i \in \scrA$ for some alphabet $\scrA$ (e.g. $\scrA = \{\tA,\tG,\tC,\tT\}$) of size $a := |\scrA|$. Let 
\begin{align*}
    \tx_i = (x_{i_1},\ldots,x_{i_{k/2}},x_i,x_{i_{k/2 + 1}},\ldots,x_{i_k})
\end{align*}
denote the \textit{context} of site $x_i$ and $\C_i$ be the set of sites lying in the context of site $i$. 
The case $k = 0$ corresponds to an independent site model. We assume that sites evolve according to a time-inhomogeneous CTMC, where
\begin{align}
\label{eqn:CD Rates}
\tg_i(b; \tx_i) = \gamma_i(b; x_i)\phi(b; \tx_i) \quad \text{ for } \; b \in \scrA \setminus x_i
\end{align}
%
is the (context-dependent) rate at which $x_i \in \scrA$ mutates to $b$, with the context-dependency given by the multiplier $\phi: \scrA^{k+1} \rightarrow (0,\infty)$ 
and the context-independent rate by $\gamma_i: \scrA^2 \rightarrow (0,\infty)$
The subscript on the rates indicates a possible dependence on the site at which the mutation occurs. Let $\tg_i(\cdot; \tx_i) = 
\sum_{b \neq x_i}  \tg_i(b; \tx_i)$ 
denote the rate at which site $i$ exits state $x_i$,
and $\tg(\cdot ; \x) =  \sum^n_{i=1} \tg_i(\cdot; \tx_i)$ the total rate at which sequence $\x$ mutates.
All of our results are stated under the standard assumption that multiple substitutions cannot occur simultaneously, a natural one for most sequence evolution models.
%
%
\paragraph*{\textbf{Example 1: Models of DNA evolution}} 
\label{sec:CpGIntro}
CpG models are DSMs that have been used to account for low observed CG frequencies across codon boundaries in lentiviral genes \cite{Jensen:2000} and mammalian genomes \cite{Hwang:2004} with rates
\begin{align}
\label{eqn:CpG Island Rates}
\tg_i(b; \tx_i) = \gamma_i(b; x_i) \lambda^{\ind_{\text{CG}}(x_{i - 1},x_i) + \ind_{\text{CG}}(x_i,x_{i + 1})},
\end{align}
where $\lambda \in (0,\infty)$ is a constant reflecting the relative bias against formation of CG pairs across codon boundaries \cite{Lunter:2004, Arndt:2005, Christensen:2005}.
\paragraph*{\textbf{Example 2: Models of antibody maturation}} 
DSMs have been used to account for sequence-context-dependent somatic hypermutation (SHM) patterns in affinity maturation of antibody sequences \cite{Wiehe:2018}, with 
$\gamma_i \equiv 1$ and $\phi$ corresponding to a set of $4^5$ context-dependent rates (i.e. $k=4$) estimated from data \cite{Yaari:2013}.
\paragraph*{\textbf{Example 3: Structure-dependent evolution of proteins}}
DSMs have been used to model dependence  between codons arising from sequence-structure compatibility in protein evolution \cite{Robinson:2003} using a codon model with context-dependent parameter
\begin{align*}
\phi(y; \tx_i) = \phi_{\text{SA}}(y; \tx_i)\phi_{\text{SC}}(y; \tx_i),
\end{align*}
where $y$ is a codon and $\phi_{\text{SA}}$ and $\phi_{\text{SC}}$ are measures of energetic compatibility of the encoded amino acid with the conserved 3D protein structure, specified in terms of solvent accessibility and 
specific intra-sequence sidechain interactions. Note that $\tx_i$ are not necessarily contiguous in the DNA sequence, allowing for long-range dependence in sequence positions arising from 3D structure.

Let $\x_t = (x_1(t),\ldots,x_n(t))$ denote the state of the sequence at time $t$. We are interested in the calculation of probabilities of the form
\begin{align*}
\text{Pr}(\x_{T} = \y \mid \x_{0} = \x) = (e^{T\tQ})_{\x,\y} := p_{(T,\tQ)}(\y \mid \x),
\end{align*}
where $T$ is a fixed observation time, $\tQ$ the $a^n \times a^n$ rate matrix defined by the context-dependent rates \eqref{eqn:CD Rates} of the Markov process operating on the space of all sequences:
\begin{align}
\label{eqn: tilde Q expanded definition}
\tQ_{\x,\x^{\prime}} = 
\begin{cases}
\tg_{i}(b; \tilde{x}_{i}) & \text{ for } \text{d}_{\text{H}}(\x,\x^{\prime}) = 1 \text{ and }x^{\prime}_{i} = b \neq x_{i}   \\
-\tg(\cdot; \x) & \text{ for } \text{d}_{\text{H}}(\x,\x^{\prime}) = 0 \\
0 & \text{ for } \text{d}_{\text{H}}(\x,\x^{\prime}) > 1,
\end{cases}  
\end{align}
and $(e^{T\tQ})_{\x,\y}$ denotes the element of the matrix $e^{T\tQ}$ corresponding to the sequences $\x$ and $\y$. However, direct computation of $e^{T\tQ}$ is intractable as $\rank{\tQ}$ grows exponentially in $n$. Alternatively, we can write $p_{(T,\tQ)}(\y\mid\x)$ as a marginalization over latent \textit{paths} that start in $\x$ and end in $\y$ at time $T$. Specifically, let
%
\begin{align*}
\Pa =   (m,t^1,\ldots,t^m,s^1,\ldots,s^m,b^1,\ldots,b^m )
\end{align*}
denote a path of \textit{length} $m$,
where $m \in \{0,1,\ldots \}$ is the number of mutations occurring along the path, $t^1,\ldots,t^m \in \mathbb{R}_+$ are the times of the mutation events satisfying $ t_0 = 0 < t^1 < \ldots < t^m < T$, $s^1,\ldots,s^m \in \{1,\ldots,n\}$ the sites at which the mutations occur, and $b^1,\ldots,b^m \in \scrA$ are the values of the base changes. At times, we will make the length of the path explicit by writing $\Pa^l$ and letting $\scrP^l$ denote the set of all length $l$ paths. Let 
\begin{alignat*}{2}   
\x^j    &:=&\; \x(t^j)     &= \x\bigl(j; s^1,\ldots,s^m,b^1,\ldots,b^m,m\bigr)\\
\tx^j_i &:=&\; \tx_i(t^j)  &= \tx_i\bigl(j; s^1,\ldots,s^m,b^1,\ldots,b^m,m\bigr)
\end{alignat*}
%
denote the sequence and context at site $i$, respectively, following the $j^{\text{th}}$ jump along a given path, i.e. in the interval $t\in[t_j,t_{j+1})$, and let $\Delta^t(j) := t^{j+1} - t^j$ with $\Delta^t(m) := T - t^m$ be the inter-arrival times between jumps. Let $\scrP = \cup^{\infty}_{l=0} \scrP^l$ denote the set of all such paths. Let $\nu^l := \nu^l_t \otimes \nu_s^l \otimes \nu^l_b$, where $\nu^l_t$ denotes the Lebesgue measure on $[0,T]^l$, and $\nu_s^l$ and $\nu^l_b$ the counting measures on $\{1,\ldots,n\}^l$ and $\scrA^l$, respectively, and define the measure $\nu(d\Pa) := \sum^{\infty}_{l=0} \ind_l(dm) \nu^l(d\bt^l, d\s^l, d\bb^l)$. Then we can write 
\begin{align}
\label{eqn:Transition Probability}
p_{(T,\tQ)}(\y \mid  \x)  = \sum^{\infty}_{l=0}\int_{\scrP^l} \tP_{(T,\tQ)}(\y, \Pa \mid \x) \nu^l(d\Pa^l) = \int_{\scrP} \tP_{(T,\tQ)}(\y, \Pa \mid \x) \nu(d\Pa) ,
\end{align}
%
where the conditional joint density of a path from $\x$ ending in $\y$ is given by 
\begin{align}
\label{eqn:Path Density}
\tP_{(T,\tQ)}(\y,\Pa \mid \x) := \left[\prod^{m(\Pa)}_{j=1} \tg_{s^j}(b^j; \tx^{j-1}_{s^j})
\me^{-\Delta^t(j-1) \tg(\cdot; \x^{j-1})}
 \right] \me^{-\Delta^t(m)\tg(\cdot; \y) } \ind_{\x^m = \y}(\Pa),
\end{align}
if the times satisfy the ordering constraint $0<t^1<\ldots<t^m<T$, and zero otherwise. We let $r = d_{\text{H}}(\x,\y)$ denote the Hamming distance between $\x$ and $\y$, and $\calS = \{i : y_i \neq x_i \}$ denote the set of \textit{observed} mutated sites. Note that \eqref{eqn:Path Density} is zero unless $\Pa \ni \calS$, so each endpoint conditioned path $\Pa$ contains $r$ \textit{required} jumps and $m(\Pa) - r$ \textit{extra} jumps.

It follows from \eqref{eqn:Transition Probability} that $p_{(T,\tQ)}(\y \mid \x)$ can be approximated by Monte Carlo integration by sampling from the distribution
%
\begin{align*}
\pi(\Pa \mid \x,\y) :=  \frac{\tP_{(T,\tQ)}(\y, \Pa \mid \x)}{p_{(T,\tQ)}(\y \mid \x)},
\end{align*}
where we have suppressed the dependence of $\pi$ on $T$ for brevity; hereafter $T$ will be assumed fixed.
\begin{remark}
At times we will abuse notation by using the same symbol for both a probability measure and its density with respect to $\nu$, e.g., the density of $\pi$ with respect to $\nu$ is written as  $\pi(\Pa \mid \x,\y)$. It will also often be convenient to leave conditioning on $\x$ and $\y$ implicit, writing $\pi(\Pa)$ and $\mu(\Pa)$ in place of $\pi(\Pa \mid \x,\y)$ and $\mu(\Pa \mid \x,\y)$.
\end{remark}
However, generating samples from the joint distribution $\pi(\Pa \mid \x, \y)$ is not straightforward: we must sample the evolution of all $n$ sites jointly such that the endpoint constraint $\x_T=\y$ is satisfied at time $T$ but, as noted above, constructing the rate matrix of the joint process on the space of all sequences (of size $a^n$) is intractable for even moderate $n$. However, under an \textit{independent} site model (ISM), paths can be sampled efficiently on a site-by-site basis by specialized algorithms  \cite{Hobolth:2009}, a fact which we will take advantage of below.

\subsection{MCMC for Endpoint-Conditioned Paths}
\label{sec: MCMC definition}
Sampling $\pi(\Pa \mid \x, \y)$ under site dependence can be performed by MCMC \cite{Jensen:2000,Robinson:2003, Hwang:2004,Rodrigue:2005,Rodrigue:2006,Hobolth:2009,HobolthThorne:2014,Li:2024}.  We begin with a
simple component-wise Metropolis algorithm which updates paths one site at a time 
using an ISM as a proposal distribution for endpoint-conditioned paths, and accepting or rejecting according the Metropolis criteria under the DSM \cite{Robinson:2003,Li:2024}. Below we will consider a modification of this chain which uses \textit{blocked} site updates.

The ISM is defined as follows. Let $\Q_i= (\gamma_i(y; x))$ for $x,y\in \scrA$ be an $a \times a$ rate matrix corresponding to the CTMC at site $i$ (see \eqref{eqn:CD Rates}, with $\phi \equiv 1$). Consider the endpoint-conditioned distribution
\begin{align}
\label{Eqn:ISM}
\mu(\Pa \mid \x, \y) \propto \tP_{(T,\Q)}(\y, \Pa \mid \x),
\end{align}
with rate matrix $\Q^{(n)} = \mathbf{I}_a \otimes \Q^{(n-1)} + \Q_n \otimes \mathbf{I}_{a^{n-1}}$
where $\Q^{(1)} = \Q_1$ and 
$\mathbf{I}_a$ is the $a$-dimensional identity matrix. The density $\tP_{(T,\Q)}(\y, \Pa \mid \x)$ is given by \eqref{eqn:Path Density} but with $\phi \equiv 1$. In this case the joint density \eqref{eqn:Path Density} can be factored by site. Let 
\begin{align*}
\Pa_i = (m_i,t^1_i,\ldots,t^{m_i}_i,b^1_i,\ldots,b^{m_i}_i)
\end{align*}
denote the path at site $i$ defined by $\Pa$. That is, 
$(t^1_i,\ldots,t^{m_i}_i,b^1_i,\ldots,b^{m_i}_i) = \{(t^j,b^j) \in \Pa: s^j = i \}$ and $m_i = \sum^m_{j=1} \ind(s^j = i)$ is the number of jumps at site $i$. We let $\scrP^l_i$ denote the set of all length $l$ paths, and $\scrP_i = \cup^{\infty}_{l=0} \scrP^l_i$ be the set of all paths, at the $i$th site. Define $\Delta_i^t(j) := t_i^{j+1} - t_i^j$ with $\Delta_i^t(m_i) := T - t_i^{m_i}$. The joint density of a path at site $i$ that begins at $x_i$ and ends at $y_i$ is given by 
\begin{align}
\label{eqn:Path Density ISM}
 \tP_{(T,\Q_i)}(y_i,\Pa_i \mid x_i) = \left[\prod^{m_i}_{j=1} \gamma_i(b^j_i; b^{j-1}_i))
\me^{-\Delta_i^t(j-1) \gamma_i(\cdot; b^{j-1}_i)}
 \right] \me^{-\Delta_i^t(m_i)\gamma_i(\cdot; y_i)} \ind_{ \{b^{m_i}_i = y_i \} }(\Pa_i).
\end{align}
Computing the transition probability \eqref{eqn:Transition Probability} under the ISM is straightforward:
\begin{align}
\label{eqn: Norm Constant Mu}
p_{(T,\Q)}(\y \mid \x) := \int_{\scrP} \tP_{(T,\Q)}(\y , \Pa \mid \x) \nu(d\Pa) &= \prod^n_{i=1} 
 \int_{\scrP_i} \tP_{(T,\Q_i)}(y_i,\Pa_i \mid x_i) \nu_i(d\Pa_i) \\
&= \prod^n_{i=1} (e^{T\Q_i})_{x_i,y_i},
\end{align}
where $\nu_i(d\Pa_i) := \sum^{\infty}_{l=0} \ind_{dm_i}(l) \nu_i^l(d\bt^l_i, d\bb^l_i) $ for $\nu^l_i = \nu^l_t \otimes \nu^l_b$.

Similarly, letting $\mu_i(\Pa_i \mid x_i, y_i) \propto \tP_{(T,\Q_i)}(y_i,\Pa_i \mid x_i)$ denote the endpoint-conditioned measure for site $i$, we have $\mu =  \mu_1 \times \ldots \times \mu_n$ under the ISM.
As noted previously, sampling paths from the endpoint-conditioned measure $\mu$ under the ISM is straightforward, as the path at each site can be drawn independently and exactly using established algorithms \cite{Hobolth:2009}.

To construct an MCMC algorithm to sample from $\pi(\Pa \mid \x, \y)$,
we define a $\pi$-invariant mutation kernel $\tK$  that randomly selects a \textit{block} of mutated sites and proposes a joint update to the paths at all sites in the block. Let $\mathscr{I} = \{\I_1,\ldots,\I_B\}$ denote a partition of all site indices $\{1,\ldots,n\}$ into $B$ \textit{blocks} and let $\Pa_{\I_j} = \{\Pa_i : i \in \I_j\}$ denote the projection of the path $\Pa$ onto the index set $\I_j$. Let $\tK_{(j)}$ be a Metropolis-Hastings kernel defined on $\mathscr{P}_{\I_{j}}$ that updates $\Pa_{\I_j}$ jointly by proposing from the ISM $\mu$:
\begin{align*}
\tK_{(j)}(\Pa_{\I_j},d\Pa^{\prime}_{\I_j}) &:=  \mu_{\I_j}(d\Pa_{\I_j}^{\prime}) \alpha_j(\Pa_{\I_j}, \Pa_{\I_j}^{\prime}) \\
&+ \delta_{\Pa_{\I_j}}(d\Pa^{\prime}_{\I_j}) \left[1- \int_{\scrP_{\I_j}}  \mu_{\I_j}(d\Pa_{\I_j}^{\prime}) \alpha_j(\Pa_{\I_j}, \Pa_{\I_j}^{\prime}) \right],
\end{align*}
where $\mu_{\I_j}(\Pa_{\I_j}) \propto \tP_{(T, \Q)}(\y_{\I_j} ,\Pa_{\I_j} \mid \x_{\I_j} )$ denotes the restriction of the ISM to the sites $\I_j$, and accepting or rejecting according to
\begin{align}
\alpha_j(\Pa_{\I_j},\Pa_{\I_j}^{\prime}) &:= 
\min\left\{1,\frac{w(\Pa_{\I_j}^{\prime})}{w(\Pa_{\I_j})}\right\},
\label{Eqn:MetropAcceptProb}
\end{align}
thus leaving the conditional distribution $\pi(\Pa_{\I_j} \mid \Pa_{\I_{[-j]}},\x,\y)$
invariant.
The mutation kernel $\tK$ chooses a partition element $\I_j$ uniformly at random and updates $\Pa_{\I_j}$ via $\tK_{(j)}$, yielding joint kernel
\begin{align}
\label{eqn:MWG Main Body}
\tK(\Pa, d\Pa^{\prime}) := \frac{1}{B}\sum^B_{j=1} \tK_{(j)}(\Pa_{\I_j}, d\Pa_{\I_j}^{\prime}) \delta_{\Pa_{\I_{[-j]}}}(d\Pa^{\prime}_{\I_{[-j]}}).
\end{align}
Note that $\tK$ is implicitly a function of the partition $\mathscr{I}$ but we do not make this dependence explicit in the notation; the choice of partition is discussed later in Section~\ref{sec: SMC Background}.

We will need the notion of a warm-start mixing time for a Markov chain. A distribution $\eta$ is said to be \textit{$\omega$-warm with respect to $\pi$} \cite{Vempala:2005} if 
\begin{align}
\label{eqn:resampling warm}
\sup_{B: \pi(B) > 0 } \frac{\eta(B)}{\pi(B)} \leq \omega.
\end{align}
Let $\mathscr{M}_{\omega }(\pi)$ denote the set of all $\omega$-warm distributions with respect to $\pi$ and define the \textit{warm mixing time} of a Markov kernel $\tK$ by
\begin{align}\label{eqn: warm mixing time definition}
\tau(\epsilon,\omega) := \inf\left\{s: \sup_{\eta \in \mathscr{M}_{\omega}(\pi)} \norm{ \eta \tK^s(\cdot) - \pi(\cdot)}_{\text{\tiny{TV}}} \leq \epsilon  \right\},
\end{align}
where $\eta \tK^s(\cdot) := \int_{  \scrP} \eta(d\Pa)\tK^s(\Pa, \cdot)$ and $\norm{\cdot}_{\text{\tiny{TV}}}$ denotes total variation distance.

\subsection{An SMC Algorithm for DSMs}
\label{sec: SMC Background}
Here, we introduce an alternative to the importance sampling scheme of \citet{Mathews3:2025}, based on sequential Monte Carlo \cite{DelMoral:2006,Chopin:2002}. This SMC scheme replaces the single-stage importance sampling of \citet{Mathews3:2025} with 
a multi-stage procedure which more finely controls the variance. Sequential Monte Carlo (SMC) introduces a set of intermediate ``bridging'' distributions, along with resampling, to form a telescoping product estimator for the marginal likelihood $p_{(T,\tQ)}(\y \mid \x)$. This reduces the $L^2$ distance required by any individual IS estimation step, where
\begin{align*}
 L^2(\pi, \mu) := \int_{\scrP}\left(\frac{\pi(\Pa \mid \x, \y)}{\mu(\Pa \mid \x, \y)}  \right)^2 \mu(d\Pa \mid \x, \y)
\end{align*}
is the squared $L^2(\mu)$ norm of $\pi/\mu$.
 The SMC algorithm introduced here sequentially samples from a sequence of distributions $\mu=\pi_0,\pi_1,\pi_2,\ldots,\pi_V= \pi$ all defined on $\scrP$. Let $q_v$ denote the corresponding unnormalized densities for $v\in \{0,\ldots,V\}$, so
\begin{align}
\pi_v(\Pa) = q_v(\Pa) / z_v,
\label{Eqn:TemperedPi}
\end{align}
with $z_v = \int_{\scrP} q_v(\Pa) d\nu(\Pa)$ the normalizing constant of $\pi_v$. The algorithm proceeds as follows: 
\begin{algorithm}[ht]
\caption{Sequential Monte Carlo (SMC) Sampler for DSMs}
\label{alg:smc}
\begin{algorithmic}[1]
  \State \textbf{Initialization:} Sample $\hat{\Pa}_0^{(1)},\ldots,\hat{\Pa}_0^{(N)} \iidsim \pi_0$.
  \For{$v = 1, \ldots, V$}
    \State \textbf{Resampling:} For $i = 1,\ldots,N$, sample $\tilde{\Pa}_v^{(i)} = \hat{\Pa}_{v-1}^{(i)}$ with probability
      \Statex \quad $\displaystyle \frac{w_v(\hat{\Pa}_{v-1}^{(i)})}{\sum_{j=1}^N w_v(\hat{\Pa}_{v-1}^{(j)})}$, where $w_v(\Pa) = q_v(\Pa)/q_{v-1}(\Pa)$.
    \State \textbf{Mutation:} For $i = 1,\ldots,N$, sample
      $\hat{\Pa}^{(i)}_v \mid \tilde{\Pa}^{(i)}_v \sim \tK^s_v(\tilde{\Pa}^{(i)}_v, \cdot)$,
      where $\tK_v$ is an ergodic $\pi_v$-invariant Markov kernel.
  \EndFor
\end{algorithmic}
\end{algorithm}

%
We define the run time of the SMC sampler as $NVs$, which is the total number of Markov transition steps required in a single run of the algorithm. Here, we choose $\pi_{0},\ldots,\pi_{V}$ to be a sequence of `tempered' DSMs with decreasing interaction strength as follows. Let $0 = \beta_{0} < \beta_1 < \ldots < \beta_{V} = 1$ denote a set of inverse temperatures and define
\begin{align}
\label{eqn:CD Rates SMC}
\tg_{i,v}(b; \tx_i) = \gamma_i(b; x_i)\phi^{\beta_{v}}(b; \tx_i) \quad \text{ for } \; b \in \scrA \setminus x_i.
\end{align}
The (unnormalized) conditional joint density \eqref{eqn:Path Density} of a path $\Pa$ from $\x$ to  $\y$ under the tempered model becomes
\begin{align}
\label{eqn:Path Density Tempered}
\tP_{(T,\tQ_{v})}(\y,\Pa \mid \x) := \left[\prod^{m(\Pa)}_{j=1} \tg_{s^j,v}(b^j; \tx^{j-1}_{s^j})
\me^{-\Delta^t(j-1) \tg_{v}(\cdot; \x^{j-1})}
\right] \me^{-\Delta^t(m)\tg_{v}(\cdot; \y) } \ind_{\x(m(\Pa)) = \y}(\Pa),
\end{align}
where $\tQ_v$ is the $a^n \times a^n$ rate matrix obtained from \eqref{eqn:CD Rates SMC}.  This defines a  sequence of endpoint-conditioned path distributions $\mu =\pi_0,\pi_1,\ldots,\pi_V = \pi$ with
\begin{align}
\label{eqn:Tempered sequence}
\pi_v(\Pa) := \pi_v(\Pa \mid \x, \y) = \frac{\tP_{(T,\tQ_v)}(\y,\Pa \mid \x)}{\int_{\scrP} \tP_{(T,\tQ_v)}(\y,\Pa \mid \x) \nu(d\Pa)} = \frac{\tP_{(T,\tQ_v)}(\y,\Pa \mid \x)}{z_v}.
\end{align}
Here $z_v = p_{(T,\tQ_v)}(\y \mid \x)$ denotes the transition probability under the DSM with rates \eqref{eqn:CD Rates SMC}. Define the product estimator of $z_v$ by the recursion
\begin{align*}
\hat{z}_v(\hat{\Pa}^{(1:N)}_{1:v}) = \hat{z}_v := z_0\prod^v_{v'=1} 
\hat{z}_{v'}
\end{align*}
and so an estimate of $p_{(T,\tQ)}(\y \mid \x)$ is obtained by
\begin{align}
\label{eqn: SMC estimator}
\hat{z}_V(\hat{\Pa}^{(1:N)}_{1:V}) = \hat{z}_V := z_0\prod^V_{v=1} \left(\frac{1}{N}\sum^N_{i=1} w_v(\hat{\Pa}^{(i)}_{v-1}) \right).
\end{align}
We let $\tK_v$ denote the $\pi_v$-invariant blocked component-wise Metropolis chain \eqref{eqn:MWG Main Body} defined in Section~\ref{sec: MCMC definition} that randomly selects a block of mutated sites and proposes a joint update to the paths at all sites in the block. Similarly, let $\tau_v(\epsilon,\omega)$ be the $\omega$-warm mixing time for the kernel $\tK_v$. Denote by $\tK_1,\ldots,\tK_V$ the mutation kernels targeting $\pi_1,\ldots,\pi_V$, respectively.

\section{Main Results}
\label{sec:MainResults}
We now state the main results of the paper, which concern the warm-start convergence rate of the MCMC algorithm defined in Section~\ref{sec: MCMC definition} and the sample-size requirements for approximating the marginal likelihood $p_{(T,\tQ)}(\y \mid \x)$ using the SMC algorithm introduced in Section~\ref{sec: SMC Background} (Algorithm~\ref{alg:smc}) using the MCMC algorithm as a mutation kernel. Supporting results are established in the following sections. Of primary interest is the scaling of these quantities with the size of the input problem; here measured by the length $n$ of the input sequences.
%
%
%
As we will see, the number of observed mutations $r =\text{d}_{\text{H}}(\x,\y)$ and the time interval $T$ also play important roles; thus we must consider the relative growth of $r(n)$ and $T(n)$ as $n$ increases. Luckily, there is a natural interval of interest for $T$ determined by $n$ and $r$, centered at $r/n$ \cite{Mathews2:2025}. Hence we will adopt the following assumption, the justification for which is discussed immediately after:
\begin{restatable}{assump}{rTAssumption}
\label{assump: T assumption}
The time interval $T(n) = \mathcal{O}(\frac{r(n)}{n})$ and the observed mutation count $r(n) = \mathcal{O}(n^{\frac{1}{2}})$.
\end{restatable}
%
%
%
In what follows, we often write $r$ and $T$ instead of $r(n)$ and $T(n)$ for brevity, except where we wish to emphasize the dependence explicitly.

The efficiency of the MCMC and SMC algorithms depends critically on the assumption that $T = \bigO(r/n)$, i.e. that $T$ not be too far from $r/n$. 
Because $T$ and mutation rates are not simultaneously identifiable, rate matrices are commonly scaled to one expected substitution per site per unit time, making $r/n$ -- a well known measure of genetic distance often called the \textit{p-distance} -- a natural estimate of $T$. However, under DSMs, estimators of $T$ such as the maximum likelihood estimate (MLE) or posterior mean are not available in closed form and require iterative optimization or MCMC sampling, with the marginal likelihood evaluated at each iteration.  \citet{Mathews2:2025} show that the likelihood decays exponentially for values of $T$ far from $r/n$, and therefore the posterior distribution of $T$ concentrates close to $r/n$, under any 
reasonable prior distribution, so that larger values of $T$ can be safely omitted from consideration in such algorithms without compromising their accuracy.

Our first main result is a bound on the $\omega$-warm mixing time of the blocked component-wise Metropolis chain introduced in Section~\ref{sec: MCMC definition}. This bound depends on the size of the largest subset of observed mutations having overlapping contexts. Formally, let
\begin{align}\label{eqn: r star def}
   r_\star \defeq \max \{m: \exists 
   i_1,i_2,\ldots,i_m \in \Smut \text{ with } \C_{i_j} \cap \C_{i_{j'}} \neq \emptyset  \text{ for some $j' < j$}   \},
\end{align}
%
so $r_\star$ is the largest component (connected subgraph) in the connectivity graph of $\Smut$.
We have the following result for $k/2$-nearest-neighbor models such as the CpG model \eqref{eqn:CpG Island Rates} and the S5F model of somatic hypermutation \cite{Yaari:2013}. 
\begin{theorem}
\label{thm: mixing time bound}
Suppose the rate function $\tilde{\gamma}_i$ at each site $i$ depends only on its immediate $k/2$ neighbors to the left and right, and Assumption~\ref{assump: T assumption} holds. Then the blocked component-wise Metropolis-Hastings chain with blocks given by the connected components of $\Smut$ has $\omega$-warm mixing time upper bounded by 
\begin{align*}
\tau\left(\epsilon ,\omega \right) = \bigO\big(\exp(c \cdot r_\star)\big),
\end{align*}
for error tolerance $\epsilon \in (0,1)$ and model-dependent constant $c \in (0,\infty)$, whenever $\frac{r_\star(n)}{\log(r(n))} \rightarrow \infty$.
\end{theorem}
Theorem~\ref{thm: mixing time bound} shows that the computational complexity of approximately sampling from $\pi(\Pa \mid \x,\y)$ under the block chain, when initialized according to a warm start, grows 
%
exponentially only in the size of the largest \textit{contiguous block} of observed mutation contexts, denoted $r_\star$. In practice $r_\star$ is often significantly smaller than the total number of mutations $r$ on which the importance sampler of \citet{Mathews3:2025} depends exponentially. The $r_\star/\log(r(n))$ condition ensures that the $\exp(r_\star)$ term is the dominant one; in cases where $r_\star$ grows slower than $\log(r)$, the complexity of the block chain is instead polynomial in $r$. (See Lemma~\ref{lemma:Mixing Time} for the general bound.) Although direct application of Theorem~\ref{thm: mixing time bound} is typically impractical due to the warm-start requirement, as we will see, this mixing time bound enables a similar complexity bound for the SMC algorithm (Algorithm~\ref{alg:smc}) that holds in practical settings, leading to a significant reduction in computational complexity for approximating $p_{(T,\tQ)}(\y \mid \x)$ using SMC compared to the importance sampling approach of \citet{Mathews3:2025}: 
%
\begin{theorem}
\label{thm:SMC bound}
Suppose each site $i$ depends on its immediate $k/2$ neighbors to the left and right, and Assumption~\ref{assump: T assumption} holds. 
Then Algorithm~\ref{alg:smc}
approximates the marginal sequence likelihood $p_{(T,\tQ)}(\y \mid \x)$ with $\epsilon$-relative error in time
\begin{align}
\label{eqn:SMC main thm body}
\bigO\big(  \exp(c \cdot r_\star) \big),
\end{align}
where $c \in (0,\infty)$ is a model-dependent constant, whenever $\frac{r_\star(n)}{\log(r(n))} \rightarrow \infty$.
%
\end{theorem}
Theorem~\ref{thm:SMC bound} says that, under the conditions of Assumption~\ref{assump: T assumption}, the SMC algorithm (Algorithm~\ref{alg:smc}) provides a significant improvement in computational complexity compared to the importance sampler studied in \cite{Mathews3:2025}. (Again, if $r_\star$ grows slower than $\log(r)$, the complexity of the SMC algorithm is instead polynomial in $r$; see Theorem~\ref{thm: SMC Complexity Bound} for the general bound.) Indeed, we will see that SMC provides a fully polynomial randomized approximation scheme (FPRAS) for the worst-case CpG model problem used to establish the exponential \textit{lower} bound for the importance sampler in Theorem~2 of \cite{Mathews3:2025}, thus providing an \textit{exponential} speed-up in runtime. 

Theorem~\ref{thm:SMC bound} is established by combining Theorem~\ref{thm: mixing time bound} with results from \cite{Marion:2018b,Marion:2023} and a uniform bound on the $\chi^2$-divergence between successive distributions $\pi_v$ and $\pi_{v-1}$ in the SMC algorithm (see Section~\ref{sec: SMC} for details). As a result (the proof of) Theorem~\ref{thm:SMC bound} immediately implies the following additional results: 

\begin{corollary}\label{cor: SMC samples}
    Under the conditions of Theorem~\ref{thm:SMC bound}, the SMC algorithm produces particles $\hat{P}^{(1)}_V,\ldots,\hat{P}^{(N)}_V$ in $\bigO(\exp(c \cdot r_\star))$ time satisfying 
   \begin{align*}
       \norm{\mathcal{L}(\hat{P}^{(1)}_V,\ldots,\hat{P}^{(N)}_{V}) - \pi^{\otimes N}}_{\text{TV}} \leq \epsilon,
   \end{align*}
  for error tolerance $\epsilon \in (0,1)$ and model-dependent constant $c \in (0,\infty)$, where $\pi^{\otimes N} = \pi \otimes \ldots \otimes \pi$ denotes the $N$-fold product measure of $\pi$.
\end{corollary}
Corollary~\ref{cor: SMC samples} states that SMC can be used to sample endpoint conditioned paths from $\pi$ under the same settings as Theorem~\ref{thm:SMC bound}; see Lemma~6 in \cite{Marion:2023}. Thus SMC provides a practical algorithm for obtaining samples approximately $i.i.d.$ from the target distribution $\pi$, overcoming the warm-start limitation of using Theorem~\ref{thm: mixing time bound} for MCMC targeting $\pi_V$ directly, by sequentially sampling from the sequence of distributions $\pi_0,\ldots,\pi_V$.

Another direct consequence is that SMC provides a randomized approximation scheme for expectations of bounded functions under $\pi$ (see Theorem~1 in \cite{Marion:2018b}):
\begin{corollary}
    Under the conditions of Theorem~\ref{thm:SMC bound}, let $f: \mathscr{P} \rightarrow \mathbb{R}$ be a bounded measurable function such that $\sup_{\Pa \in \mathscr{P}} |f(\Pa)| \leq 1$ and define the SMC estimator $\hat{f}$ of $\E_\pi[f]$ by
    \begin{align*}
        \hat{f}(\hat{P}^{1:N}_{V}) = \hat{f} := \frac{1}{N}\sum^{N}_{i=1} f(\hat{P}^{i}_{V}).
    \end{align*}
    Then in $\bigO(\exp(c \cdot r_\star))$ time the SMC estimator $\hat{f}$ satisfies
    \begin{align*}
        |\hat{f} - \E_\pi[f]| \leq \epsilon,
    \end{align*}
    with probability at least $3/4$, for error tolerance $\epsilon \in (0,1)$ and model-dependent constant $c \in (0,\infty)$.
\label{Cor:RAS}
\end{corollary}


Finally, we note that Theorems~\ref{thm: mixing time bound} and~\ref{thm:SMC bound} apply to DSMs where the context $\C_i$ of each site $i$ is restricted to the $k$ contiguous nearest neighbors (Assumption~\ref{assumption:neighborhood context} in Section~\ref{sec:island partitions} below).  DSMs with long-range interactions (e.g. Example 3 in Section~\ref{sec:Models of DNA Evolution}) can violate this condition.  The following result generalizes Theorems~\ref{thm: mixing time bound} and~\ref{thm:SMC bound} (and so Corollaries~\ref{cor: SMC samples} and~\ref{Cor:RAS}) to such ``non-neighboring'' contexts:
\begin{theorem}
\label{thm:Nonlocal}
Let $\mathscr{I}$ be a partition of $\{1,\ldots,n\}$ and let
\begin{align*}
\I_{j,\text{e}} \defeq \{i \in \I_j :  \C_i \cap \I_j^c \neq \emptyset\}
\end{align*}
%
denote the set of \emph{edge sites} in $\I_j$. If $\mathscr{I}$ satisfies
%
$\max_j|\I_{j,e}| = \bigO(r)$, 
$x_i = y_i$ for all $i \in \cup_{j}\I_{j,e}$ (no observed mutations at edge sites), and $\max r_{\I_j} = r_\star$,
%
then whenever $\frac{r_\star(n)}{\log(r(n))} \rightarrow \infty$ we have under Assumption~\ref{assump: T assumption}
\begin{enumerate}
\item the $\omega$-warm mixing time of the component-wise Metropolis-Hastings chain is upper bounded by
\begin{align*}
\tau\left(\epsilon ,\omega \right) = \bigO\big(\exp(c \cdot r_\star )\big),
\end{align*}
for error tolerance $\epsilon \in (0,1)$ and model-dependent constant $c \in (0,\infty)$, and
\item Algorithm~\ref{alg:smc} approximates the marginal sequence likelihood $p_{(T,\tQ)}(\y \mid \x)$ with $\epsilon$-relative error in time
\begin{align}
\bigO(\exp(c' \cdot r_\star )),
\end{align}
with $c' \in (0,\infty)$ a model-dependent constant.
\end{enumerate}
\end{theorem}
Theorem~\ref{thm:Nonlocal} thus provides conditions under which the marginal likelihood $p_{(T,\tQ)}(\y \mid \x)$ \eqref{eqn:Transition Probability} can be efficiently approximated even in DSMs exhibiting long-range dependencies.

\section{Warm-Start Mixing Time Bound}
\label{sec: mixing time bound}
In this section we provide the proof of Theorem~\ref{thm: mixing time bound}; supporting results are given in Appendix~\ref{Appdx:SMC Results}. We first establish a general upper bound on the mixing time of $\tK$ that holds for any $r(n)$ and $T$; application of Assumption~\ref{assump: T assumption} then allows a simplification leading to the $\bigO(\exp(c \cdot r_\star))$ complexity in Theorem~\ref{thm: mixing time bound}. We first introduce notation and an overview of the proof, before stating key lemmas.

\subsection{Background and Notation}
Define the inner product $\langle f,g \rangle_{\pi} := \int f(x)g(x)\pi(dx)$. Our approach to bounding $\tau(\epsilon,\omega)$ (defined in \eqref{eqn: warm mixing time definition}) will frequently involve the \textit{spectral gap} of a Markov kernel $\tK$: 
\begin{align}
\label{eqn:SpecGap definition}
\Gap(\tK) := \inf_{\substack{f \in L^2(\pi)\\ \V_{\pi}[f] \neq 0}} \frac{\calE_{\tK}(f,f)}{\V_{\pi}[f]},
\end{align}
where $\calE_{\tK}(f,f)=\langle f,(I-\tK)f\rangle_{\pi}$ is the Dirichlet form and $\V_{\pi}[f] = \langle f,f \rangle_{\pi} - \langle f,1 \rangle_{\pi} $. The spectral gap characterizes the rate at which $\tK$ converges to $\pi$. For example, a standard argument (see e.g. \cite{Marion:2023}) gives
\begin{align}
\label{eqn:Warm vs Spec}
\tau(\epsilon,\omega) \leq\frac{\log(2\epsilon^{-1}) + \log(\omega-1)}{\Gap(\tK)}.
\end{align}
%
%
Moreover, motivated by the notion of the $s$-conductance \cite{Lovasz:1999}, \citet{Atchade:2021} 
showed that under a warm start it suffices to bound $\Gap(\tK_{\mid \scrP_0})$ for any $\scrP_0 \subset \scrP$ satisfying $\pi(\scrP_0) \geq 9/10$, where $\tK_{\mid \scrP_0}$ denotes the \textit{restriction of} $\tK$ to the subset $\scrP_0$:
\begin{align}
\label{eqn:K restriction}
\tK_{\mid \scrP_0}(\Pa, A) := \tK(\Pa,A) + 
\delta_A(\Pa)\tK(\Pa,\mathscr{P}^c_0) \quad \text{ for } A \subset \scrP_0 .
\end{align}
%
(Note that $\tK_{\mid \scrP_0}$ is reversible with respect to $\pi_{\mid \scrP_0} = \pi \cdot \ind_{\scrP_0}/\pi(\scrP_0)$ when $\tK$ is reversible with respect to $\pi$.) For technical reasons, we will assume that $\tK$ is \textit{lazy} and remains in its current state with probability $1/2$ so that $\tK(\Pa, \Pa) \geq 1/2$; if necessary this can be ensured by taking $\tK'=\frac{1}{2}(I + \tK)$. A bound on $\Gap(\tK')$ implies a bound on $\Gap(\tK)$ as the two quantities differ only by a factor of $1/2$.
%
\begin{theorem}
\label{thm:Atchade Bound}(\citet{Atchade:2021})
Let $\epsilon \in (0,1)$ be fixed. Assume $\tK$ is reversible and lazy. If $\pi(\scrP^c_0) \leq \epsilon^2/ (20 \omega^2)$, then
\begin{align*}
\tau(\epsilon,\omega) \leq \frac{\log(2\epsilon^{-2}) + \log(\omega^2)}{\Gap(\tK_{\mid \scrP_0})}.
\end{align*}
%
\end{theorem}

We will refer to the block-update chain $\tK$ \eqref{eqn:MWG Main Body} defined in 
Section~\ref{sec: SMC Background} as a \textit{product chain} in the special case that
$\pi = \pi_1 \times \ldots \times \pi_B$ is a product distribution, with $\pi_j$ a distribution defined on $\scrP_{\I_j}$. Later we will use the following result regarding product chains:
\begin{theorem}(\citet{Diaconis:1996})
\label{thm: Product Chain SpecGap}
Let $\tK$ be a product chain. Then
\begin{align*}
\Gap(\tK) = \frac{1}{B} \min_j \Gap(\tK_{(j)}).
\end{align*}
\end{theorem}

\subsubsection{Neighborhood Models and Island Partitions}
\label{sec:island partitions} 
We will establish bounds for the blocked Metropolis algorithm in the case of DSMs where the context of a site is given by the $k/2$-\textit{nearest neighbors} to the left and to the right of each site. We refer to this as a $k$-\textit{neighborhood}. For example, the $2$-neighborhood of a site is its immediate left and right neighbors. The CpG model \eqref{eqn:CpG Island Rates} of genome sequence evolution \cite{Hwang:2004,Jensen:2000} is an example of a $2$-neighborhood DSM, while the S5F model \cite{Yaari:2013} of somatic hypermutation in B cell receptors \cite{Wiehe:2018,Mathews:2023b,Li:2025b} is a $4$-neighborhood DSM. Later these results will be extended to SMC algorithms using the blocked Metropolis chain as a mutation kernel (Algorithm~\ref{alg:smc}).
\begin{assump}
\label{assumption:neighborhood context}
The context at each site is contained in its $k$-neighborhood:
\begin{align*}
\tx_i = (x_{i-\frac{k}{2}},\ldots,x_{i-1},x_i,x_{i+1},\ldots,x_{i+\frac{k}{2}}) .
\end{align*}
\end{assump}
Recall that the blocked Metropolis algorithm \eqref{eqn:MWG Main Body} defined in Section~\ref{sec: MCMC definition} requires a partitioning of the sites $\{1,\ldots,n\}$ into $B$ blocks. Under Assumption~\ref{assumption:neighborhood context}, each $\tx_i$ forms a set of contiguous sites. Under this assumption, it will be natural to choose a partition $\mathscr{I}$ of $\{1,\ldots,n\}$ in which the sites of each block are also contiguous. When no observed mutations lie among the sites at the beginning and ending of each contiguous block, we refer to such a partition as an \textit{island partition}. More formally, recall $\C_i \subset \{1,\ldots,n\}$ denotes the set of sites lying in the context of site $i$ and define an \textit{edge site} of a partition element $\I_j$ to be a site in $\I_j$ whose context overlaps a neighboring block. Let 
$\D_j$ be the set of edge sites for block $j$:
\begin{align}
\label{eqn: Division Def}
\D_j = \{i \in \I_j: (\C_i \cap \I_{j-1}) \cup (\C_i \cap \I_{j+1}) \neq \emptyset\} \qquad \text{ for } j=1,\ldots,B,
\end{align}
and let $\mathscr{D} = \{\D_1,\ldots,\D_{B}\}$. We formalize this choice of partition in the following assumption which will be used in stating our main results:
\begin{assump}
\label{assumption:island partition}
$\mathscr{I}$ is an \emph{island partition}, i.e., the sites in each of the blocks $\I_1,\ldots,\I_B$ are contiguous and $x_i = y_i$ for all $i \in \mathscr{D}$.
\end{assump}
The following lemma  establishes a connection between $r_\star$ \eqref{eqn: r star def} and island partitions that will  be used below to establish Theorem~\ref{thm: mixing time bound}.
\begin{lemma}
\label{lemma: r star and island partition}
    For any subset of sites $\A \subset \{1,\ldots,n\}$, let $r_{\A} := \text{d}_{\text{H}}(\x_{\A}, \y_{\A})$ be the number of observed mutations between subsequences $\x_{\A}$ and $\y_{\A}$. Recall from \eqref{eqn: r star def} that 
    $r_\star$ is the size of the largest connected subgraph of the connectivity graph of $\Smut$. Under Assumption~\ref{assumption:neighborhood context}, there exists an island partition $\mathscr{I} = \{\I_1,\ldots,\I_B\}$ satisfying
    \begin{align*}
        \max_j r_{\I_j} \leq r_\star \quad \text{ and } \quad \max_j r_{\D_j} = 0 \quad \text{ for } \quad j = 1,\ldots,B.
    \end{align*}
\end{lemma}
\begin{proof}
    Consider the connectivity graph of $\Smut$, defined as the graph $\G$ with vertex set $\Smut$ and an edge between sites $i,j \in \Smut$ iff $\C_i \cap \C_j \neq \emptyset$. Let $\G_1,\ldots,\G_h$ for $1 \leq h \leq \abs{\Smut}$ denote the corresponding connected subgraphs of $\G$, and let $\mathcal{V}(\G_j)$ denote the vertices of $\G_j$. We construct an island partition $\mathscr{I}$ from $\G$ as follows. First, set
\begin{align*}
    \I_j = \cup_{i \in \mathcal V(\G_j)} \C_i \quad \text{ for } j = 1,\ldots,h.
\end{align*}
Under Assumption~\ref{assumption:neighborhood context}, $\C_i$ corresponds to the $k$-neighborhood around site $i$. Hence, the sets $\I_1,\ldots,\I_h$ are pairwise disjoint blocks of contiguous sites satisfying
\begin{align*}
    \max_j r_{\I_j} \leq r_\star \quad \text{ and } \quad \max_j r_{\D_j} = 0.
\end{align*}
The remaining sites $\{1,\ldots,n\} \setminus \cup^h_{j=1} \I_j$ may be partitioned arbitrarily into contiguous blocks containing no observed mutations; augmenting $\{\I_1,\ldots,\I_h\}$ with these blocks yields an island partition $\mathscr{I}$.
\end{proof}

Establishing our result involves bounding the number of jumps that occur at sites in $\calS^c$, and in particular the number of jumps $m(\Pa_{\D_j})$ among the division sites $\mathscr{D}$. Bounding $m(\Pa_{\D_j})$ will also require bounding the number of jumps at the boundaries of the division sites. Let $\partial \D_j = (\cup_{i \in \D_j} \C_i \setminus \D_j) \cap \I_j$ denote the division boundary for block $j$, and $\partial \mathscr{D} =  \cup_{j=1}^B \partial \D_j$ the set of all such boundary sites. An example is shown in Figure~\ref{fig: DandC}.
\begin{figure}
\centering
\begin{tikzpicture}[scale=0.5, baseline=(current bounding box.center)]
\node at (-1,0.5) {$\x =$};
\foreach \x/\y in {
  0/$\tA$,1/$\tA$,2/$\tG$,3/$\tG$,4/$\tA$,5/$\tA$,6/$\tA$,7/$\tA$,
  8/$\tA$,9/$\tA$,10/$\tA$,11/$\tG$,12/$\tG$,13/$\tA$,14/$\tA$,15/$\tA$,
  16/$\tA$,17/$\tA$,18/$\tA$,19/$\tA$,20/$\tG$,21/$\tG$,22/$\tA$,23/$\tA$}{
    \draw (\x,0) rectangle (\x+1,1);
    \node at (\x+0.5,0.5) {\y};
}
\draw[decorate,decoration={brace,amplitude=5pt}] (0,1) -- (6,1)  node[midway,yshift=0.4cm] {$\I_1$};
\draw[decorate,decoration={brace,amplitude=5pt}] (6,1) -- (18,1) node[midway,yshift=0.4cm] {$\I_2$};
\draw[decorate,decoration={brace,amplitude=5pt}] (18,1) -- (24,1) node[midway,yshift=0.4cm] {$\I_3$};
\draw[decorate,decoration={brace,amplitude=5pt,mirror}] (4,0)  -- (6,0)  node[midway,yshift=-0.4cm] {$\D_1$};
\draw[decorate,decoration={brace,amplitude=5pt,mirror}] (6,0)  -- (8,0)  node[midway,yshift=-0.4cm] {$\D_2$};
\draw[decorate,decoration={brace,amplitude=5pt,mirror}] (16,0) -- (18,0) node[midway,yshift=-0.4cm] {$\D_2$};
\draw[decorate,decoration={brace,amplitude=5pt,mirror}] (18,0) -- (20,0) node[midway,yshift=-0.4cm] {$\D_3$};
\foreach \x in {2.5,3.5}{ \draw[<-] (\x,-1) -- (\x,0); }
\node at (3.1,-1.6) {$\partial \D_1$};

\foreach \x in {8.5,9.5}{ \draw[<-] (\x,-1) -- (\x,0); }
\node at (9.1,-1.6) {$\partial \D_2$};

\foreach \x in {14.5,15.5}{ \draw[<-] (\x,-1) -- (\x,0); }
\node at (15.1,-1.6) {$\partial \D_2$};

\foreach \x in {20.5,21.5}{ \draw[<-] (\x,-1) -- (\x,0); }
\node at (21.1,-1.6) {$\partial \D_3$};

\def\ybot{-4.5}           
\def\ytop{-3.5}           
\def\ymid{-4.0}           
\def\arrowtail{-5.5}       
\def\arrowhead{-4.5}       
\def\labely{-6.1}          

\node at (-1,\ymid) {$\y =$};
\foreach \x/\y in {
  0/$\tA$,1/$\tA$,2/$\tT$,3/$\tT$,4/$\tA$,5/$\tA$,6/$\tA$,7/$\tA$,
  8/$\tA$,9/$\tA$,10/$\tA$,11/$\tT$,12/$\tT$,13/$\tA$,14/$\tA$,15/$\tA$,
  16/$\tA$,17/$\tA$,18/$\tA$,19/$\tA$,20/$\tT$,21/$\tT$,22/$\tA$,23/$\tA$}{
    \draw (\x,\ybot) rectangle (\x+1,\ytop);
    \node at (\x+0.5,\ymid) {\y};
}
\draw[decorate,decoration={brace,amplitude=5pt}] (0,\ytop)  -- (6,\ytop)  node[midway,yshift=0.4cm] {$\I_1$};
\draw[decorate,decoration={brace,amplitude=5pt}] (6,\ytop)  -- (18,\ytop) node[midway,yshift=0.4cm] {$\I_2$};
\draw[decorate,decoration={brace,amplitude=5pt}] (18,\ytop) -- (24,\ytop) node[midway,yshift=0.4cm] {$\I_3$};
\draw[decorate,decoration={brace,amplitude=5pt,mirror}] (4,\ybot)  -- (6,\ybot)  node[midway,yshift=-0.4cm] {$\D_1$};
\draw[decorate,decoration={brace,amplitude=5pt,mirror}] (6,\ybot)  -- (8,\ybot)  node[midway,yshift=-0.4cm] {$\D_2$};
\draw[decorate,decoration={brace,amplitude=5pt,mirror}] (16,\ybot) -- (18,\ybot) node[midway,yshift=-0.4cm] {$\D_2$};
\draw[decorate,decoration={brace,amplitude=5pt,mirror}] (18,\ybot) -- (20,\ybot) node[midway,yshift=-0.4cm] {$\D_3$};
\foreach \x in {2.5,3.5}{ \draw[<-] (\x,\arrowtail) -- (\x,\arrowhead); }
\node at (3.1,\labely) {$\partial \D_1$};

\foreach \x in {8.5,9.5}{ \draw[<-] (\x,\arrowtail) -- (\x,\arrowhead); }
\node at (9.1,\labely) {$\partial \D_2$};

\foreach \x in {14.5,15.5}{ \draw[<-] (\x,\arrowtail) -- (\x,\arrowhead); }
\node at (15.1,\labely) {$\partial \D_2$};

\foreach \x in {20.5,21.5}{ \draw[<-] (\x,\arrowtail) -- (\x,\arrowhead); }
\node at (21.1,\labely) {$\partial \D_3$};
\end{tikzpicture}
 
\caption{Example of an island partition with $B = 3$ blocks showing $\mathscr{D}$ and $\partial \mathscr{D}$ for $k=4$.}
\label{fig: DandC}
\end{figure}

\subsubsection{Proof Overview}
In light of Theorem~\ref{thm:Atchade Bound}, in order to upper bound the warm-start mixing time of $\tK$ it suffices to obtain a lower bound on $\Gap(\tK_{\mid \scrP_0})$ for some high probability set $\scrP_0 \subset \scrP.$
We choose $\scrP_0$ to be a subset of paths where the number of jumps at each site is uniformly bounded.  We can then apply the following lemma, which is straightforward to verify. Recall that $\mathscr{I} = \{\I_1,\ldots,\I_{B}\}$ denotes a partition of $\{1,\ldots,n\}$ into $B$ blocks. 
\begin{lemma}
\label{lemma:product restriction}
Let $\scrP_0 = \cap^B_{j=1} \scrP_{0,j}$ where $\scrP_{0,j} \subset \scrP_{\I_j}$.  Then
\begin{align*}
\tK_{\mid \scrP_0}(\Pa, d\Pa^{\prime}) = \frac{1}{B}\sum^B_{i=1} \tK_{(j) \mid  \scrP_{0,j}}(\Pa_{\I_j}, d\Pa_{\I_j}^{\prime}) \delta_{\Pa_{\I_{[-j]} }}(d\Pa^{\prime}_{\I_{[-j]} }).
\end{align*}
%
\end{lemma}
%
%
\noindent We will lower bound the spectral gap of $\tK_{\mid \scrP_0}(\Pa, d\Pa^{\prime}) $ by the spectral gap of
a carefully-chosen product chain which admits more direct analysis. This product chain is obtained by omitting all context-dependence \textit{between} partition blocks, by setting $\phi \equiv 1$ for rates defined at all edge sites $\scrD = \{\D_1,\ldots,\D_{B}\}$ (see \eqref{eqn: Division Def} and Figure~\ref{fig: DandC}). In particular, let 
\begin{equation}
\label{eqn: hat rates}
\hat{\gamma}_i(b; \tx_i)  =
\begin{cases}
\gamma_i(b; x_i) & \text{ for }  i \in \scrD \\
\tg_i(b; \tx_i) & \text{ for }  i \in \{1,\ldots,n\} \setminus \scrD.
\end{cases}
\end{equation}
Let $\hat{\Q}$ be the $a^n \times a^n$ rate matrix with rates \eqref{eqn: hat rates} and define $\hat{\pi}$ for the corresponding DSM by
\begin{align*}
\hat{\pi}(\Pa) := \hat{\pi}(\Pa \mid \x, \y) \propto \tP_{(T,\hat{\Q})}(\y, \Pa \mid \x).
\end{align*}
Note that sites in $\I_j \setminus \D_j$ evolve under the same rates in both $\pi$ and $\hat{\pi}$, but
$\Pa_{\I_1},\ldots,\Pa_{\I_{B}}$ are independent under $\hat{\pi}$ since $\hat{\gamma}_i = \gamma_i$ for $i \in \scrD$. Therefore $\hat{\pi}$ is a product measure
\begin{align}
\label{eqn: hat pi definition}
\hat{\pi}(\Pa) = \hat{\pi}_1(\Pa_{\I_1}) \times \ldots \times \hat{\pi}_B(\Pa_{\I_B}),
\end{align}
where $\hat{\pi}_j(\Pa_{\I_j}) :=   \hat{\pi}_j(\Pa_{\I_j} \mid \x_{\I_j}, \y_{\I_j})  \propto \tP_{(T,\hat{\Q})}(\y_{\I_j}, \Pa_{\I_j} \mid \x_{\I_j})$ is an endpoint conditioned model with rates \eqref{eqn: hat rates}. Consider the blockwise Metropolis-Hastings chain defined in Section~\ref{sec: SMC Background} with invariant distribution $\pi$. Let $\hat{\tK}$ be the $\hat{\pi}$-invariant modified chain which uses the same (blockwise) independent-site proposal distribution, but utilizes $\hat{\pi}$ in place of $\pi$ in the acceptance \eqref{Eqn:MetropAcceptProb}: 
\begin{align}
\hat{\tK}(\Pa, d\Pa^{\prime}) := \frac{1}{B}\sum^{B}_{i=1} \hat{\tK}_{(j)}(\Pa_{\I_j}, d\Pa_{\I_j}^{\prime}) \delta_{\Pa_{\I_{[-j]}}}(d\Pa^{\prime}_{\I_{[-j]}}).
\label{Eqn:KHatDef}
\end{align}
Then for any $\scrP_0 = \cap^{B}_{j=1} \scrP_{0,j}$ with $\scrP_{0,j} \subset \scrP_{\I_j}$,
$\hat{\tK}_{\mid \scrP_0}$ is a product chain by \eqref{eqn: hat pi definition} and Lemma~\ref{lemma:product restriction}, and hence Theorem~\ref{thm: Product Chain SpecGap} gives
\begin{align*}
\Gap(\hat{\tK}_{\mid \scrP_0}) = \frac{1}{B} \min_j \Gap(\hat{\tK}_{(j) \mid \scrP_{0,j}}).
\end{align*}
We will choose $\scrP_0 $ such that all of $\pi(\Pa)$, $\hat{\pi}(\Pa)$, and $\mu(\Pa)$ are uniformly bounded, enabling us to define a function $\delta(\x,\y,\tQ,\Q,T)$ such that, by a simple comparison argument \cite{Diaconis:1993}
%
\begin{align}
\label{eqn:Mixing Time Overview One}
\Gap(\tK_{\mid \scrP_0}) \geq \delta\,   \Gap(\hat{\tK}_{\mid \scrP_0}) =  \frac{\delta}{B}  \min_j \Gap(\hat{\tK}_{(j) \mid \scrP_{0,j}}).
\end{align}
Finally, since each $\hat{\tK}_{(j) \mid \scrP_{0,j}}$ is a \textit{Metropolized independence sampler}  with uniformly bounded proposal and target densities, a bound on $\min_j \Gap(\hat{\tK}_{(j) \mid \scrP_{0,j}})$ follows easily. 

For brevity, in what follows we adopt the following shorthand notation for the restricted kernel, its component kernels, and their corresponding stationary distributions:
%
\begin{align*}
\tK_0 :=  \tK_{\mid \scrP_0}\quad\qquad 
\tK_{0,j} :=  \tK_{(j) \mid \scrP_{0,j}}\quad\qquad 
\pi_0 :=  \pi_{\mid \scrP_0}\quad\qquad 
\pi_{0,j} :=  \pi_{j \mid \scrP_{0,j}}  
\end{align*}
and use analogous notation for the restricted kernel of the modified chain $\hat{\tK}$ and its corresponding component kernels and stationary distributions:
%

\begin{align*}
\hat{\tK}_0 := \hat{\tK}_{\mid \scrP_0}\quad\qquad 
\hat{\tK}_{0,j} := \hat{\tK}_{(j) \mid \scrP_{0,j}}\quad\qquad 
\hat{\pi}_0 := \hat{\pi}_{\mid \scrP_0}\quad\qquad 
\hat{\pi}_{0,j} := \hat{\pi}_{j \mid \scrP_{0,j}} .  
\end{align*}

%


\subsection{Main Results for Spectral Gaps}
The bounds below are stated in terms of the following quantities. For a subset of sites $\calA \subset \{1,\ldots,n\}$, recall $\x_{\calA}$ and $\y_{\calA}$ denote the corresponding subsequences and define
\begin{align}
\label{eqn: zeta definition}
n_{\calA} := \abs{\calA} \quad\quad r_{\calA} := \text{d}_{\text{H}}(\x_{\calA}, \y_{\calA}) \quad\quad \zeta_{\calA} := r_{\calA} + r_{\calA}T + (n_{\calA} - r_{\calA}) T^2.
\end{align}
%
In addition, let $\zeta \defeq \zeta_{\{1,\ldots,n\}}$ denote the special case that $\A = \{1,\ldots,n\}$.  Key to our analysis is the following bound on the MGF of $m(\Pa_{\A})$, the number of jumps in sequence path $\Pa$ which occur at sites in the subset $\A$, under the DSM $\pi$; the proof and explicit constants are deferred to Appendix~\ref{Appdx:SMC Results}.
\begin{lemma}
\label{lemma:DSM MGF Subset Main Body}
Let $\calA \subset \{1,\ldots,n\}$ be a set of contiguous site indices and $\theta \in (0,\infty)$. If $\pi$ is a \emph{$k$-neighborhood} DSM  (Assumption~\ref{assumption:neighborhood context}), then there exists a model-dependent constant $\lambda(\theta)$ such that
\begin{equation}
\label{Eqn:DSMMGFBound}
\E_{\pi}[\theta^{m(\Pa_{\calA})}] \leq  e^{\lambda(\theta)\zeta_{\calA}}.
\end{equation}
\end{lemma}
Lemma~\ref{lemma:DSM MGF Subset Main Body} provides a bound on the MGF $M(t)$ of $m(\Pa)$ (taking $\theta = e$). Applying the Chernoff bound $\Pr[m(\Pa_{\calA}) > M] 
\leq  \E[e^{m(\Pa_{\calA})}]e^{-M}
$ gives that the number $m(\Pa_{\calA})$ of jumps in $\calA$ is bounded above by $\lambda(e)\zeta_{\calA} = \bigO(\zeta_{\calA})$
with high probability, decaying as $e^{-(M-\lambda(e)\zeta_{\calA})}$, while applying Jensen's inequality $\E_{\pi}[e^{m(\Pa_{\calA})}] \geq e^{\E_{\pi}[m(\Pa_{\calA})]}$ gives a bound on the expected number of jumps in $\calA$:
\begin{align}
\E_{\pi}[m(\Pa_{\calA})] \leq \bigO(\zeta_{\calA}).
\end{align}
Therefore, the expected number of jumps in $\calA$ is bounded by $\bigO(r_{\calA} + r_{\calA}T + (n_{\calA} - r_{\calA}) T^2)$, which under Assumption~\ref{assump: T assumption} is  $\bigO(1)$ when $r_{\calA} = 0$.
 This fact will play an important role in bounding the number of extra jumps among the edge sites $\scrD$ of a partition defined in Section~\ref{sec: SMC}, recalling that $r_{\scrD} = 0$ for \textit{island partitions} (Assumption~\ref{assumption:island partition}).

We will use Lemma~\ref{lemma:DSM MGF Subset Main Body} to find a high probability subset of paths $\scrP_0 \subset \scrP$ on which $m(\Pa)$ is uniformly bounded. The set $\scrP_0$ will play an important role in bounding the warm mixing time of $\tK$ to establish Theorem~\ref{thm: mixing time bound}.   

\begin{lemma}
\label{lemma:DSM Conc Bounds}
Let $\epsilon \in (0,1)$ and for $\calA \subset \{1,\ldots,n\}$ define the event
\begin{align}
\scrP_0(\epsilon,\calA) = \{\Pa: m(\Pa_{\calA}) \leq M_{\epsilon}(\calA) \},
\label{Eqn:ScriptQDef}
\end{align}
where $M_{\epsilon}(\calA) = M(\zeta_{\calA},B,\epsilon) := \zeta_{\calA} \lambda(e) + \log(3B/\epsilon)$, and $\lambda(\cdot)$ and $\zeta_{\calA}$ are defined in Lemma~\ref{lemma:DSM MGF Subset} (section~\ref{Sec:MGFbound}). Let $\scrP_{0,j}(\epsilon) = \scrP_{0,j} := \scrP_0(\epsilon,\I_j)$, 
\begin{align}
\overline{\scrP}_{0,j}(\epsilon) := \scrP_0(\epsilon,\I_j) \cap \scrP_0(\epsilon,\D_j) \cap \scrP_0(\epsilon,\partial \D_j),
\label{Eqn:P0jDef}
\end{align}
and $\scrP_0(\epsilon) = \scrP_0 := \cap^B_{j=1} \overline{\scrP}_{0,j}(\epsilon)$.
Then $\pi(\scrP_0) \geq 1 - \epsilon$.
\end{lemma}
\begin{proof}
The result follows immediately by application of the Chernoff bound (using the MGF bound of Lemma~\ref{lemma:DSM MGF Subset}) to obtain tail inequalities for each set:
\begin{align*}
\Prob_{\pi}(m(\Pa_{\D_j}) > \lambda(e)\zeta_{\D_j} + \log(3B\epsilon^{-1})) \leq \E_{\pi}[e^{m(\Pa_{\D_j})}] e^{-M_{\epsilon}(\D_j)} \leq \frac{\epsilon}{3B}
\end{align*}
and taking a union bound over all $3B$ events.
\end{proof}
\noindent We will now establish the spectral gap bound for the restricted kernel given in \eqref{eqn:Mixing Time Overview One}. To do so, We will use the following bound on the density ratio with respect to the invariant distribution of the product chain $\hat{\tK}$ defined above:
\begin{lemma}
\label{lemma: hat weight uniform bound}
Define the importance weight
\begin{align}
\label{eqn:Hat Weight}
\hat{w}(\Pa) := \tP_{(T,\tQ)}(\y, \Pa \mid \x)/ \tP_{(T,\hat{\Q})}(\y, \Pa \mid \x).
\end{align}
Let $\xi$ be a probability measure defined on the set of paths $\scrP$. Then the following bounds hold with probability one under $\xi$ under Assumptions~\ref{assumption:neighborhood context} and \ref{assumption:island partition}
\begin{align}
\label{eqn:Hat Weight Bound}
\Prob_{\xi}\left(\phi^{m(\Pa_{\scrD})}_{\min} e^{-T (m(\Pa_{\scrD}) + m(\Pa_{\partial \scrD})) (\tilde{\delta} + \delta) + c } \leq \hat{w}(\Pa) \leq \phi_{\max}^{m(\Pa_{\scrD})} e^{T (m(\Pa_{\scrD}) + m(\Pa_{\partial \scrD})) (\tilde{\delta} + \delta) + c }  \right) = 1,
\end{align}
for constant $c = -T(\tg(\cdot; \y) - \hat{\gamma}(\cdot; \y))$. In particular, define the following constant which depends on the full set of division sites $\mathscr{D}$ and their boundaries
$\partial \mathscr{D}$:
\begin{align}
\theta_{\scrD} := \exp\left(\log(\phi_{\star})M_{\epsilon}(\scrD)  + 2T(M_{\epsilon}(\scrD) + M_{\epsilon}(\partial \scrD))(\tilde{\delta } + \delta)  \right),
\label{Eqn:ThetaDef}
\end{align}
where again $M_{\epsilon}(\calA) = \zeta_{\calA} \lambda(e) + \log(3B/\epsilon)$ and recall $\zeta_{\calA} := r_{\calA} + r_{\calA}T + (n_{\calA} - r_{\calA}) T^2$ for $\calA \subset \{1,\ldots,n\}$ \eqref{eqn: zeta definition}. Then the
following uniform bound for $\scrP_0$ given in Lemma~\ref{lemma:DSM Conc Bounds} holds:
\begin{align}
\label{eqn:density ratio bounds}
\Prob_{\xi}\left(\theta^{-1}_{\scrD} \leq \frac{\pi_0(\Pa)}{\hat{\pi}_0(\Pa)} \leq \theta_{\scrD} \right) = 1.
\end{align}
\end{lemma}
\begin{proof}
Recall that the substitution rates under $\hat{\pi}$ given in \eqref{eqn: hat rates} are identical to those under $\pi$ except at sites in $\scrD$, which follow the rates of the ISM $\mu$. From the path density \eqref{eqn:Path Density}, we see that
\begin{align}
\hat{w}(\Pa) = \left[\prod_{ \{l: s^l \in \scrD \}} \phi(b_{s^l } ; \tx^{l-1}_{s^l} )  \right] \me^{\tilde{\psi}_{\scrD}(\s(\Pa),\bb(\Pa)) + c},
\label{Eqn:what}
\end{align}
where the term in the exponent corresponding to the difference between the exit rates under $\pi$ and $\hat{\pi}$ for sites $i \in \scrD$ is denoted by $\tilde{\psi}_{\scrD}(\s(\Pa),\bb(\Pa)) := \sum^{m(\Pa)}_{l=1}t^l \Gamma^l_{\scrD}(\Pa)$ where
\begin{align*}
\Gamma^l_{\scrD}(\Pa) &:= \sum_{i \in \scrD} \left(\tg_i(\cdot; \tx^l_i) - \gamma_i(\cdot; x^l_i) - \tg_i(\cdot; \tx^{l-1}_i) + \gamma_i(\cdot; x^{l-1}_i)\right).
\end{align*}
For the bracketed term in \eqref{Eqn:what}, note that for any path $\Pa$ we have $|\{l: s^l \in \scrD \}| = \sum_{i \in \scrD} m(\Pa_i) = m(\Pa_{\scrD})$ and so
\begin{align*}
\Prob_{\xi}\left(\phi_{\min}^{m(\Pa_{\scrD})} \leq \prod_{ \{l: s^l \in \scrD \}} \phi(b_{s^l } ; \tx^{l-1}_{s^l} )  \leq \phi_{\max}^{m(\Pa_{\scrD})}\right)  = 1.
\end{align*}
Turning to the exponential term in \eqref{Eqn:what}, observe that $\Gamma^l_{\scrD}(\Pa)$ is non-zero only when $s^l \in \scrD \cup \partial \scrD$. When $s^l \in \scrD \cup \partial \scrD$, the context of at most $k+1$ sites change (the $k$ sites lying in the context of site $s^l$, and $s^l$ itself). An argument identical to the one given in the proof of Lemma~\ref{lemma: Uniform Bounds on Delta} yields:
\begin{align*}
\Prob_{\xi}\left(|\tilde{\psi}_{\scrD}(\s(\Pa),\bb(\Pa))| \leq T (m(\Pa_{\scrD}) + m(\Pa_{\partial \scrD}) ) (\tilde{\delta} + \delta)\right) =1.
\end{align*}
The first stated bound \eqref{eqn:Hat Weight Bound} follows. The second statement \eqref{eqn:density ratio bounds} follows from the first and the uniform bounds on $m(\Pa_{\scrD})$ and $m(\Pa_{\partial \scrD})$ used in proving Lemma~\ref{lemma:DSM Conc Bounds}, which hold for $\Pa \in \scrP_0$.
\end{proof}
With the bounds on $\pi_0/\hat{\pi}_0$ from Lemma~\ref{lemma: hat weight uniform bound} in place, we are now in a position to lower bound the spectral gap of $\tK_0$ by the spectral gap of $\hat{\tK}_0$. This is done in Lemma~\ref{lemma:Spec Gaps 1} below. Later we will obtain an explicit lower bound on the spectral gap of $\hat{\tK}_0$ itself (Lemma~\ref{lemma:Spec Gaps 2}).
First, we state the following result which will be used in Lemma~\ref{lemma:Spec Gaps 1}.
%
%
%
\begin{lemma}
\label{lemma:spec gap lemma}
Recall that the spectral gap of a $\pi$-invariant Markov kernel $\tK$ is defined by
\begin{align}
\label{eqn:spec gap def recall}
\Gap(\tK) := 
\inf_{\substack{f \in L^2(\pi)\\ \V_{\pi}(f) \neq 0}} \frac{\calE_{\tK}(f,f)}{\V_{\pi}(f)}
= 
\inf_{\substack{f \in L^2(\pi)\\ \V_{\pi}[f] \neq 0}} \frac{\int \int \pi(dx)\tK(x,dy)(f(x)-f(y))^2}{\int \int \pi(dx)\pi(dy)(f(x)-f(y))^2}.
\end{align}
Let $\hat{K}$ be a $\hat{\pi}$-invariant Markov kernel, with $\hat{\pi}(x)$ and $\pi(x)$  densities 
defined on a common state space $\X$ and with respect to a common dominating measure $\rho$, i.e., $\hat{\pi}(A) = \int_A\hat{\pi}(x)\rho(dx)$ and $\pi(A) = \int_A\pi(x) \rho(dx)$ for $A \subset \X$. Suppose further that
\begin{enumerate}
\item \label{eqn:spec gap lemma 1} The ratio of each density with respect to $\rho$ is bounded: $a_0 \leq \pi(x)/\hat{\pi}(x) \leq a_1$ for all $x \in \X$ for some $a_0,a_1 \in \mathbb{R}^+$ with $0 < a_0 \leq a_1$.
\item \label{eqn:spec gap lemma 2} There exists $a_2 > 0$ such that for all $f \in L^2(\pi)$
\begin{align*}
\int_A\tK(x,dy)(f(x)-f(y))^2 \; \geq \; a_2\int_A \hat{\tK}(x,dy)(f(x)-f(y))^2 \quad \text{ for all } x \in \X \text{ and } A \subset \X.
\end{align*}
%
\end{enumerate}
Then $\Gap(\tK) \geq \frac{a_0a_2}{a_1} \Gap(\hat{\tK})$.
\end{lemma}
\begin{proof}
Note that $\V_{\pi}[f] \neq 0 \iff \V_{\hat{\pi}}[f] \neq 0$ and $f \in L^2(\pi) \iff f \in L^2(\hat{\pi})$ since $a_0 < \pi(x)/\hat{\pi}(x) < a_1$ uniformly for all $x \in \X$.  By assumption $\V_{\pi}[f] \leq a_1 \V_{\hat{\pi}}[f]$ and
\begin{align*}
\calE_{\tK}(f,f) \geq a_0a_2 \calE_{\hat{\tK}}(f,f).
\end{align*}
Taking the infimum over all non-constant $f \in L^2(\pi)$ yields the stated bound.
\end{proof}
We now lower bound $\Gap(\tK_0)$ by  $\min_j  \Gap(\hat{\tK}_{0,j})$; later we will obtain a lower bound on $\min_j  \Gap(\hat{\tK}_{0,j})$ as well. The constants in the following result involve the ratio $\phi_{\star} = \phi_{\max}/\phi_{\min}$ of the maximum and minimum context-dependent rates as well as $\lambda(e)$ with $\lambda(\cdot)$ defined in Lemma~\ref{lemma:DSM MGF Subset}. Recall from Lemma~\ref{lemma:DSM MGF Subset} that  $\lambda(e)$ is a constant which satisfies the following bound on the MGF of the number of jumps $m(\Pa)$ under $\pi$:
\begin{align*}
\E_{\pi}[e^{m(\Pa_{\calA})}] \leq e^{\lambda(e) \zeta_{\calA}  } \quad \text{ with } \zeta_{\calA} = r_{\calA} + r_{\calA}T + T^2(n_{\calA}-r_{\calA}) \text{ and } \calA \subset \{1,\ldots,n\}.
\end{align*}
Recall also that by definition, the number of total jumps $m(\Pa)$ for any $\Pa \in \scrP_0$ is no more than $\lambda(e) (\zeta_{\I_j} + \zeta_{\scrD} + \zeta_{\partial \scrD}) + 3\log(3B\epsilon^{-1})$, where $\I_j$, $\scrD$, and $\partial \scrD$ denote the $j$th block of sites, the set of division sites, and the boundary of the division sites, respectively, as defined in Section~\ref{sec:island partitions}.
This property of $\scrP_0$ will enable us to obtain uniform bounds on the density ratio $\pi(\Pa)/\hat{\pi}(\Pa)$, thus satisfying the conditions of Lemma~\ref{lemma:spec gap lemma}. We can then appeal to Theorem~\ref{thm: Product Chain SpecGap} to establish the following result.
\begin{restatable}{lemma}{SpecGapsOne}
\label{lemma:Spec Gaps 1}
Let $\theta_{\scrD}$ be defined as in \eqref{Eqn:ThetaDef}, and for each block $j$ define the block-specific constant
\begin{align}
\label{Eqn:ThetaDef D}
\theta_{\D_j} &:= \exp\left(\log(\phi_{\star})M_{\epsilon}(\D_j)  + 2T(M_{\epsilon}(\D_j) + M_{\epsilon}(\partial \D_j))(\tilde{\delta } + \delta)  \right).
\end{align}
The following lower bound holds: 
\begin{align*}
\Gap(\tK_0) \geq \frac{1}{B \theta^2_{\scrD} \max_j}\theta^2_{\D_j}\min_j \Gap(\hat{\tK}_{0,j}) .
\end{align*}
\end{restatable}
We will use the following lemma to establish Lemma~\ref{lemma:Spec Gaps 1}. For brevity we define the following shorthand notation for all $f \in L^2(\pi)$
\begin{align*}
(\nabla f(\Pa,\Pa^{\prime}))^2 := (f(\Pa) - f(\Pa^{\prime}))^2.
\end{align*}
\begin{lemma}
\label{lemma: Spec Gaps 1 intermediate}
The following bound holds for any $\Pa \in \scrP_0$ and $A \subset \scrP_0$:
\begin{align}
\label{eqn: gap K0 and hat K0 wts}
\int_A\tK_0(\Pa,d\Pa^{\prime})  (\nabla f(\Pa,\Pa^{\prime}))^2 \ind_{\scrP_0}(\Pa)  \geq  \frac{1}{\max_j \theta^2_{\D_j}} \int_A \hat{\tK}_0(\Pa,d\Pa^{\prime})  (\nabla f(\Pa,\Pa^{\prime}))^2 \ind_{\scrP_0}(\Pa).
\end{align}
\end{lemma}
\begin{proof}
Let $\Pa \in \scrP_0$ and $A \subset \scrP_0$. We have
\begin{align*}
\MoveEqLeft[7] \int_A \tK_0(\Pa,d\Pa^{\prime}) (\nabla f(\Pa,\Pa^{\prime}))^2 \ind_{\scrP_0}(\Pa)
= \int_A\tK(\Pa,d\Pa^{\prime})  (\nabla f(\Pa,\Pa^{\prime}))^2 \ind_{\scrP_0}(\Pa) \\
& = \int_A \Big(\frac{1}{B} \sum^B_{j=1} 
\mu_{\I_j}(d\Pa_{\I_j}^{\prime}) 
\alpha_j(\Pa_{\I_j}, \Pa_{\I_j}^{\prime}) \delta_{\Pa_{\I_{[-j]}}}(d\Pa^{\prime}_{\I_{[-j]}})\Big)(\nabla f(\Pa,\Pa^{\prime}))^2\ind_{\scrP_0}(\Pa) \\
& \geq \int_A \frac{1}{\max_j \theta^2_{\D_j}}\hat{\tK}(\Pa,d\Pa^{\prime})(\nabla f(\Pa,\Pa^{\prime}))^2\ind_{\scrP_0}(\Pa) \\
& = \int_A \frac{1}{\max_j \theta^2_{\D_j}}\hat{\tK}_0(\Pa,d\Pa^{\prime}) (\nabla f(\Pa,\Pa^{\prime}))^2 \ind_{\scrP_0}(\Pa) . 
\end{align*}
The first and final equalities hold because $\tK$ and $\hat{\tK}$ share the same proposal distribution. To see why the inequality holds, recall that the acceptance ratio for the $\hat{\tK}_{0,j}$ chain is given by
\begin{align}
\label{eqn:accept ratio alpha hat}
\hat{\alpha}_j(\Pa_{\I_j}, \Pa_{\I_j}^{\prime}) := \min\left\{1, \frac{\tP_{(T,\hat{\Q})}(\y_{\I_j}, \Pa^{\prime}_{\I_j} \mid \x_{\I_j}, \Pa^{\prime}_{\I_{[-j]}})}{\tP_{(T,\hat{\Q})}(\y_{\I_j}, \Pa_{\I_j} \mid \x_{\I_j}, \Pa_{\I_{[-j]}}) }  \frac{\tP_{(T,\Q)}(\y_{\I_j}, \Pa_{\I_j} \mid \x_{\I_j})}{\tP_{(T,\Q)}(\y_{\I_j}, \Pa_{\I_j}^{\prime} \mid \x_{\I_j})} \right\},
\end{align}
and by Lemma~\ref{lemma: hat weight uniform bound} and \eqref{eqn: hat pi definition} we have
\begin{align*}
\alpha_j(\Pa_{\I_j}, \Pa_{\I_j}^{\prime}) \ind_{\scrP_0}(\Pa) \ind_{\scrP_0}(\Pa^{\prime}) &\geq  \frac{1}{\theta^2_{\D_j}}\hat{\alpha}_j(\Pa_{\I_j}, \Pa_{\I_j}^{\prime}) \ind_{\scrP_0}(\Pa) \ind_{\scrP_0}(\Pa^{\prime})  \qquad \forall j \in \{1,\ldots,B\}.
\end{align*}
\end{proof}
We are now ready to prove Lemma~\ref{lemma:Spec Gaps 1}.
\begin{proof}(Lemma~\ref{lemma:Spec Gaps 1})
Recall $\tK_0 = \tK_{\mid \scrP_0}$ for $\tK$ defined in Section~\ref{sec: SMC Background}. By Lemma~\ref{lemma:spec gap lemma} we have:
\begin{align}
\label{eqn:spec gap first proof}
\Gap(\tK_0)  \geq \frac{1}{\theta^2_{\scrD} \max_j \theta^2_{\D_j} } \, \Gap(\hat{\tK}_0).
\end{align}
for $\theta_{\scrD}$ defined in \eqref{Eqn:ThetaDef}
since the first condition of Lemma~\ref{lemma:spec gap lemma} is satisfied  by the bound \eqref{eqn:density ratio bounds} from Lemma~\ref{lemma: hat weight uniform bound} with $a_0 = \theta^{-1}_{\scrD}$ and $a_1 = \theta_{\scrD}$, and the second condition of Lemma~\ref{lemma:spec gap lemma} with $a_2 = (\max_j \theta_{\D_j})^{-2}$ holds by Lemma~\ref{lemma: Spec Gaps 1 intermediate}. Hence
Lemma~\ref{lemma:spec gap lemma} with $a_0 = \theta^{-1}_{\scrD}$, $a_1 = \theta_{\scrD}$, and $a_2 =(\max_j \theta_{\D_j})^{-2}$ implies \eqref{eqn:spec gap first proof}. Finally, 
recalling $\hat{\tK}_0$ is a product chain and applying Theorem~\ref{thm: Product Chain SpecGap} gives the result.
\end{proof}
It remains to bound $ \min_j  \Gap(\hat{\tK}_{0,j})$. We will again do so by obtaining uniform bounds on the density ratio appearing in the acceptance probability, which apply on the subspace $\scrP_{0,j}$
having a bounded number of extra mutations. That is, we will obtain constants $a_0,a_1 > 0$ satisfying
\begin{align}
\label{eqn:proof overview second spec gap}
a_0 \cdot \ind_{\scrP_0(\epsilon, \I_j)}(\Pa_{\I_j}) \leq \frac{\hat{\pi}_{0,j}(\Pa_{\I_j})}{\mu_{\I_j}(\Pa_{\I_j})}\cdot \ind_{\scrP_0(\epsilon, \I_j)}(\Pa_{\I_j}) \leq a_1 \cdot \ind_{\scrP_0(\epsilon, \I_j)}(\Pa_{\I_j}),
\end{align}
for $\scrP_{0}(\epsilon, \I_j) = \scrP_{0,j} = \{\Pa: m(\Pa_{\I_j}) \leq M_{\epsilon}(\I_j) \}$ the event \eqref{Eqn:ScriptQDef} that the number of extra mutations in $\I_j$ is bounded by $M_{\epsilon}(\I_j) = \zeta_{\I_j} \lambda(e) + \log(3B/\epsilon)$. Once \eqref{eqn:proof overview second spec gap} is established, recalling that $\mu_{\I_j}$ is the proposal distribution used by $\hat{\tK}_{0,j}$, a lower bound on $\Gap(\hat{\tK}_{0,j})$ will follow immediately.
\begin{restatable}{lemma}{SpecGapsTwo}
\label{lemma:Spec Gaps 2}
For each block $j$, define the block-specific constant 
\begin{align}
\label{eqn: theta I_j}
\theta_{\I_j} := \exp\left(\log(\phi_{\star}) M_{\epsilon}(\I_j) + 2TM_{\epsilon}(\I_j)(\tilde{\delta}+\delta) \right).
\end{align}
For any $j \in \{1,\ldots,B\}$:
\begin{align*}
\Gap(\hat{\tK}_{0,j}) \geq \frac{2}{\theta^4_{\I_j}}.
\end{align*}
\end{restatable}
\begin{proof}
Let $c = - T (\tg(\cdot; \y_{\I_j}) - \gamma(\cdot; \y_{\I_j}))$.
An identical argument to that used to show Lemma 2 in \citet{Mathews3:2025} 
gives the following uniform bound on the density ratio:
\begin{align}
\label{eqn:unnormalized ratio hat and ism}
\Prob_{\xi_{\I_j}}\left(\phi^{m(\Pa_{\I_j})}_{\min} e^{-Tm(\Pa_{\I_j})(\tilde{\delta} + \delta) + c } \leq \frac{\tP_{(T,\hat{\Q})}(\y_{\I_j}, \Pa_{\I_j} \mid \x_{\I_j})}{\tP_{(T,\Q)}(\y_{\I_j}, \Pa_{\I_j} \mid \x_{\I_j})} \leq \phi^{m(\Pa_{\I_j})}_{\max} e^{Tm(\Pa_{\I_j})(\tilde{\delta} + \delta) + c}\right) =1 ,
\end{align}
where $\xi_{\I_j}$ is any probability measure supported on $\scrP_{\I_j}$, and recall that
\begin{equation*}
m(\Pa_{\I_j}) \leq M_{\epsilon}(\I_j) \qquad \forall \Pa_{\I_j} \in \scrP_{0,j}.
\end{equation*}
Consequently
\begin{align}
\label{eqn:density ratio hat and ism}
\Prob_{\xi_{\I_j}}\left( \frac{\ind_{\scrP_{0,j}}(\Pa_{\I_j})}{\theta_{\I_j} \cdot \mu_{\I_j}(\scrP_{0,j})}   \leq  \frac{\hat{\pi}_{0,j}(\Pa_{\I_j})}{\mu_{\I_j}(\Pa_{\I_j})}  \cdot \ind_{\scrP_{0,j}}(\Pa_{\I_j}) \leq \frac{\theta_{\I_j} \cdot \ind_{\scrP_{0,j}}(\Pa_{\I_j}) }{\mu_{\I_j}(\scrP_{0,j})} \right) = 1.
\end{align}
Recall that by definition
\begin{align}
\label{eqn:spec gap restricted j}
\Gap(\hat{\tK}_{0,j}) &=  \inf_{\substack{f \in L^2(\hat{\pi}_{0,j})\\ \V_{\hat{\pi}_{0,j}}[f] \neq 0}} \frac{\calE_{\hat{\tK}_{0,j}}(f,f)}{\V_{\hat{\pi}_{0,j}}[f]}.
\end{align}
Recalling from \eqref{Eqn:P0jDef} that $\overline{\scrP}_{0,j} \subset \scrP_{0,j}$, we have by \eqref{eqn:density ratio hat and ism} that
\begin{equation*}
\V_{\hat{\pi}_{0,j}}[f] \leq \theta_{\I_j} \mu_{\I_j}(\overline{\scrP}_{0,j})\V_{\mu_{\I_j \mid \overline{\scrP}_{0,j}}}[f].
\end{equation*}
Next, by the definition \eqref{eqn:accept ratio alpha hat} of $\hat{\alpha}_j(\Pa_{\I_j},\Pa^{\prime}_{\I_j})$ we see again using \eqref{eqn:density ratio hat and ism} that
\begin{align}
\label{eqn:KhatAcceptLowerBound}
\hat{\alpha}_j(\Pa_{\I_j},\Pa^{\prime}_{\I_j}) \geq \frac{1}{\theta^2_{\I_j}} \qquad \forall \Pa_{\I_j},\Pa^{\prime}_{\I_j} \in \scrP_{0,j},
\end{align}
and since $\overline{\scrP}_{0,j} \subset \scrP_{0,j}$
\begin{align*}
\calE_{\hat{\tK}_{0,j}}(f,f) &=  \int_{\overline{\scrP}_{0,j}} \int_{\overline{\scrP}_{0,j}} \hat{\pi}_{0,j}(d\Pa_{\I_j})\hat{\tK}_{0,j}(\Pa_{\I_j},d\Pa_{\I_j}^{\prime}) (\nabla f(\Pa_{\I_j}, \Pa^{\prime}_{\I_j}))^2  \\
&= \int_{\overline{\scrP}_{0,j}} \int_{\overline{\scrP}_{0,j}} \hat{\pi}_{0,j}(d\Pa_{\I_j}) \mu_{\I_j}(d\Pa^{\prime}_{\I_j}) \hat{\alpha}_j(\Pa_{\I_j},\Pa^{\prime}_{\I_j})   (\nabla f(\Pa_{\I_j}, \Pa^{\prime}_{\I_j}))^2 \\
&\geq \frac{\mu_{\I_j}(\scrP_{0,j})}{\theta^3_{\I_j}} \int_{\overline{\scrP}_{0,j}} \int_{\overline{\scrP}_{0,j}}   \mu_{\I_j \mid \scrP_{0,j}}(d\Pa_{\I_j}) \mu_{\I_j \mid \scrP_{0,j}}(d\Pa^{\prime}_{\I_j})   (\nabla f(\Pa_{\I_j}, \Pa^{\prime}_{\I_j}))^2  \\
&=  \frac{\mu_{\I_j}(\overline{\scrP}_{0,j})}{\theta^3_{\I_j}}2 \V_{\mu_{\I_j \mid \overline{\scrP}_{0,j}}}[f],
\end{align*}
where the inequality uses \eqref{eqn:density ratio hat and ism}  and \eqref{eqn:KhatAcceptLowerBound}.
It follows that for any non-constant $f \in L^2(\hat{\pi}_{0,j})$ we have
\begin{align}
\label{eqn:spec gap two final}
\frac{\calE_{\hat{\tK}_{0,j}}(f,f)}{\V_{\hat{\pi}_{0,j}}[f]} \geq \frac{2}{\theta^{4}_{\I_j}}.
\end{align}
Taking the infimum on both sides gives the stated bound by \eqref{eqn:spec gap restricted j}
\end{proof}

Combining Lemmas~\ref{lemma:Spec Gaps 1} and \ref{lemma:Spec Gaps 2} provides us with a lower bound on $\Gap(\tK_0)$. Consequently, we immediately obtain a bound on the spectral gap of the lazy chain $\tK' := \frac{1}{2}(\tK +I)$ restricted to $\scrP_0$.
\begin{lemma}
\label{lemma: lazy chain gap}
Let $\theta_{\scrD}$ be the constant defined in \eqref{Eqn:ThetaDef} that depends on the full sets $\mathscr{D}$ and $\partial \mathscr{D}$, and let $\theta_{\D_j}$ and $\theta_{\I_j}$ be the block-specific constants for block $j$ defined in \eqref{Eqn:ThetaDef D}  and \eqref{eqn: theta I_j}, respectively. Then
\begin{align*}
\Gap(\tK'_0) \geq \frac{1}{B\theta^2_{\scrD} \max_j\theta^2_{\D_j} \max_j\theta^4_{\I_j}} .
\end{align*}
\end{lemma}
%
%
\begin{proof}
It suffices to lower bound $\Gap(\tK_0)$ since $\Gap(\tK'_0) = \frac{1}{2}\Gap(\tK_0)$. Applying 
Lemmas~\ref{lemma:Spec Gaps 1} and \ref{lemma:Spec Gaps 2}
\begin{align*}
\Gap(\tK_0) 
&\geq \frac{1}{B \theta^2_{\scrD} \max_j \theta^2_{\D_j} } \min_j \Gap(\hat{\tK}_{0,j}) \geq  \frac{2}{B \theta^2_{\scrD} \max_j \theta^2_{\D_j} \max_j \theta^4_{\I_j}}.
\end{align*}
The stated bound follows.
\end{proof}

\begin{lemma}
\label{lemma:Mixing Time}
Define the constants 
\begin{align*}
    c_1 = 2\lambda(e) \log(\phi_{\star}) \quad c_2 = 4\lambda(e)(\tilde{\delta} + \delta) \quad c_3 = \frac{3}{\lambda(e)}(c_1 + 2Tc_2).
\end{align*}
Then for any $\epsilon \in (0,1)$
\begin{align*}
\tau(\epsilon,\omega) &\leq B\log\left(\frac{80 \omega^4}{\epsilon^2} \right)\left(\frac{60 B \omega^2}{\epsilon^2} \right)^{c_3}\exp\left(c_1[2\max_j \zeta_{\I_j} + \max_j \zeta_{\D_j} + \zeta_{\scrD}]\right) \\
&\times \exp\left(c_2 T [\zeta_{\scrD} + \zeta_{\partial \scrD} + \max_j \zeta_{\D_j} + \max_j \zeta_{\partial \D_j} + 2\max_j\zeta_{\I_j}]\right).
\end{align*}
%
\end{lemma}
\begin{proof}
We will apply Theorem~\ref{thm:Atchade Bound} to the lazy (reversible) chain $\tK'_0$. From Lemma~\ref{lemma:DSM Conc Bounds} we have $\pi(\scrP_0(\epsilon^2/20 \omega^2) ) \geq 1 - \epsilon^2/20 \omega^2$ satisfying the conditions of Theorem~\ref{thm:Atchade Bound}. Lemma~\ref{lemma: lazy chain gap} gives
\begin{align*}
\Gap(\tK'_0) &\geq \frac{1}{B}\left(\frac{\epsilon^2}{60 \omega^2 B}\right)^{c_3} \exp\left(-c_1[2\max_j \zeta_{\I_j} + \max_j \zeta_{\D_j} + \zeta_{\scrD}]\right) \\
&\times \exp\left(-c_2 T [2\max_j\zeta_{\I_j} + \max_j \zeta_{\D_j} + \max_j \zeta_{\partial \D_j} + \zeta_{\scrD} + \zeta_{\partial \scrD}]\right)
\end{align*}

for the chosen $c_1,c_2,c_3 > 0$.
Applying Theorem~\ref{thm:Atchade Bound} yields the result. 
\end{proof}

We can now complete the proof of Theorem~\ref{thm: mixing time bound} in Section~\ref{sec:MainResults}.

\begin{proof}(Theorem~\ref{thm: mixing time bound})
By Lemma~\ref{lemma: r star and island partition}, there exists an island partition $\mathscr{I}$ satisfying Assumption~\ref{assumption:island partition}, such that
\begin{align*}
    r_{\mathcal{I}_j} \leq r_\star \quad \text{ and } \quad r_{\mathcal{D}_j} = 0 \quad \text{ for } j =1,\ldots,B.
\end{align*}

Consequently, $\max_j \zeta_{\D_j} = \bigO(1)$, $\zeta_{\scrD} = \bigO(1)$, and $\max_j \zeta_{\I_j} = \bigO(r_{\star})$
under Assumption~\ref{assump: T assumption}. Applying Lemma~\ref{lemma:Mixing Time} then yields the stated bound.
\end{proof}

\section{Sequential Monte Carlo for Endpoint-Conditioned CTMCs}
\label{sec: SMC}
The mixing time bound in Theorem~\ref{thm: mixing time bound} holds for the component Metropolis chain initialized according to a \textit{warm} starting distribution. However, obtaining a warm starting distribution is generally non-trivial and thus in most practical settings Theorem~\ref{thm: mixing time bound} cannot be applied directly. In this section we show how recent SMC complexity bounds given in \citet{Marion:2023} may be combined with our warm mixing time bound to provide finite sample error bounds for the SMC estimator \eqref{eqn: SMC estimator} of $p_{(T,\tQ)}(\y \mid \x)$. In particular, we will show that the complexity of the SMC sampler also grows at the same rate as our warm-start mixing time bound (at most exponentially in $r_{\star}$ rather than $r$). As in Section~\ref{sec: mixing time bound}, we first obtain a bound for general $r(n)$ and $T$, and then invoke  Assumption~\ref{assump: T assumption} to simplify in order to obtain  Theorem~\ref{thm:SMC bound}.
This demonstrates that the SMC algorithm introduced in Section~\ref{sec: SMC Background} (Algorithm~\ref{alg:smc}) provides a dramatic improvement in computational complexity over results available previously \cite{Mathews3:2025} for this problem, under conditions satisfied in most practical problems.

\subsection{Bounds for SMC}\label{sec:SMC results}
We state our main result for SMC as a consequence of the results given in Appendix~\ref{Appdx:SMC Results}. Before stating the main result (Theorem~\ref{thm: SMC Complexity Bound}) in Section~\ref{sec:SMC main result}, we briefly state a previous result for SMC obtained by \citet{Marion:2023} and then establish an upper bound on $\max_v L^2(\pi_v , \pi_{v-1})$ needed to apply the result of \citet{Marion:2023}. These two results will be used in conjunction with the warm mixing time bound obtained in Section~\ref{sec: mixing time bound} to establish our main result. 
\subsubsection{Notation and Previous Results}

\citet{Marion:2023,Marion:2018b} established finite sample complexity bounds for SMC in terms of the largest $2$-warm mixing time $\max_v\tau_{v}(\epsilon,2)$ and largest $L^2$ distance $\max_vL^2(\pi_v,\pi_{v-1})$. We will need this result below, along with the mixing time bound of the previous section, to show that SMC provides a randomized approximation scheme for $p_{(T,\tQ)}(\y \mid \x)$.

The following result bounds the relative error of the product estimator \eqref{eqn: SMC estimator} with high probability, and follows directly from  the bounds given in \citet{Marion:2018b}. (This statement is with respect to the probability measure of the full set of particles produced by the SMC algorithm; see \cite{Marion:2023} for details).
\begin{theorem}(\citet{Marion:2018b})
\label{thm:Marion SMC}
Let $\epsilon \in (0,1)$ and $\delta \in (0,1)$ be fixed and assume $\hat{\Pa}^{(1)}_0,\ldots,\hat{\Pa}^{(N)}_0 \iidsim \pi_0$. Let
\begin{enumerate}
\item $N\geq \max_v L^2(\pi_v,\pi_{v-1}) \max\big\{18 \log\big(5\delta^{-1} V\big), 20 \delta^{-1} \epsilon^{-2}  V^3\big\}$
\item $s \geq \max_v\; \tau_v \big( \frac{\delta}{5NV},\; 2\big).$
\end{enumerate}
Then with probability $1-\delta$
\begin{align*}
\left|\hat{z}_V(\hat{\Pa}^{1:N}_{1:V}) - z_V\right| \leq \epsilon z_V.
\end{align*}
\end{theorem}

The key requirements to apply this result are a bound on the largest $2$-warm mixing time $\max_v\tau_v(\epsilon,2)$ of the mutation kernels, a bound on the largest $L^2$ distance $\max_vL^2(\pi_v,\pi_{v-1})$, and the ability to generate \textit{i.i.d.} samples from $\pi_0$. The first follows immediately from the results of Section~\ref{sec: mixing time bound}, which  hold for any $k$-neighbor DSM  (see Appendix~\ref{sec: SMC complexity bound proof} for details). The latter two are obtained by judicious choice of the distribution sequence $\pi_0,\pi_1,\pi_2,\ldots,\pi_{V-1}$. To ensure the last condition, $\pi_0$ is chosen here to be an ISM; sampling from endpoint-conditioned paths under ISMs can be performed efficiently site-by-site using specialized algorithms (see e.g., \cite{Hobolth:2009}). Hence, to establish our main result, it to remains specify the sequence of distributions $\pi_1,\pi_2,\ldots,\pi_{V-1}$ and obtain a bound on $ \max_v L^2(\pi_v,\pi_{v-1})$, and then directly apply Theorem~\ref{thm:Marion SMC}.

\subsubsection{Bounding $L^2(\pi_v, \pi_{v-1})$}
A key step in obtaining an efficient SMC algorithm is specifying the distribution sequence in such a way that all neighboring distributions are sufficiently ``close''.  For DSMs, we can do so by tempering the interaction terms (see Section~\ref{sec: SMC Background}) with the difference in successive (inverse) temperatures $\beta_{v} - \beta_{v-1}$ chosen to be sufficiently small.
\begin{restatable}{theorem}{ChisqTempering}
\label{thm:chi square bound tempering}
Let $0 = \beta_0 < \beta_1 < \ldots < \beta_V = 1$ be a sequence of inverse temperatures. For any $k$-neighbor DSM (Assumption~\ref{assumption:neighborhood context}), there exists a constant $c_1 \in (0,1)$, independent of $\x$ and $\y$, such that if
\begin{align*}
\Db :=  \beta_v - \beta_{v-1} \leq \frac{c_1}{\zeta(n)} \qquad \text{ for } \quad v = 1,\ldots,V,
\end{align*}
then $\max_vL^2(\pi_{v},\pi_{v-1}) = \bigO(1)$. Consequently, $V = \bigO(n)$
temperatures and $N = \bigO(V^3) = \bigO(n^3)$ particles suffice to satisfy the first condition of Theorem~\ref{thm:Marion SMC}.
\end{restatable}

We will do so by obtaining a bound on $L^2(\pi_{v},\pi_{v-1})$ as a function of $\Db$. Let
\begin{align}\label{eqn: rate diff}
\psi_v(\Pa) = \psi_v(\s(\Pa),\bb(\Pa)) 
:=  \sum^{m(\Pa)}_{l=1}t^l \Delta\tg_v(l),
\end{align}
%
and note
\begin{align*}
w_v(\Pa) := \frac{\tP_{(T,\tQ_v)}(\y,\Pa \mid \x)}{\tP_{(T,\tQ_{v-1})}(\y,\Pa \mid \x)} = e^{-T (\tg_v(\cdot; \y) - \tg_{v-1}(\cdot; \y))}
\prod^{m(\Pa)}_{l=1} \phi^{\Db_v}(b^l; \tx^{l-1}_{s^l}) \me^{\psi_v(\Pa)  - \psi_{v-1}(\Pa)}.
\end{align*}
%
Then we can write
\begin{align}
\label{eqn:l2 smc}
L^2(\pi_v, \pi_{v-1}) = \E_{\pi_{v-1}}[w_v^2(\Pa)]/ \left(\E_{\pi_{v-1}}[w_v(\Pa)] \right)^2 = \E_{\pi_{v-1}}[\tilde{w}_v^2(\Pa)]/ \left(\E_{\pi_{v-1}}[\tilde{w}_v(\Pa)] \right)^2,
\end{align}
where
\begin{align*}
\tilde{w}_v(\Pa) \defeq \prod^{m(\Pa)}_{l=1} \phi^{\Db_v}(b^l; \tx^{l-1}_{s^l}) \me^{\psi_v(\Pa)  - \psi_{v-1}(\Pa)}.
\end{align*}
%

Our approach will be to bound the numerator in \eqref{eqn:l2 smc} by finding a constant $l$ which bounds the total number of mutations $m(\Pa)$ with high probability, and a uniform bound $\tilde{w}_v(\Pa^l) \leq e^{\theta l}$ for paths of length $l$, in order to decompose
\begin{align*}
\E_{\pi_{v-1}}[\tilde{w}_v^2(\Pa)] & = \E_{\pi_{v-1}}[\tilde{w}_v^2(\Pa) \ind_{m(\Pa) \leq l} ] + \E_{\pi_{v-1}}[\tilde{w}_v^2(\Pa) \ind_{m(\Pa) > l}]
\\
& \leq e^{\theta l} + \E^{\frac{1}{2}}_{\pi_{v-1}}[e^{\theta m(\Pa)}] \Prob_{\pi_{v-1}}^{\frac{1}{2}}(m(\Pa) > l) ,
\end{align*}
where the inequality uses the Cauchy-Schwarz inequality.  We can then apply the MGF bound from Lemma~\ref{lemma:DSM MGF Subset} along with Markov's inequality to bound the right-hand term.  Lower bounding the denominator $\left(\E_{\pi_{v-1}}[\tilde{w}_v(\Pa)] \right)^2$ follows by considering only length $r$ paths:
\begin{align*}
\left(\E_{\pi_{v-1}}[\tilde{w}_v(\Pa)] \right)^2 \geq \left(\E_{\pi_{v-1}}[\tilde{w}_v(\Pa) \mathbbm{1}_{m(\Pa) = r}(\Pa)] \right)^2.
\end{align*}
We first obtain the bound on $\tilde{w}_v(\Pa)$ as a function of the path length $m(\Pa)$.

\begin{lemma}
\label{lemma:weight bound tempering}
Let 
\begin{align*}
\tilde{\delta}_v \defeq 2q(k+1)\gamma_{\max} \max(1,\phi_{\max}) [\max(1,\phi^{\Db_v}_{\max}) - \min(1,\phi^{\Db_v}_{\min})  ].
\end{align*}
Then
\begin{align*}
\Prob_{\pi_{v-1}}\left(\phi^{m(\Pa)\Db_v }_{\min} e^{-Tm(\Pa) \tilde{\delta}_v}  \leq  \tilde{w}_v(\Pa) \leq \phi^{m(\Pa)\Db_v}_{\max} e^{Tm(\Pa) \tilde{\delta}_v } \right) = 1.
\end{align*}
\end{lemma}
\begin{proof}
First notice
\begin{align*}
\phi^{m(\Pa)\Db_v }_{\min} \leq \prod^{m(\Pa)}_{l=1} \phi^{\Db_v}(b^l; \tx^{l-1}_{s^l}) \leq \phi^{m(\Pa)\Db_v}_{\max}.
\end{align*}
Next, write
\begin{align*}
\Delta\tg_v(l) - \Delta\tg_{v-1}(l)  
&= 
\sum^n_{i=1} 
\sum^n_{i=1} \sum_{a \neq x^l_i} (\tg_v(a; x^l_i) - \tg_{v-1}(a; x^l_i)) \\
&- \sum^n_{i=1}\sum_{a \neq x^{l-1}_i} (\tg_v(a; x^{l-1}_i) - \tg_{v-1}(a; x^{l-1}_i)).
\end{align*}
Note that there are at most $2q(k+1)$ non-zero  summands since $\text{d}_{\text{H}}(\x^{l-1},\x^l)=1$. In addition, for any $l \in \{1,\ldots,m(\Pa)\}$:
\begin{align*}
\gamma_{\max}\phi^{\beta_{v-1}}_{\max}(\min(1,\phi_{\min}^{\Db_v})-1) \leq \tg_v(a; x^l_i) - \tg_{v-1}(a; x^l_i)  \leq \gamma_{\max}\phi^{\beta_{v-1}}_{\max}( \max(1,\phi^{\Db_v}_{\max}) -1).
\end{align*}
%
By the triangle inequality
\begin{align*}
\abs{\Delta\tg_v(l) - \Delta\tg_{v-1}(l)} & \leq 2q(k+1)\gamma_{\max} \max(1,\phi_{\max}) [\max(1,\phi^{\Db_v }_{\max}) - \min(1,\phi^{\Db_v}_{\min})  ] = \tilde{\delta}_v,
\end{align*}
using $\phi^{\beta_{v-1}}_{\max} \leq \max(1,\phi_{\max})$ since $\beta_{v-1} \in (0,1)$. The stated bound follows by the definition of $\tilde{w}_v(\Pa)$.
\end{proof}
The next lemma upper bounds the quantity  $\tilde{\delta}_v$ introduced in Lemma~\ref{lemma:weight bound tempering} by $c\Db $ for constant $c \in (0,\infty)$ that depends on the DSM rates. This upper bound will be used to prove Theorem~\ref{thm:chi square bound tempering}.
\begin{lemma}
\label{lemma:tilde delta bound DSM}
Suppose $\Db_v \leq 1 / \log(1+\phi_{\max})$. Then
\begin{align*}
[\max(1,\phi^{\Db_v}_{\max}) - \min(1,\phi^{\Db_v}_{\min})  ] \leq 
\log\left( \frac{\max(1,\phi^2_{\max})}{\min(1,\phi_{\min})} \right)\Db_v  \defeq  \bar{\phi}\Db_v,
\end{align*} 
and therefore $\tilde{\delta}_v \leq 2q(k+1) \gamma_{\max}\max(1,\phi_{\max}) \bar{\phi}\Db_v$.
\end{lemma}

\begin{proof}
Write
\begin{align*}
\max(1,\phi^{\Db_v}_{\max}) - \min(1,\phi^{\Db_v}_{\min})  = (\max(1,\phi^{\Db_v}_{\max}) - 1) + (1 - \min(1,\phi^{\Db_v}_{\min}) ).
\end{align*}
Focusing on the right-hand term, when $\phi_{\min}^{\Db_v} <1$, using $1 + x \leq e^x$ we obtain
\begin{align*}
1 - \phi^{\Db_v}_{\min} \leq \log(1/\min(1,\phi_{\min}))\Db_v.
\end{align*}
Next since $e^x \leq 1 + x + x^2 \leq 1+2x$ for $x < 1$ and  $\Db_v \leq \frac{1}{\log(1+\phi_{\max})}$ we have
\begin{align*}
\phi^{\Db_v}_{\max} - 1 \leq 2\log(\max(1,\phi_{\max}))\Db_v .
\end{align*}
It is straightforward to check that this implies the result.
\end{proof}

With these lemmas in hand, we are now ready to complete the proof of Theorem~\ref{thm:chi square bound tempering}.
\begin{proof}(Theorem~\ref{thm:chi square bound tempering})
Let
\begin{align*}
\theta_v &= 4\Db_v 2T\gamma_{\max}\max(1,\phi_{\max})q(k+1) \log\Big( \frac{\max(1,\phi^2_{\max})}{\min(1,\phi_{\min})} \Big) \\
&+   4\Db_v\max\left\{\log\left(\phi^{-1}_{\min}\right),\log(\phi_{\max})\right\}.
\end{align*}
We have $\theta_v > 0$ (unless $\phi \equiv 1$). For any positive integer $l$, we have by Lemmas~\ref{lemma:weight bound tempering} and \ref{lemma:tilde delta bound DSM}

\begin{align}
\E_{\pi_{v-1}}[\tilde{w}_v^2(\Pa)] & \leq \E_{\pi_{v-1}}[e^{\frac{\theta_v m(\Pa) }{2}} \ind_{m(\Pa) \leq l} ] + \E_{\pi_{v-1}}[e^{\frac{\theta_v m(\Pa) }{2}}\ind_{m(\Pa) > l}] \nonumber \\ &\leq e^{\frac{\theta_v l }{2}} + \E^{\frac{1}{2}}_{\pi_{v-1}}[e^{\theta_v m(\Pa)}] \Prob_{\pi_{v-1}}^{\frac{1}{2}}(m(\Pa) > l), \label{eqn:l2 bound first ineq}
\end{align}
where \eqref{eqn:l2 bound first ineq} follows by the Cauchy-Schwarz inequality. We now obtain an upper bound on the right hand side of \eqref{eqn:l2 bound first ineq} using Markov's inequality,
\begin{equation*}
\E^{\frac{1}{2}}_{\pi_{v-1}}[e^{\theta_v m(\Pa)}] \Prob_{\pi_{v-1}}^{\frac{1}{2}}(m(\Pa) > l) 
\leq \E_{\pi_{v-1}}[e^{\theta_v m(\Pa)}]e^{-\frac{\theta_v l }{2}} \\ 
=  e^{\frac{\theta_v l }{2}} (\E_{\pi_{v-1}}[e^{\theta_v m(\Pa)}]e^{-\theta_v l}  ).
\end{equation*}
Now let $l = 2\lambda_{v-1}(e) \zeta$, where $\lambda_{v-1}(\cdot)$ is the function $\lambda(\cdot)$  defined in Lemma~\ref{lemma:DSM MGF Subset} for the DSM with rates $\tg_{v-1}$, and recall $\zeta = r + rT + (n-r)T^2$.
Applying Lemma~\ref{lemma:DSM MGF Subset} we obtain
\begin{align*}
\E_{\pi_{v-1}}[e^{\theta_v m(\Pa)}]e^{-\theta_v2\lambda_{v-1}(e) \zeta }    \leq  e^{\theta_v \lambda_{v-1}(e)\zeta  }    e^{-\theta_v 2 \lambda_{v-1}(e) \zeta } \leq e^{-\theta_v \lambda_{v-1}(e) \zeta}  \leq 1.
\end{align*}
Hence, $\E_{\pi_{v-1}}[\tilde{w}_v^2(\Pa)] \leq 2e^{\theta_v \lambda_{v-1}(e) \zeta} $. By Lemmas~\ref{lemma:weight bound tempering} and \ref{lemma:tilde delta bound DSM} we have
\begin{align*}
\left(\E_{\pi_{v-1}}[\tilde{w}_v(\Pa)]\right)^2 \geq \left(\E_{\pi_{v-1}}[\tilde{w}_v(\Pa) \mathbbm{1}_{m(\Pa) \leq l}(\Pa)]\right)^{2}  \geq e^{-\frac{\theta_v l}{2}}(1 - \Prob_{\pi_{v-1}}(m(\Pa) > l))^2.
\end{align*}
By Markov's inequality and Lemma~\ref{lemma:DSM MGF Subset}, we have $\Prob_{\pi_{v-1}}(m(\Pa) > l) \leq \E_{\pi_{v-1}}[e^{m(\Pa)}] e^{-l} \leq e^{\lambda_{v-1}(e) \zeta} e^{-l} = e^{-\lambda_{v-1}(e) \zeta}$ since $l = 2\lambda_{v-1}(e) \zeta$. Thus we obtain 
\begin{align*}
1 - \Prob_{\pi_{v-1}}(m(\Pa) > l) \geq 1 - e^{-\lambda_{v-1}(e) \zeta} \geq  1 - e^{-1}.
\end{align*}
Consequently,
\begin{align*}
L^2(\pi_{v},\pi_{v-1}) = \frac{\E_{\pi_{v-1}}[\tilde{w}_v^2(\Pa)]}{\left(\E_{\pi_{v-1}}[\tilde{w}_v(\Pa)]\right)^2} \leq  2\left(\frac{e}{e-1}\right)^2 e^{2\theta_v \lambda_{v-1}(e) \zeta} \leq 2 e^3,
\end{align*}
where the final inequality follows by the definition of $\theta_v$ and choosing
\begin{align*}
\Db_v &\leq \frac{1}{\zeta} \left( 8\lambda_{v-1}(e) \log\left( \frac{\max(1,\phi^2_{\max})}{\min(1,\phi_{\min})} \right)  \left( 1+T\gamma_{\max}\max(1,\phi_{\max})q(k+1)\right) \right)^{-1},
\end{align*}
and we use $\max(\log\left(\phi^{-1}_{\min}\right),\log(\phi_{\max})) \leq \log\left( \max(1,\phi^2_{\max}) /\min(1,\phi_{\min}) \right)$.
\end{proof}

\subsubsection{Proof of Theorem~\ref{thm:SMC bound}}
\label{sec:SMC main result}
The first key requirement to apply Theorem~\ref{thm:Marion SMC} is a bound on the largest 2-warm mixing time $\max_v\tau_v(\epsilon,2)$ of the mutation MCMC kernels when initialized according to a warm start. The mixing time bound stated in Lemma~\ref{lemma:Mixing Time} (see Section~\ref{sec: mixing time bound}) provides a bound on $\max_v\tau_v$ for arbitrary DSMs. The second key requirement needed to apply Theorem~\ref{thm:Marion SMC} is a bound on $\max_v L^{2}(\pi_v,\pi_{v-1})$, which is provided by Theorem~\ref{thm:chi square bound tempering}. Combining these two results gives 
a bound on the runtime $NVs$ of the SMC algorithm necessary to approximate $p_{(T,\tQ)}(\y \mid \x)$ with $\epsilon$-relative error, provided by the following theorem.
\begin{theorem}
\label{thm: SMC Complexity Bound}
For a $k$-neighbor DSM (Assumption~\ref{assumption:neighborhood context}), the SMC algorithm (Algorithm~\ref{alg:smc}) 
provides a randomized approximation scheme for $p_{(T,\tQ)}(\y \mid \x)$ in time
\begin{align*}
\bigO\left( \mathrm{poly}\left(\epsilon^{-1},\zeta,B\right)  \exp(c \cdot  ( (2\max_j \zeta_{\I_j} + \max_j\zeta_{\D_j} + \zeta_{\mathscr{D}})(1+T) + (\zeta_{\partial \D_j} + \zeta_{\partial \mathscr{D}})T  ) )\right),
\end{align*}
where $\zeta_{\calA}$ was defined in \eqref{eqn: zeta definition} for any  subset $\calA \subset \{1,\ldots,n\}$ and $c = c(k,\phi_{\star},\tg_{\star}) \in (0,\infty)$ is a model-dependent constant such that 
\begin{align*}
c(k,\phi_{\star},\tg_{\star}) = \bigO(k\log(\phi_{\star})\log(\tg_{\star})),
\end{align*}
assuming $\max\{\tg_{\max},e\} \ll \phi_{\star}$, where $\phi_{\star} = \phi_{\max}/\phi_{\min}$, $\gamma_{\star} = \gmax/\gmin$ and $\tg_{\star} = \phi_{\star} \gamma_{\star}$.
\end{theorem}
The proof of Theorem~\ref{thm: SMC Complexity Bound} is deferred to Appendix~\ref{Appdx:SMC Results}. Theorem~\ref{thm:SMC bound} is a special case of Theorem~\ref{thm: SMC Complexity Bound}:
\begin{proof}(Theorem~\ref{thm:SMC bound})
    Recall that $\zeta_{\calA} := r_{\calA} + r_{\calA}T + (n_{\calA} - r_{\calA}) T^2$ for $\calA \subset \{1,\ldots,n\}$. Choosing $\mathscr{I}$ to be an island partition (Assumption~\ref{assumption:island partition}), we have that $\max_j r_{\D_j} = 0$ and $\max_j r_{\I_j} = r_\star$. Therefore, under Assumption~\ref{assump: T assumption}
    \begin{align*}
        (2\max_j \zeta_{\I_j} + \max_j\zeta_{\D_j} + \zeta_{\mathscr{D}})(1+T) + (\zeta_{\partial \D_j} + \zeta_{\partial \mathscr{D}})T   = \bigO(r_\star).
    \end{align*}
\end{proof}
Critically, the bound \eqref{eqn:SMC main thm body} in Theorem~\ref{thm:SMC bound} does \textit{not} grow exponentially in the observed mutation count $r$ but rather in the max island size $r_\star$.
As a result, the SMC algorithm provides a substantial improvement in computational complexity over the base importance sampler studied in \cite{Mathews3:2025}, which scales
exponentially in the \textit{sum} of the mutation counts. Indeed, the following Proposition is a direct corollary of Theorem 3 in \cite{Mathews3:2025}.
\begin{proposition}(\citet{Mathews3:2025})
\label{prop: IS island complexity}
The running time of the importance sampler studied in \citet{Mathews3:2025} grows as
\begin{align*}
\bigO\left(\exp(c \cdot r_{\star}^{\text{IS}}) \epsilon^{-2} \right),
\end{align*}
for model-dependent constant $c \in (0,\infty)$, where $r_{\star}^{\text{IS}} \defeq \sum_{\{j: r_{\I_{j}} > 1\}} r_{\I_{j}}$ is the total number of mutations observed in all islands of size greater than one. 
\end{proposition}
While Proposition~\ref{prop: IS island complexity} provides an upper bound on the complexity of the importance sampler, \cite{Mathews3:2025} also showed that the complexity \textit{necessarily} grows exponentially in $r$ by considering the following problem:
\begin{definition}[Island problem \cite{Mathews3:2025}]
\label{def: island problem}
Let $r_I(n) = r(n) / 2$ and consider the sequence $\x^{\star} = \text{T-(TCAT)}^{r_I}\text{-T}$ evolving to
$\y^{\star} = \text{T-(TTGT)}^{r_I}\text{-T}$
under the CpG model \eqref{eqn:CpG Island Rates} with context-dependent rates given by
\begin{align}
\label{eqn:CpG LB}
\tilde{\gamma}(b; \tx_i) = \gamma(b; x_i) \lambda^{\ind_{\text{CG}}(x_{i-1},x_i) + \ind_{\text{CG}}(x_i,x_{i+1})},
\end{align}
with $\gamma(b;b^{\prime}) \equiv 1$ for $b,b^{\prime} \in \{\tA,\tG,\tC,\tT \}$ and $\lambda \in (1,\infty)$. Approximate $p_{(T,\tQ)}(\y^{\star} \mid \x^{\star})$. 
\end{definition}
As noted, this problem provides a \textit{lower} bound on the sample complexity of the importance sampling algorithm \cite{Mathews3:2025}.  However, it follows from Theorem~\ref{thm: SMC Complexity Bound} that the SMC algorithm provides a fully polynomial time randomized approximation scheme (FPRAS) for the island problem:
\begin{corollary}
\label{cor:SMC on LB problem}
  Under the setting of Theorem 2 in \cite{Mathews3:2025}, the SMC algorithm using $r_I = B$ blocks corresponding to each of the $r_I$ subsequences provides a FPRAS for the island problem.
\end{corollary}
\begin{proof}
Since $\max_j r_{\I_j} = 2$, SMC provides a FPRAS for approximating the marginal likelihood by Theorem~\ref{thm: SMC Complexity Bound}.
\end{proof}

\subsubsection{Proof of Theorem~\ref{thm:Nonlocal}}
The proof of Theorem~\ref{thm:Nonlocal} follows that of Theorems~\ref{thm: mixing time bound} and \ref{thm:SMC bound}, with the division sites $\D_j$ replaced by edge sites $\I_{j,e} = \{i \in \I_j: \C_i \cap \I^c_j \neq \emptyset \}$ (the boundary sets $\partial \I_{j,e}$ are defined identically to $\partial \D_j$ -- see Section~\ref{sec:island partitions}). Indeed, Theorem~\ref{thm:Nonlocal} follows immediately by the following more general form of Lemma~\ref{lemma:DSM MGF Subset} in Appendix~\ref{Appdx:SMC Results}:
\begin{lemma}
\label{lemma:DSM MGF Subset Nonlocal}
Let $\calA \subset \{1,\ldots,n\}$ be a set of site indices and $\theta \in (0,\infty)$. Let $\calA_{\text{e}} \defeq \{i \in \calA :  \C_i \cap \calA^c \neq \emptyset\}$ be the set of \textit{edge sites} in $\calA$. Then there exists a model-dependent constant $\lambda(\theta)$ such that
\begin{align*}
    \E_{\pi}[\theta^{m(\Pa_{\calA})}] \leq e^{Tq\abs{\calA_{\text{e}}}(\tgmax - \tgmin)} e^{\lambda(\theta)\zeta_{\calA}}. 
\end{align*}
\end{lemma}
\begin{proof}(Theorem~\ref{thm:Nonlocal})
Recall that we require $\max_j|\I_{j,e}| = \bigO(r)$. Letting $\calA \in \{\I_j, \I_{j,e}, \partial \I_{j,e} \}$, we have $|\calA_{\text{e}}| = \bigO(r)$ since 
the context of each site is at most size $k$. Therefore, under Assumption~\ref{assump: T assumption}
\begin{align*}
    Tq\abs{\calA_{\text{e}}}(\tgmax - \tgmin) = \bigO(1), \quad \text{ for } \calA \in \{\I_j, \I_{j,e}, \partial \I_{j,e} \}.
\end{align*}
Hence, Lemma~\ref{lemma:DSM MGF Subset Nonlocal} can be used in place of the MGF bound for neighbor-dependent models (Lemma~\ref{lemma:DSM MGF Subset Main Body}) to generalize the mixing time bound (Lemma~\ref{lemma:Mixing Time}) and $L^2$ bound (Theorem~\ref{thm:chi square bound tempering}) to non-local context dependence. Theorem~\ref{thm:Nonlocal} then follows since $\max_j r_{\I_{j,e}} = 0$ and $\max_j r_{\I_j} = r_\star$ by assumption.
\end{proof}

\section{Numerical Results}
\label{sec: sim experiment}
We compare the performance of the SMC algorithm (Algorithm~\ref{alg:smc}) and direct importance sampling (IS) \cite{Mathews3:2025} on the Island problem (Definition~\ref{def: island problem}). \citet{Mathews3:2025} show that for this problem
%
    $L^2(\pi,\mu) = \Omega(\exp(c(\lambda) \cdot r_I))$,
%
for model-dependent constant $c(\lambda) \in (0,\infty)$, where $\mu$ is the ISM distribution with rate matrix $\Q$. Hence the number of samples needed for the IS algorithm to approximate the marginal sequence likelihood to within a fixed accuracy grows exponentially in the number of islands $r_I = r/2$.
In contrast, by Corollary~\ref{cor:SMC on LB problem} SMC provides a FPRAS for this problem.

Figure~\ref{fig: relative variance simulation} shows an empirical comparison of the two algorithms on this problem. To ensure comparable computational costs we set the number of IS samples to be
%
    $N_{\text{IS}} = N_{\text{SMC}} V s$,
%
where $N_{\text{SMC}}$ is the number of particles used in the SMC algorithm. Here $N_{\text{SMC}} V s$ is the number of mutation steps taken in one run of the SMC algorithm, which typically dominates its computational cost. Empirically, we found $N_{\text{SMC}} = 500$, $V = r_I$, and $s = r_I$ sufficient to obtain a large effective sample size (ESS) across repeated runs of the SMC algorithm. 

The right-hand panel of Figure~\ref{fig: relative variance simulation} shows the \textit{ESS fraction} of the two algorithms as a function of $r_I$, which for SMC is defined as the \textit{minimum} estimated ESS observed across all $V$ steps of an SMC run divided by $N_{\text{SMC}}$. Similarly, the ESS fraction for IS is the estimated ESS divided by $N_{\text{IS}}$. Results are averaged over $500$ simulation replicates for each $r_I \in \{1,\ldots,20\}$. (Note that since $V=r_I$, at $r_I=1$ the two algorithms are almost equivalent.) When $V > 1$, the ESS fraction for IS decreases exponentially in $r_I$, while for SMC it remains roughly constant, in fact growing slightly for $r_I > 3$.
%
%
\begin{figure}[h]
    \centering
    \includegraphics[width=1.0\linewidth]{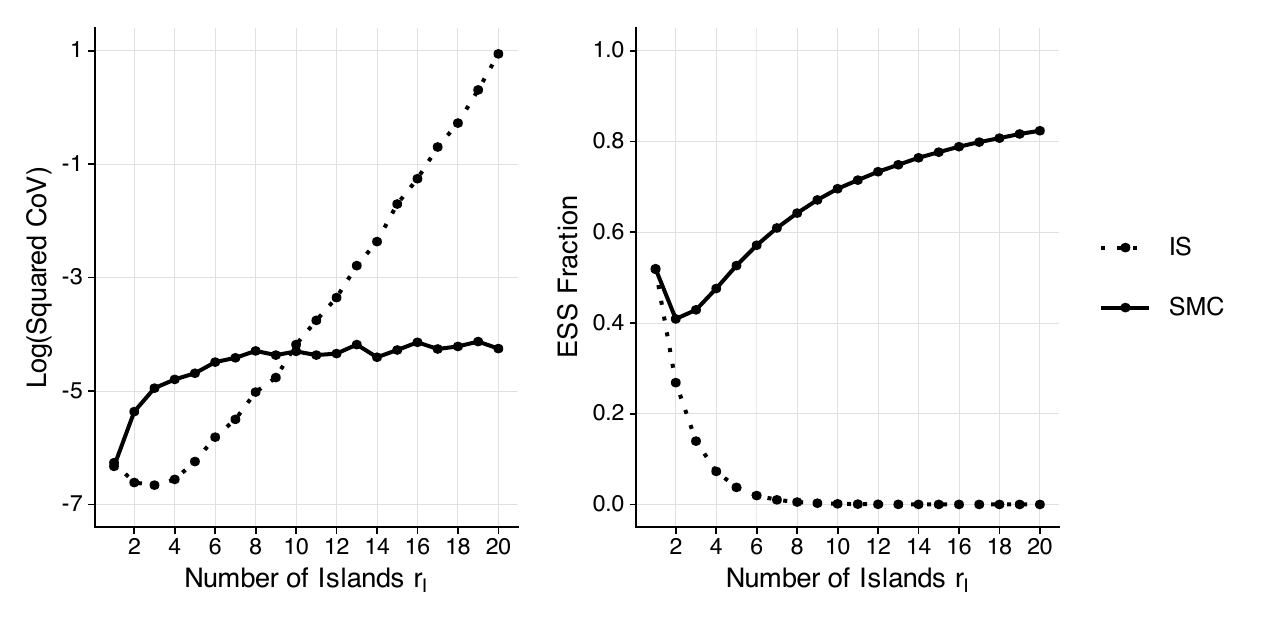}
    \caption{Numerical comparison of the SMC and IS algorithms on the Island Problem (Definition~\ref{def: island problem}), for $N_{\text{SMC}} = 500$, $V = r_I$, $s = r_I$, and $N_{\text{IS}} = N_{\text{SMC}} V s = 500 r^2_I$. In all cases, $T = 0.001$ and $\lambda = 50$. Left: Log of the estimated squared CoV of the SMC and IS estimators. Right: ESS fraction. For each $r_I$, values are based on 500 independent simulation replicates.} 
    \label{fig: relative variance simulation}
\end{figure}

The left panel of Figure~\ref{fig: relative variance simulation} shows the log of the estimated squared coefficient of variation (CoV)  of the IS estimator (see (13) in \cite{Mathews3:2025}) and the SMC estimator \eqref{eqn: SMC estimator}. Squared CoVs shown are computed using
$500$ simulation replicates for each $r_I \in \{1,\ldots,20\}$. We observe that the log of the estimated squared CoV for the IS estimator grows linearly,
consistent with the results of \cite{Mathews3:2025}. In contrast, the log of the estimated squared CoV of the SMC estimator \eqref{eqn: SMC estimator} grows logarithmically in $r_I$, 
as expected by Corollary~\ref{cor:SMC on LB problem}.
\section{Conclusion}
\label{sec:Conclusion} Calculation of marginal likelihoods under context-dependent evolutionary models is an important problem in phylogenetics and molecular evolution.  Given the extensive machinery available for independent site models, the idea of performing inference under DSMs by importance sampling from ISMs is an attractive one.  However, as shown by \citet{Mathews3:2025} the sample complexity of importance sampling on this problem grows exponentially in the number of observed mutations, which can be prohibitively expensive in some real-world applications. Here, we introduce an SMC algorithm for this problem, and show that this algorithm provides a significant improvement in the sample complexity required to accurately approximate the marginal likelihood. Section~\ref{sec: SMC} shows that when sites in $\calS$ can be grouped into islands consisting of neighboring mutated sites, the complexity of the SMC algorithm is at most exponential in the size of the largest island, yielding an exponential improvement over the importance sampler. Along the way, we have also obtained a bound on the mixing time of the component-wise Metropolis algorithm used previously in applications, under the assumption of a warm start, the first such result for an MCMC algorithm for this problem. It remains an open question whether this mixing time bound can be improved, and whether matching \textit{lower} bounds for this algorithm can be obtained, as well as to explore the effects of alternative MCMC algorithms for the mutation kernel of the SMC algorithm.
It is also worth noting again that our results hold under the assumption that $T = \bigO(r/n)$ and $r = \bigO(n^{\frac{1}{2}})$ (Assumption~\ref{assump: T assumption}); this   scaling assumption on $T$ is justified by the probability concentration tail bound on $T$ established in \citet{Mathews2:2025}, but that result differs from Assumption~\ref{assump: T assumption} by containing an additional factor of $\log(n)$.

Finally, the proof strategy used here combines several recent results from the literature on Monte Carlo theory.  In particular, we (1) established  concentration of the target distribution on a restricted set, (2) used that restriction to simplify the obtaining of spectral warm-start mixing bounds on the Markov kernel \cite{Atchade:2021}, and (3) applied recent results establishing warm-start conditions for SMC \cite{Marion:2023,Marion:2018b} to obtain finite sample error bounds for approximating the marginal sequence likelihood.  This strategy may be of broader interest for analyzing other problems of marginal likelihood and Bayes factor approximation, where  posterior concentration bounds may be available, or in some cases already exist, to satisfy the first step.

\bibliography{Bibliography.bib}

\bibliographystyle{imsart-nameyear.bst}

\appendix

\input{REVISION_CLEAN_AppendixD}

\end{document}

%% file: REVISION_CLEAN_AppendixD.tex
\section{Supporting Results for Sequential Monte Carlo}
\label{Appdx:SMC Results}
\subsection{Bound on MGF of $m(\Pa)$ 
Under $\pi$}
\label{Sec:MGFbound}
Recall that for a subset $\calA \subset \{1,\ldots,n\}$ we let $\x_{\calA}$ and $\y_{\calA}$ denote the corresponding subsequences and
\begin{align}\label{eqn: zeta definition appdx}
n_{\calA} := \abs{\calA} \quad\quad r_{\calA} := \text{d}_{\text{H}}(\x_{\calA}, \y_{\calA}) \quad\quad \zeta_{\calA} := r_{\calA} + r_{\calA}T + (n_{\calA} - r_{\calA}) T^2.
\end{align}
We state the bound for $m(\Pa_{\calA})$ for any subset of sites $\calA$; choosing $n_{\calA}=n$ yields the bound for $m(\Pa)$.
\begin{lemma}
\label{lemma:DSM MGF Subset}
Let $\calA \subset \{1,\ldots,n\}$ be a set of site indices and $\theta \in (0,\infty)$. Define
\begin{align*}
 \lambda_1(\theta) &:= q^2e^{Tq}  + \theta \tgmax e^{T\tilde{\delta}} q^2 e^{Tq\theta \tgmax  e^{T\tilde{\delta}} } \\
\lambda_2(\theta) &:= \log(\theta e^{2T\tilde{\delta}}\tgmax /\tgmin)  \\
\lambda_3(\theta) &:= q^2e^{Tq} +  \theta^2 \tg^2_{\max}  e^{2T\tilde{\delta}} q^2 e^{Tq \theta \tgmax  e^{T\tilde{\delta}}} \\
\lambda(\theta) & :=  \max\left\{ \lambda_1(\theta),\lambda_2(\theta),\lambda_3(\theta) \right\}.
\end{align*}
Let $\calA_{\text{e}} \defeq \{i \in \calA :  \C_i \cap \calA^c \neq \emptyset\}$ be the set of \textit{edge sites} in $\calA$. Then
\begin{align*}
    \E_{\pi}[\theta^{m(\Pa_{\calA})}] \leq e^{Tq\abs{\calA_{\text{e}}}(\tgmax - \tgmin)} e^{\lambda(\theta)\zeta_{\calA}}. 
\end{align*}
In particular, if $\pi$ is a \emph{$k$-neighborhood} DSM  (Assumption~\ref{assumption:neighborhood context}) and $\calA$ is a set of contiguous sites, 
\begin{equation}
\label{Eqn:DSMMGFBound}
\E_{\pi}[\theta^{m(\Pa_{\calA})}] \leq  e^{Tqk(\tgmax - \tgmin)} e^{\lambda(\theta)\zeta_{\calA}}.
\end{equation}
\end{lemma}
\begin{remark}
\label{remark: lambda complexity}
Under Assumption~\ref{assump: T assumption}, $\lambda(\theta) = \bigO(\log(\tg_{\star}))$ when $\max\{\tgmax,\theta\} \ll \tg_{\star}$,
where
$\tg_{\star} := \tgmax/\tgmin$ is the ratio of the maximum and minimum DSM rates. Indeed, recall that under Assumption~\ref{assump: T assumption}, $T = \bigO(r/n)$ and $r = \bigO(n^{\frac{1}{2}})$, in which case $T = o(1)$. Hence, $\lambda(\theta)$ is dominated by the $\bigO(\log(\tg_{\star}))$ term in $\lambda_2(\theta)$ after ignoring $o(1)$ terms involving $T$. 
\end{remark}

The proof of Lemma~\ref{lemma:DSM MGF Subset} will take advantage of two lemmas of \citet{Mathews3:2025}. The first provides an upper bound on the MGF of $m(\Pa)$ and a lower bound on $p_r \defeq \Prob_{\mu}(m(\Pa) = r)$, the probability of \textit{exactly} $r$ mutations, under an ISM $\mu$:
\begin{lemma}(\citet{Mathews3:2025})
\label{lemma:expectation and pr bound}
Let $\theta \in \mathbb{R}$ and $c = \gmax^2/\gmin q^2 e^{Tq(\gmax - \gmin)}$. Then
\begin{align*}
\E_{\mu}[\theta^{m(\Pa)}] \leq \theta^r\exp\left(rT \theta c \exp(Tq \theta \gmax) + (n-r)T^2 \theta^2 c\gmin \exp(Tq\theta \gmax )  \right).
\end{align*}
In addition, the following lower bound on $p_r \defeq \Prob_{\mu}(m(\Pa) = r)$ holds
\begin{align*}
p_r \geq 
\exp\left(- rTc\exp(T q \gmax) - (n-r)T^2c \exp(Tq\gmax)\gmin \right).
\end{align*}
\end{lemma}
The second lemma will be used to bound the exponential term appearing in the path density \eqref{eqn:Path Density} under the ISM and DSM. Let
\begin{align*}
\Delta^{\tg}(j) := \Delta^{\tg}(j; \Pa) &=  \tg(\cdot; \x^j) - \tg(\cdot; \x^{j-1}) \\
\Delta^{\gamma}(j) := \Delta^{\gamma}(j; \Pa) &=  \gamma(\cdot; \x^j) - \gamma(\cdot; \x^{j-1})
\end{align*}
and 
\begin{align*}
\tpsi(\s(\Pa),\bb(\Pa)) \defeq \sum^{m(\Pa)}_{j=1} t^j \Delta^{\tg}(j) \qquad\text{ and }\qquad \psi(\s(\Pa),\bb(\Pa)) \defeq \sum^{m(\Pa)}_{j=1} t^j \Delta^{\gamma}(j).
\end{align*}
\begin{restatable}{lemma}{LemmaOne}(\citet{Mathews3:2025})
\label{lemma: Uniform Bounds on Delta}
Let $q = a-1$ and define
\begin{align*}
\delta := q(\gmax - \gmin)
\qquad \text{and} \qquad 
\tilde{\delta} := q(k+1) (\tgmax - \tgmin).
\end{align*}
Then the following statements hold for the random variables $\Delta^{\tg}(j)$ and $\Delta^{\gamma}(j)$: 
\begin{enumerate}
\item $\Prob_{\mu}(|\tpsi(\s(\Pa),\bb(\Pa))| \leq m(\Pa) T\tilde{\delta} ) = 1 $
\item $\Prob_{\mu}(|\psi(\s(\Pa),\bb(\Pa))| \leq m(\Pa)T\delta ) = 1 $
\item $\Prob_{\mu}(|\tpsi(\s(\Pa),\bb(\Pa)) - \psi(\s(\Pa),\bb(\Pa))| \leq m(\Pa)T(\tilde{\delta}+\delta)) = 1 $.
\end{enumerate}
\end{restatable}
The proof of the DSM MGF bound \eqref{Eqn:DSMMGFBound} proceeds by relating the DSM $\pi$ to a modified DSM where the subset of sites in $\calA$ evolve according to a \textit{standard symmetric evolution model}, while the remaining sites in $\calA^c$ continue to evolve under the original DSM rates. The standard symmetric evolution model $\Q^{\text{sym}}$ is defined by
\begin{align}
\label{eqn:Standard Symmetric}
\Q_{\x,\x^{\prime}}^{\text{sym}} = 1 \text{ if } \text{d}_{\text{H}}(\x,\x^{\prime}) = 1 \quad\quad \Q_{\x,\x^{\prime}}^{\text{sym}} = 0 \text{ if } \text{d}_{\text{H}}(\x,\x^{\prime}) > 1.
\end{align}
For example, the standard symmetric evolution model for DNA ($\mathscr{A} = \{\tA,\tG,\tC,\tT\}$) is the JC69 model \cite{Jukes:1969} with unit rate ($\gamma \equiv 1$). Observe that any standard symmetric evolution model is an ISM and the number of non-zero elements along a given row of $\Q^{\text{sym}}$ is equal to $nq = n(\abs{\mathscr{A}} - 1)$.
\begin{proof}
We first define the modified DSM model; from there we can then apply  Lemma~\ref{lemma:expectation and pr bound}. Let $\pi^{\prime}$ be a DSM with rate matrix $\tQ^{\prime}$ such that $\tQ^{\prime}_{\x,\x^{\prime}} = 0$ if $\text{d}_{\text{H}}(\x,\x^{\prime}) > 1$, and define $\tQ^{\prime}_{\x,\x^{\prime}}$ for $\text{d}_{\text{H}}(\x,\x^{\prime}) = 1$ by
\begin{align}
\label{eqn: pi hat rates}
\tg^{\prime}_i(b; \tx_i) =
\begin{cases} \tg_i(b; \tx_i) & \text{ for } i \notin \calA \\
1 & \text{ for } i \in \calA ,
\end{cases}
\end{align}
so $\Pa_{\calA}$ is  distributed according to a standard symmetric evolution model with endpoint constraints $\x_{\calA}$ and $\y_{\calA}$.
We will bound the likelihood ratio $\tP_{(T, \tQ)}(\y, \Pa \mid \x)/\tP_{(T, \tQ^{\prime})}(\y, \Pa \mid \x)$,
 considering the terms in \eqref{eqn:Path Density} in turn. We first consider the product of rates, where we have
\begin{align}
\label{eqn: product gamma upper bound}
\prod^{m(\Pa)}_{l=1} \tg_{s^l}(b^l; \tx^{l-1}_{s^l}) &= \prod_{l : s^l \in \calA}\tg_{s^l}(b^l; \tx^{l-1}_{s^l}) \prod_{l : s^l \notin \calA}\tg_{s^l}(b^l; \tx^{l-1}_{s^l}) \leq \tg^{m(\Pa_{\calA})}_{\max}  \prod^{m(\Pa)}_{l=1} \tg^{\prime}_{s^l}(b^l; \tx^{l-1}_{s^l}),
\end{align}
with the last inequality holding since $\prod_{l : s^l \in \calA} \tg_{s^l}'(b^l; \tx_{s^l}^{l-1}) = 1$. Similarly,
\begin{align}
\label{eqn: product gamma lower bound}
\prod^{m(\Pa)}_{l=1} \tg_{s^l}(b^l; \tx^{l-1}_{s^l}) \geq \tg^{m(\Pa_{\calA})}_{\min}  \prod^{m(\Pa)}_{l=1} \tg^{\prime}_{s^l}(b^l; \tx^{l-1}_{s^l}).
\end{align}
Next we consider the exponential terms in \eqref{eqn:Path Density}. First recall that $\calA_{\text{e}} = \{i \in \calA :  \C_i \cap \calA^c \neq \emptyset\}$ denotes the \textit{edge} sites of $\calA$, and let $\calA_{\text{int}} = \calA \setminus \calA_{\text{e}}$ the \textit{interior} sites,  so $\calA = \calA_{\text{int}} \cup \calA_{\text{e}}$ is a partition of $\calA$. Note that $|\calA_{\text{e}} | \leq k$ if the sites in $\calA$ are assumed to be contiguous,
and the context of each site is limited to its $k$-neighborhood (Assumption~\ref{assumption:neighborhood context}). Under $\pi^{\prime}$, the rate that site $i \in \calA$ mutates is $\gamma^{\prime}_i(\cdot; x_i(t)) \equiv q = |\mathscr{A}| - 1$ since $\gamma^{\prime}_i(b;x_i(t)) \equiv 1$ for $i \in \calA$ and $b \neq x_i(t)$ by \eqref{eqn: pi hat rates}. Therefore, the rate at which $\x_{\calA}(t)$ mutates under the modified DSM $\pi^{\prime}$ is
\begin{align*}
\sum_{i \in \calA} \gamma^{\prime}_i(\cdot; x_i(t)) = \sum_{i \in \calA} q = q|\calA|.
\end{align*}
Now returning to the exponential terms in \eqref{eqn:Path Density}, write
\begin{align}
\label{eqn: seq mutation rate tilde and prime}
\tg(\cdot; \x^{l-1}) = \sum^n_{i=1} \tg_i(\cdot; \tx^{l-1}_i) \nonumber &= \tg^{\prime}(\cdot; \x^{l-1}) + \sum_{i \in \calA} \tg_i(\cdot; \tx^{l-1}_i) - q\abs{\calA} \nonumber \\ 
&\leq \tg^{\prime}(\cdot; \x^{l-1}) + 
 q\abs{\calA_{\text{e}}} \tgmax + \sum_{i \in \calA_{\text{int}}} \tg_i(\cdot; \tx^{l-1}_i)  - q\abs{\calA}.
\end{align}
Applying the upper bound \eqref{eqn: seq mutation rate tilde and prime} we obtain
\begin{align}
\label{eqn: exponential term tilde and prime upper bound}
\sum^{m(\Pa)}_{l=1} \Delta^t(l) \tg(\cdot; \x^{l-1}) &\leq \sum^{m(\Pa)}_{l=1} \Delta^t(l) \tg^{\prime}(\cdot; \x^{l-1}) + Tq(\abs{ \calA_{\text{e}}}\tgmax - |\calA|) \\
&+ \sum^{m(\Pa)}_{l=1} \Delta^t(l) \sum_{i \in \calA_{\text{int}}} \tg_i(\cdot; \tx^{l-1}_i).
\end{align}
The right hand side of \eqref{eqn: exponential term tilde and prime upper bound} can be upper bounded using Lemma~\ref{lemma: Uniform Bounds on Delta} (recall $\tilde{\delta} := q(k+1)(\tgmax - \tgmin)$):
\begin{align}
\label{eqn: bound on interior mutation rate}
\sum^{m(\Pa)}_{l=1} \Delta^t(l) \sum_{i \in \calA_{\text{int}}} \tg_i(\cdot; \tx_i^{l-1})  &=  \sum^{m(\Pa)}_{l=1} t^{l} \sum_{i \in \calA_{\text{int}}} (\tg_i(\cdot; \tx^{l}_i) - \tg_i(\cdot; \tx^{l-1}_i)) - T\sum_{i \in \calA_{\text{int}}} \tg_i(\cdot; \tilde{y}_i) \nonumber \\
&= \sum^{m(\Pa)}_{\{l: s^{l} \in \calA \}} t^{l} \sum_{i \in \calA_{\text{int}}} (\tg_i(\cdot; \tx^{l}_i) - \tg_i(\cdot; \tx^{l-1}_i)) - T\sum_{i \in \calA_{\text{int}}} \tg_i(\cdot; \tilde{y}_i) \nonumber \\
&\leq T\tilde{\delta} m(\Pa_{\calA})  - T\sum_{i \in \calA_{\text{int}}} \tg_i(\cdot; \tilde{y}_i).
\end{align}
The second equality follows since $\sum_{i \in \calA_{\text{int}}} (\tg_i(\cdot; \tx^{l}_i) - \tg_i(\cdot; \tx^{l-1}_i)) = 0$ if $s^{l} \notin \calA = \calA_{e} \cup \calA_{\text{int}}$ as the mutation rates of sites in $\calA$ are unchanged in this case. The final inequality follows by Lemma~\ref{lemma: Uniform Bounds on Delta} since
\begin{align*}
\sum_{i \in \calA_{\text{int}}} (\tg_i(\cdot; \tx^l_i) - \tg_i(\cdot; \tx^{l-1}_i)) \leq \tg(\cdot; \x^l) - \tg(\cdot; \x^{l-1}) \leq \tilde{\delta}.
\end{align*}
For brevity, denote the constant $c = T(q\abs{\calA} + \sum_{i \in \calA_{\text{int}}} \tg_i(\cdot; \tilde{y}_i))$. Using the bound \eqref{eqn: bound on interior mutation rate}, we obtain by \eqref{eqn: exponential term tilde and prime upper bound}
\begin{align*}
\sum^{m(\Pa)}_{l=1} \Delta^t(l) \tg(\cdot; \x^{l-1}) \leq \sum^{m(\Pa)}_{l=1} \Delta^t(l) \tg^{\prime}(\cdot; \x^{l-1}) + Tq\abs{ \calA_{\text{e}}}\tgmax + T\tilde{\delta} m(\Pa_{\calA}) - c,
\end{align*}
yielding a lower bound for the exponential terms in \eqref{eqn:Path Density}:
\begin{align}
e^{-\sum^{m(\Pa)}_{l=1} \Delta^t(l) \tg(\cdot; \x^{l-1}) } \geq e^{ -\sum^{m(\Pa)}_{l=1} \Delta^t(l) \gamma^{\prime}(\cdot; \x^{l-1}) } e^{-Tq| \calA_{e}|\tgmax -T\tilde{\delta}m(\Pa_{\calA}) + c }.
\label{eqn: sum gamma lower bound}
\end{align}
A similar argument yields the upper bound
\begin{align}
 e^{-\sum^{m(\Pa)}_{l=1} \Delta^t(l) \tg(\cdot; \x^{l-1}) } \leq e^{ -\sum^{m(\Pa)}_{l=1} \Delta^t(l) \gamma^{\prime}(\cdot; \x^{l-1}) } e^{-Tq| \calA_{e}|\tgmin + T\tilde{\delta} m(\Pa_{\calA}) + c}.
\label{eqn: sum gamma upper bound}
\end{align}
Combining (\ref{eqn: product gamma upper bound},\ref{eqn: product gamma lower bound}) and (\ref{eqn: sum gamma lower bound}, \ref{eqn: sum gamma upper bound}) and applying to \eqref{eqn:Path Density} yields the uniform bounds
\begin{align*}
\tg^{m(\Pa_{\calA })}_{\min} e^{-T \tilde{\delta} m(\Pa_{ \calA }) }  e^{-Tq| \calA_{\text{e}}|\tgmax + c} \; \leq \; 
 \frac{\tP_{(T, \tQ)}(\y, \Pa \mid \x)}{\tP_{(T, \tQ^{\prime})}(\y, \Pa \mid \x)} \; \leq \; \tg^{m(\Pa_{\calA})}_{\max} e^{T \tilde{\delta} m(\Pa_{\calA}) }  e^{-Tq| \calA_{e}|\tgmin + c} .
\end{align*}
It follows that
\begin{align*}
\E_{\pi}[\theta^{m(\Pa_{\calA })}] &= \frac{\int_{\scrP} \theta^{m(\Pa_{ \calA})} \tP_{(T, \tQ)}(\y, \Pa \mid \x) \nu(d\Pa)}{\int_{\scrP}  \tP_{(T, \tQ)}(\y, \Pa \mid \x) \nu(d\Pa)} \\
&\leq e^{Tq\abs{\calA_{e}}(\tgmax - \tgmin)} \frac{\E_{\pi^{\prime}}[ (\theta\tgmax e^{T\tilde{\delta}})^{m(\Pa_{\calA })}]}{\E_{\pi^{\prime}}[(\tgmin e^{-T \tilde{\delta}})^{m(\Pa_{\calA })}]}.
\end{align*}
Now recalling that $\Pa_{\calA}$ has marginal distribution under the modified DSM $\pi^{\prime}$ given by a standard symmetric evolution model, we can apply Lemma~\ref{lemma:expectation and pr bound} to the denominator. In particular, letting $\mu^{\text{Sym}}(\cdot \mid \x,\y)$ denote an endpoint-constrained ISM with rate matrix $\Q^{\text{Sym}}$ defined in \eqref{eqn:Standard Symmetric} we have 
\begin{align*}
\E_{\pi^{\prime}}[(\tgmin e^{-T \tilde{\delta}})^{m(\Pa_{\calA })}] &= \E_{\mu^{\text{Sym}}}[(\tgmin e^{-T \tilde{\delta}})^{m(\Pa_{\calA })}] \\
&\geq \tgmin^{r_{\calA}} \me^{-r_{\calA}T\tilde{\delta}} \Prob_{\mu^{\text{sym}}}(m(\Pa_{\calA}) = r_{\calA} \mid \x_{\calA}, \y_{\calA}).
\end{align*}
By Lemma~\ref{lemma:expectation and pr bound}
\begin{align*}
\Prob_{\mu^{\text{sym}}}(m(\Pa_{\calA}) = r_{\calA} \mid \x_{\calA}, \y_{\calA}) &\geq 
\exp\left(- r_{\calA} c^{\prime} \exp(T q \gmax)T\right) \\
&\times \exp\left(- (n_{\calA}-r_{\calA}) c^{\prime} \exp(T q\gmax)\gmin T^2\right) \\
&= \exp(-r_{\calA}q^2\exp(Tq) T - (n_{\calA} - r_{\calA})q^2\exp(Tq)T^2),
\end{align*}
where $c^{\prime} = \gmax^2/\gmin q^2 e^{Tq(\gmax - \gmin)} = q^2$ since $\gmax = \gmin = 1$ under $\Q^{\text{sym}}$.
Applying Lemma~\ref{lemma:expectation and pr bound} again to the numerator, we obtain
\begin{align*}
\E_{\pi^{\prime}}[(\theta \tgmax e^{T\tilde{\delta}})^{m(\Pa_{\calA})} ] =  \E_{\mu^{\text{sym}}}[(\theta \tgmax e^{T\tilde{\delta}})^{m(\Pa_{\calA})} ]  
& \leq e^{ r_{\calA} c_1 + r_{\calA}T c_2 +  (n_{\calA}-r_{\calA})T^2 c_3 },
\end{align*}
where
\begin{align*}
c_0 = \theta \tgmax e^{T\tilde{\delta}} \qquad c_1 = \log(c_0) \qquad c_2 = c_0 q^2 e^{c_0 Tq } \qquad c_3 = c_0 c_2.
\end{align*}
By the definition of $\lambda(\theta)$
\begin{align*}
\E_{\pi}[\theta^{m(\Pa_{\calA})}] &\leq e^{Tq\abs{\calA_{\text{e}}}(\tgmax - \tgmin)} \frac{\E_{\pi^{\prime}}[ (\theta\tgmax e^{T\tilde{\delta}})^{m(\Pa_{\calA})}]}{\E_{\pi^{\prime}}[(\tgmin e^{-T \tilde{\delta}})^{m(\Pa_{\calA})}]} \leq e^{Tq\abs{\calA_{\text{e}}}(\tgmax - \tgmin)}e^{\lambda(\theta) \zeta_{\calA}}.
\end{align*}
The stated bound follows.
\end{proof}
\subsection{Proof of Theorem~\ref{thm: SMC Complexity Bound}}\label{sec: SMC complexity bound proof}

\begin{proof}(Theorem~\ref{thm: SMC Complexity Bound})
The proof follows by combining the bound on $\max_v L^{2}(\pi_v,\pi_{v-1})$ obtained in  Theorem~\ref{thm:chi square bound tempering} and the $\omega$-warm mixing time bound for arbitrary DSMs obtained in Lemma~\ref{lemma:Mixing Time} to satisfy the two conditions of Theorem~\ref{thm:Marion SMC}. Indeed, first observe that by Theorem~\ref{thm:chi square bound tempering} we can guarantee $\max_v L^{2}(\pi_v,\pi_{v-1}) = \bigO(1)$ by choosing $V = \bigO(\zeta)$, satisfying the first condition of Theorem~\ref{thm:Marion SMC} by choosing $N = \bigO(\epsilon^{-2}V^{3}) = \bigO(\epsilon^{-2} \zeta^3)$. Next, we need to bound $\max_v \tau_v(\frac{\delta}{5NV}, 2)$, where $\delta \in (0,1)$ and $\tau_v$ is the warm mixing time for the kernel $\tK_v$ targeting $\pi_v$. To bound this quantity, recall that Lemma~\ref{lemma:Mixing Time} holds for $\tK$ targeting an arbitrary DSM $\pi$. Hence, consider $\tK_v$ targeting $\pi_v$ \eqref{Eqn:TemperedPi} and let $c_{v,1}$, $c_{v,2}$, and $c_{v,3}$ be the corresponding constants defined in  Lemma~\ref{lemma:Mixing Time}. Applying the bound of Lemma~\ref{lemma:Mixing Time} with $c'_1 = \max c_{v,1}$, $c'_2 = \max_v c_{v,2}$, and $c'_3 = \max c_{v,3}$
and choosing the warmness parameter $\omega = 2$ and $N= \delta/5V\epsilon$ for $\delta \in (0,1)$, yields the bound on $\max_v \tau_v(\frac{\delta}{5NV}, 2)$, satisfying the second condition of Theorem~\ref{thm:Marion SMC}.
\end{proof}
We note that $c^{\prime}$ in Theorem~\ref{thm: SMC Complexity Bound} is $\bigO(k\log(\phi_{\star})\log(\tg_{\star}))$ under certain conditions.
Specifically, assuming $\max\{\tg_{\max},e\} \ll \phi_{\star}$, then we have under Assumption~\ref{assump: T assumption} that for $c^{\prime}_1,c^{\prime}_2$ and $c^{\prime}_3$ defined above in the proof of Theorem~\ref{thm: SMC Complexity Bound} 
\begin{align*}
\max\{c^{\prime}_1,c^{\prime}_2,c^{\prime}_3\} = \bigO(k\log(\phi_{\star})\log(\tg_{\star})),
\end{align*}
where we used Remark~\ref{remark: lambda complexity}, which implies $\lambda(e) = \bigO(\log(\tg_{\star}))$ and so $c^{\prime}_1 = \bigO(\log(\tg_{\star})\log(\phi_{\star}))$ and $c^{\prime}_2 =  \bigO(\log(\tg_{\star}) k)$ (since $\tilde{\delta} = \bigO(k)$) by Lemma~\ref{lemma:Mixing Time}. 

%% file: REVISION_CLEAN_SMC_Context_Dependent.bbl
\begin{thebibliography}{56}

\bibitem[\protect\citeauthoryear{Arndt and Hwa}{2005}]{Arndt:2005}
\begin{barticle}[author]
\bauthor{\bsnm{Arndt},~\bfnm{P.~F.}\binits{P.~F.}} \AND
  \bauthor{\bsnm{Hwa},~\bfnm{T.}\binits{T.}}
(\byear{2005}).
\btitle{Identification and Measurement of Neighbour-Dependent Nucleotide
  Substitution Processes}.
\bjournal{Bioinformatics}
\bvolume{21}
\bpages{2322--2328}.
\end{barticle}
\endbibitem

\bibitem[\protect\citeauthoryear{Atchad\'{e}}{2021}]{Atchade:2021}
\begin{barticle}[author]
\bauthor{\bsnm{Atchad\'{e}},~\bfnm{Y.~F.}\binits{Y.~F.}}
(\byear{2021}).
\btitle{Approximate Spectral Gaps for {M}arkov Chain Mixing Times in High
  Dimensions}.
\bjournal{SIAM Journal on Mathematics of Data Science}
\bvolume{3}
\bpages{854-872}.
\end{barticle}
\endbibitem

\bibitem[\protect\citeauthoryear{Brooks and Gelman}{1998}]{Brooks:1998}
\begin{barticle}[author]
\bauthor{\bsnm{Brooks},~\bfnm{S.~P.}\binits{S.~P.}} \AND
  \bauthor{\bsnm{Gelman},~\bfnm{A.}\binits{A.}}
(\byear{1998}).
\btitle{General Methods for Monitoring Convergence of Iterative Simulations}.
\bjournal{Journal of Computational and Graphical Statistics}
\bvolume{7}
\bpages{434--455}.
\end{barticle}
\endbibitem

\bibitem[\protect\citeauthoryear{Chib}{1995}]{Chib:1995}
\begin{barticle}[author]
\bauthor{\bsnm{Chib},~\bfnm{S.}\binits{S.}}
(\byear{1995}).
\btitle{Marginal Likelihood from the {G}ibbs Output}.
\bjournal{Journal of the American Statistical Association}
\bvolume{90}
\bpages{1313--1321}.
\end{barticle}
\endbibitem

\bibitem[\protect\citeauthoryear{Chopin}{2002}]{Chopin:2002}
\begin{barticle}[author]
\bauthor{\bsnm{Chopin},~\bfnm{Nicolas}\binits{N.}}
(\byear{2002}).
\btitle{A sequential particle filter method for static models}.
\bjournal{Biometrika}
\bvolume{89}
\bpages{539--551}.
\end{barticle}
\endbibitem

\bibitem[\protect\citeauthoryear{Christensen, Hobolth and
  Jensen}{2005}]{Christensen:2005}
\begin{barticle}[author]
\bauthor{\bsnm{Christensen},~\bfnm{O.~F.}\binits{O.~F.}},
  \bauthor{\bsnm{Hobolth},~\bfnm{A.}\binits{A.}} \AND
  \bauthor{\bsnm{Jensen},~\bfnm{J.~L.}\binits{J.~L.}}
(\byear{2005}).
\btitle{Pseudo-Likelihood Analysis of Context-Dependent Codon Substitution
  Models}.
\bjournal{Journal of Computational Biology}
\bvolume{12}
\bpages{1166--1182}.
\end{barticle}
\endbibitem

\bibitem[\protect\citeauthoryear{Cowles and Carlin}{1996}]{Cowles:1996}
\begin{barticle}[author]
\bauthor{\bsnm{Cowles},~\bfnm{M.~K.}\binits{M.~K.}} \AND
  \bauthor{\bsnm{Carlin},~\bfnm{B.~P.}\binits{B.~P.}}
(\byear{1996}).
\btitle{{M}arkov Chain {M}onte {C}arlo Convergence Diagnostics: A Review}.
\bjournal{Journal of the American Statistical Association}
\bvolume{91}
\bpages{883--904}.
\end{barticle}
\endbibitem

\bibitem[\protect\citeauthoryear{Del~Moral, Doucet and
  Jasra}{2006}]{DelMoral:2006}
\begin{barticle}[author]
\bauthor{\bsnm{Del~Moral},~\bfnm{Pierre}\binits{P.}},
  \bauthor{\bsnm{Doucet},~\bfnm{Arnaud}\binits{A.}} \AND
  \bauthor{\bsnm{Jasra},~\bfnm{Ajay}\binits{A.}}
(\byear{2006}).
\btitle{Sequential Monte Carlo samplers}.
\bjournal{Journal of the Royal Statistical Society: Series B (Statistical
  Methodology)}
\bvolume{68}
\bpages{411--436}.
\end{barticle}
\endbibitem

\bibitem[\protect\citeauthoryear{Diaconis and
  Saloff-Coste}{1993}]{Diaconis:1993}
\begin{barticle}[author]
\bauthor{\bsnm{Diaconis},~\bfnm{P.}\binits{P.}} \AND
  \bauthor{\bsnm{Saloff-Coste},~\bfnm{L.}\binits{L.}}
(\byear{1993}).
\btitle{Comparison Techniques for Random Walk on Finite Groups}.
\bjournal{The Annals of Probability}
\bvolume{21}
\bpages{2131--2156}.
\end{barticle}
\endbibitem

\bibitem[\protect\citeauthoryear{Diaconis and
  Saloff-Coste}{1996}]{Diaconis:1996}
\begin{barticle}[author]
\bauthor{\bsnm{Diaconis},~\bfnm{P.}\binits{P.}} \AND
  \bauthor{\bsnm{Saloff-Coste},~\bfnm{L.}\binits{L.}}
(\byear{1996}).
\btitle{Logarithmic Sobolev inequalities for finite Markov chains}.
\bjournal{Annals of Applied Probability}
\bvolume{6}
\bpages{695--750}.
\end{barticle}
\endbibitem

\bibitem[\protect\citeauthoryear{Felsenstein}{1973}]{Felenstein:1973}
\begin{barticle}[author]
\bauthor{\bsnm{Felsenstein},~\bfnm{J.}\binits{J.}}
(\byear{1973}).
\btitle{Maximum Likelihood and Minimum-Steps Methods for Estimating
  Evolutionary Trees from Data on Discrete Characters}.
\bjournal{Systematic Zoology}
\bvolume{22}
\bpages{240-249}.
\end{barticle}
\endbibitem

\bibitem[\protect\citeauthoryear{Felsenstein}{1985}]{Felenstein:1985}
\begin{barticle}[author]
\bauthor{\bsnm{Felsenstein},~\bfnm{Joseph}\binits{J.}}
(\byear{1985}).
\btitle{Phylogenies and the Comparative Method}.
\bjournal{The American Naturalist}
\bvolume{125}
\bpages{1-15}.
\end{barticle}
\endbibitem

\bibitem[\protect\citeauthoryear{Fourment et~al.}{2020}]{Fourment:2020}
\begin{barticle}[author]
\bauthor{\bsnm{Fourment},~\bfnm{M.}\binits{M.}},
  \bauthor{\bsnm{Magee},~\bfnm{A.~F.}\binits{A.~F.}},
  \bauthor{\bsnm{Whidden},~\bfnm{C.}\binits{C.}},
  \bauthor{\bsnm{Bilge},~\bfnm{A.}\binits{A.}},
  \bauthor{\bsnm{Matsen},~\bfnm{F.~A.~.}\binits{F.~A.~.}} \AND
  \bauthor{\bsnm{Minin},~\bfnm{V.~N.}\binits{V.~N.}}
(\byear{2020}).
\btitle{19 Dubious Ways to Compute the Marginal Likelihood of a Phylogenetic
  Tree Topology}.
\bjournal{Systematic Biology}
\bvolume{69}
\bpages{209--220}.
\end{barticle}
\endbibitem

\bibitem[\protect\citeauthoryear{Gelman and Rubin}{1992}]{Gelman:1992}
\begin{barticle}[author]
\bauthor{\bsnm{Gelman},~\bfnm{A.}\binits{A.}} \AND
  \bauthor{\bsnm{Rubin},~\bfnm{D.~B.}\binits{D.~B.}}
(\byear{1992}).
\btitle{Inference from Iterative Simulation Using Multiple Sequences}.
\bjournal{Statistical Science}
\bvolume{7}
\bpages{457--472}.
\end{barticle}
\endbibitem

\bibitem[\protect\citeauthoryear{Goldman and Yang}{1994}]{Goldman:1994}
\begin{barticle}[author]
\bauthor{\bsnm{Goldman},~\bfnm{N.}\binits{N.}} \AND
  \bauthor{\bsnm{Yang},~\bfnm{Z.}\binits{Z.}}
(\byear{1994}).
\btitle{A codon-based model of nucleotide substitution for protein-coding {DNA}
  sequences}.
\bjournal{Molecular Biology and Evolution}
\bvolume{11}
\bpages{725--736}.
\end{barticle}
\endbibitem

\bibitem[\protect\citeauthoryear{Halpern and Bruno}{1998}]{Halpern:1998}
\begin{barticle}[author]
\bauthor{\bsnm{Halpern},~\bfnm{A.~L.}\binits{A.~L.}} \AND
  \bauthor{\bsnm{Bruno},~\bfnm{W.~J.}\binits{W.~J.}}
(\byear{1998}).
\btitle{Evolutionary distances for protein-coding sequences: modeling
  site-specific residue frequencies}.
\bjournal{Molecular Biology and Evolution}
\bvolume{15}
\bpages{910--917}.
\end{barticle}
\endbibitem

\bibitem[\protect\citeauthoryear{Hobolth and Stone}{2009}]{Hobolth:2009}
\begin{barticle}[author]
\bauthor{\bsnm{Hobolth},~\bfnm{A.}\binits{A.}} \AND
  \bauthor{\bsnm{Stone},~\bfnm{E.}\binits{E.}}
(\byear{2009}).
\btitle{Simulation from endpoint-conditioned, continuous-time {M}arkov chains
  on a finite state space, with applications to molecular evolution}.
\bjournal{Annals of Applied Statistics}
\bvolume{3}
\bpages{1204--1231}.
\end{barticle}
\endbibitem

\bibitem[\protect\citeauthoryear{Hobolth and Thorne}{2014}]{HobolthThorne:2014}
\begin{bincollection}[author]
\bauthor{\bsnm{Hobolth},~\bfnm{A.}\binits{A.}} \AND
  \bauthor{\bsnm{Thorne},~\bfnm{J.}\binits{J.}}
(\byear{2014}).
\btitle{Sampling and summary statistics of endpoint-conditioned paths in {DNA}
  sequence evolution}.
In \bbooktitle{Bayesian Phylogenetics: Methods Algorithms, and Applications}
(\beditor{\bfnm{M.~H.}\binits{M.~H.}~\bsnm{Chen}},
  \beditor{\bfnm{L.}\binits{L.}~\bsnm{Kuo}} \AND
  \beditor{\bfnm{P.}\binits{P.}~\bsnm{Lewis}}, eds.)
\bpages{247--273}.
\bpublisher{Chapman and Hall}.
\end{bincollection}
\endbibitem

\bibitem[\protect\citeauthoryear{Hwang and Green}{2004}]{Hwang:2004}
\begin{barticle}[author]
\bauthor{\bsnm{Hwang},~\bfnm{DG}\binits{D.}} \AND
  \bauthor{\bsnm{Green},~\bfnm{P}\binits{P.}}
(\byear{2004}).
\btitle{Bayesian {M}arkov chain {M}onte {C}arlo sequence analysis reveals
  varying neutral substitution patterns in mammalian evolution}.
\bjournal{Proceedings of the National Academy of Science}
\bvolume{101}
\bpages{13994-14001}.
\end{barticle}
\endbibitem

\bibitem[\protect\citeauthoryear{Jensen and Pedersen}{2000}]{Jensen:2000}
\begin{barticle}[author]
\bauthor{\bsnm{Jensen},~\bfnm{J.}\binits{J.}} \AND
  \bauthor{\bsnm{Pedersen},~\bfnm{A-MK}\binits{A.-M.}}
(\byear{2000}).
\btitle{Probabilistic Models of {DNA} Sequence Evolution with Context Dependent
  Rates of Substitution}.
\bjournal{Advances in Applied Probability}
\bvolume{32}
\bpages{499--517}.
\end{barticle}
\endbibitem

\bibitem[\protect\citeauthoryear{Jones and Hobert}{2001}]{Jones:2001}
\begin{barticle}[author]
\bauthor{\bsnm{Jones},~\bfnm{G.~L.}\binits{G.~L.}} \AND
  \bauthor{\bsnm{Hobert},~\bfnm{J.~P.}\binits{J.~P.}}
(\byear{2001}).
\btitle{Honest Exploration of Intractable Probability Distributions via
  {M}arkov Chain {M}onte {C}arlo}.
\bjournal{Statistical Science}
\bvolume{16}
\bpages{312--334}.
\end{barticle}
\endbibitem

\bibitem[\protect\citeauthoryear{Jukes and Cantor}{1969}]{Jukes:1969}
\begin{bincollection}[author]
\bauthor{\bsnm{Jukes},~\bfnm{T.~H.}\binits{T.~H.}} \AND
  \bauthor{\bsnm{Cantor},~\bfnm{C.~R.}\binits{C.~R.}}
(\byear{1969}).
\btitle{Evolution of protein molecules}.
In \bbooktitle{Mammalian Protein Metabolism}
(\beditor{\bfnm{H.~N.}\binits{H.~N.}~\bsnm{Munro}}, ed.)
\bpages{121--132}.
\bpublisher{Academic Press}, \baddress{New York}.
\end{bincollection}
\endbibitem

\bibitem[\protect\citeauthoryear{Kishino, Thorne and
  Bruno}{2001}]{Kishino:2001}
\begin{barticle}[author]
\bauthor{\bsnm{Kishino},~\bfnm{H.}\binits{H.}},
  \bauthor{\bsnm{Thorne},~\bfnm{J.~L.}\binits{J.~L.}} \AND
  \bauthor{\bsnm{Bruno},~\bfnm{W.~J.}\binits{W.~J.}}
(\byear{2001}).
\btitle{Performance of a Divergence Time Estimation Method under a
  Probabilistic Model of Rate Evolution}.
\bjournal{Molecular Biology and Evolution}
\bvolume{18}
\bpages{352-361}.
\end{barticle}
\endbibitem

\bibitem[\protect\citeauthoryear{Larson, Thorne and
  Schmidler}{2020}]{Larson:2020}
\begin{barticle}[author]
\bauthor{\bsnm{Larson},~\bfnm{Gary}\binits{G.}},
  \bauthor{\bsnm{Thorne},~\bfnm{Jeffrey~L.}\binits{J.~L.}} \AND
  \bauthor{\bsnm{Schmidler},~\bfnm{Scott~C.}\binits{S.~C.}}
(\byear{2020}).
\btitle{Incorporating Nearest-Neighbor Site Dependence into Protein Evolution
  Models}.
\bjournal{Journal of Computational Biology}
\bvolume{27}
\bpages{361-375}.
\end{barticle}
\endbibitem

\bibitem[\protect\citeauthoryear{Li, Mathews and Schmidler}{2025}]{Li:2024}
\begin{barticle}[author]
\bauthor{\bsnm{Li},~\bfnm{Y.}\binits{Y.}},
  \bauthor{\bsnm{Mathews},~\bfnm{J.}\binits{J.}} \AND
  \bauthor{\bsnm{Schmidler},~\bfnm{Scott~C.}\binits{S.~C.}}
(\byear{2025}).
\btitle{On Gibbs Sampling for Endpoint-Conditioned Neighbor-Dependent Sequence
  Evolution Models}.
\bjournal{Journal of Graphical and Computational Statistics}.
\bnote{(provisionally accepted)}.
\end{barticle}
\endbibitem

\bibitem[\protect\citeauthoryear{Li, Wiehe and Schmidler}{2024}]{Li:2025b}
\begin{barticle}[author]
\bauthor{\bsnm{Li},~\bfnm{Yongkang}\binits{Y.}},
  \bauthor{\bsnm{Wiehe},~\bfnm{Kevin}\binits{K.}} \AND
  \bauthor{\bsnm{Schmidler},~\bfnm{Scott~C.}\binits{S.~C.}}
(\byear{2024}).
\btitle{Algorithms for Reconstructing {B} Cell Lineages in the Presence of
  Context-Dependent Somatic Hypermutation}.
\bjournal{ArXiv e-prints}.
\bnote{arXiv:2512.11693 [q-bio.PE]}.
\end{barticle}
\endbibitem

\bibitem[\protect\citeauthoryear{Lov\'asz}{1999}]{Lovasz:1999}
\begin{barticle}[author]
\bauthor{\bsnm{Lov\'asz},~\bfnm{L.}\binits{L.}}
(\byear{1999}).
\btitle{Hit-and-run mixes fast}.
\bjournal{Mathematical Programming}
\bvolume{86}
\bpages{443--61}.
\end{barticle}
\endbibitem

\bibitem[\protect\citeauthoryear{Lunter and Hein}{2004}]{Lunter:2004}
\begin{barticle}[author]
\bauthor{\bsnm{Lunter},~\bfnm{G.}\binits{G.}} \AND
  \bauthor{\bsnm{Hein},~\bfnm{J.}\binits{J.}}
(\byear{2004}).
\btitle{A nucleotide substitution model with nearest-neighbour interactions}.
\bjournal{Bioinformatics}
\bvolume{20 Suppl 1}
\bpages{i216--i223}.
\end{barticle}
\endbibitem

\bibitem[\protect\citeauthoryear{Marion, Mathews and
  Schmidler}{2023a}]{Marion:2018b}
\begin{bunpublished}[author]
\bauthor{\bsnm{Marion},~\bfnm{Joseph}\binits{J.}},
  \bauthor{\bsnm{Mathews},~\bfnm{Joe}\binits{J.}} \AND
  \bauthor{\bsnm{Schmidler},~\bfnm{Scott~C.}\binits{S.~C.}}
(\byear{2023}a).
\btitle{Finite Sample ${L}_2$ Bounds for Sequential {M}onte {C}arlo and
  Adaptive Path Selection}.
\bnote{arXiv:1807.01346 [stat.CO]}.
\end{bunpublished}
\endbibitem

\bibitem[\protect\citeauthoryear{Marion, Mathews and
  Schmidler}{2023b}]{Marion:2023}
\begin{barticle}[author]
\bauthor{\bsnm{Marion},~\bfnm{J.}\binits{J.}},
  \bauthor{\bsnm{Mathews},~\bfnm{J.}\binits{J.}} \AND
  \bauthor{\bsnm{Schmidler},~\bfnm{S.~C.}\binits{S.~C.}}
(\byear{2023}b).
\btitle{Finite-Sample Complexity of Sequential Monte Carlo Estimators}.
\bjournal{Annals of Statistics}
\bvolume{51}
\bpages{1357--1375}.
\end{barticle}
\endbibitem

\bibitem[\protect\citeauthoryear{Mathews and Schmidler}{2025a}]{Mathews3:2025}
\begin{barticle}[author]
\bauthor{\bsnm{Mathews},~\bfnm{J.}\binits{J.}} \AND
  \bauthor{\bsnm{Schmidler},~\bfnm{S.~C.}\binits{S.~C.}}
(\byear{2025}a).
\btitle{Importance Sampling Approximation of Sequence Evolution Models with
  Site-Dependence}.
\bjournal{arXiv preprint arXiv:2507.19659}.
\end{barticle}
\endbibitem

\bibitem[\protect\citeauthoryear{Mathews and Schmidler}{2025b}]{Mathews2:2025}
\begin{barticle}[author]
\bauthor{\bsnm{Mathews},~\bfnm{J.}\binits{J.}} \AND
  \bauthor{\bsnm{Schmidler},~\bfnm{S.~C.}\binits{S.~C.}}
(\byear{2025}b).
\btitle{Posterior bounds on divergence time of two sequences under
  dependent-site evolutionary models}.
\bjournal{arXiv preprint arXiv:2507.19659}.
\end{barticle}
\endbibitem

\bibitem[\protect\citeauthoryear{Mathews et~al.}{2023}]{Mathews:2023b}
\begin{barticle}[author]
\bauthor{\bsnm{Mathews},~\bfnm{Joseph}\binits{J.}},
  \bauthor{\bsnm{Itallie},~\bfnm{Elizabeth~Van}\binits{E.~V.}},
  \bauthor{\bsnm{Li},~\bfnm{Yongkang}\binits{Y.}},
  \bauthor{\bsnm{Wiehe},~\bfnm{Kevin}\binits{K.}} \AND
  \bauthor{\bsnm{Schmidler},~\bfnm{Scott~C.}\binits{S.~C.}}
(\byear{2023}).
\btitle{{Computing the Inducibility of {B} Cell Lineages Under a
  Context-Dependent Model of Affinity Maturation: {A}pplications to Sequential
  Vaccine Design}}.
\bjournal{\textit{The Journal of Immunology}}.
\bnote{(in press)}.
\end{barticle}
\endbibitem

\bibitem[\protect\citeauthoryear{Pagel, Meade and Barker}{2004}]{Pagel:2004}
\begin{barticle}[author]
\bauthor{\bsnm{Pagel},~\bfnm{M.}\binits{M.}},
  \bauthor{\bsnm{Meade},~\bfnm{A.}\binits{A.}} \AND
  \bauthor{\bsnm{Barker},~\bfnm{D.}\binits{D.}}
(\byear{2004}).
\btitle{Bayesian Estimation of Ancestral Character States on Phylogenies}.
\bjournal{Systematic Biology}
\bvolume{53}
\bpages{673--684}.
\end{barticle}
\endbibitem

\bibitem[\protect\citeauthoryear{Pedersen, Wiuf and
  Christiansen}{1998}]{Pederson:1998}
\begin{barticle}[author]
\bauthor{\bsnm{Pedersen},~\bfnm{A.~K.}\binits{A.~K.}},
  \bauthor{\bsnm{Wiuf},~\bfnm{C.}\binits{C.}} \AND
  \bauthor{\bsnm{Christiansen},~\bfnm{F.~B.}\binits{F.~B.}}
(\byear{1998}).
\btitle{A codon-based model designed to describe lentiviral evolution}.
\bjournal{Molecular Biology and Evolution}
\bvolume{15}
\bpages{1069-1081}.
\end{barticle}
\endbibitem

\bibitem[\protect\citeauthoryear{Pederson and Jensen}{2001}]{Jensen:2001}
\begin{barticle}[author]
\bauthor{\bsnm{Pederson},~\bfnm{A-MK}\binits{A.-M.}} \AND
  \bauthor{\bsnm{Jensen},~\bfnm{J.}\binits{J.}}
(\byear{2001}).
\btitle{A dependent rates model and {MCMC} based methodology for the maximum
  likelihood analysis of sequences with overlapping reading frames}.
\bjournal{Molecular Biology and Evolution}
\bvolume{18}
\bpages{763--776}.
\end{barticle}
\endbibitem

\bibitem[\protect\citeauthoryear{Robinson et~al.}{2003}]{Robinson:2003}
\begin{barticle}[author]
\bauthor{\bsnm{Robinson},~\bfnm{D.}\binits{D.}},
  \bauthor{\bsnm{Jones},~\bfnm{D.}\binits{D.}},
  \bauthor{\bsnm{Kishino},~\bfnm{H.}\binits{H.}},
  \bauthor{\bsnm{Goldman},~\bfnm{N.}\binits{N.}} \AND
  \bauthor{\bsnm{Thorne},~\bfnm{J.}\binits{J.}}
(\byear{2003}).
\btitle{Protein Evolution with Dependence Among Codons Due to Tertiary
  Structure}.
\bjournal{Molecular Biology and Evolution}
\bvolume{20}
\bpages{1692--1704}.
\end{barticle}
\endbibitem

\bibitem[\protect\citeauthoryear{Rodrigue, Philippe and
  Lartillot}{2006}]{Rodrigue:2006}
\begin{barticle}[author]
\bauthor{\bsnm{Rodrigue},~\bfnm{N.}\binits{N.}},
  \bauthor{\bsnm{Philippe},~\bfnm{H.}\binits{H.}} \AND
  \bauthor{\bsnm{Lartillot},~\bfnm{N.}\binits{N.}}
(\byear{2006}).
\btitle{Assessing site-interdependent phylogenetic models of sequence
  evolution}.
\bjournal{Molecular Biology and Evolution}
\bvolume{23}
\bpages{1762-1775}.
\end{barticle}
\endbibitem

\bibitem[\protect\citeauthoryear{Rodrigue et~al.}{2005}]{Rodrigue:2005}
\begin{barticle}[author]
\bauthor{\bsnm{Rodrigue},~\bfnm{N.}\binits{N.}},
  \bauthor{\bsnm{Lartillot},~\bfnm{N.}\binits{N.}},
  \bauthor{\bsnm{Bryant},~\bfnm{D.}\binits{D.}} \AND
  \bauthor{\bsnm{Philippe},~\bfnm{H.}\binits{H.}}
(\byear{2005}).
\btitle{Site interdependence attributed to tertiary structure in amino acid
  sequence evolution}.
\bjournal{Gene}
\bvolume{347}
\bpages{207-217}.
\end{barticle}
\endbibitem

\bibitem[\protect\citeauthoryear{Rodríguez et~al.}{1990}]{Rodriguez:1990}
\begin{barticle}[author]
\bauthor{\bsnm{Rodríguez},~\bfnm{F.}\binits{F.}},
  \bauthor{\bsnm{Oliver},~\bfnm{J.~L.}\binits{J.~L.}},
  \bauthor{\bsnm{Marín},~\bfnm{A.}\binits{A.}} \AND
  \bauthor{\bsnm{Medina},~\bfnm{J.~R.}\binits{J.~R.}}
(\byear{1990}).
\btitle{The general stochastic model of nucleotide substitution}.
\bjournal{Journal of Theoretical Biology}
\bvolume{142}
\bpages{485--501}.
\end{barticle}
\endbibitem

\bibitem[\protect\citeauthoryear{Ronquist et~al.}{2012}]{MrBayes:2012}
\begin{barticle}[author]
\bauthor{\bsnm{Ronquist},~\bfnm{F.}\binits{F.}},
  \bauthor{\bsnm{Teslenko},~\bfnm{M.}\binits{M.}}, \bauthor{\bsnm{{van der
  Mark}},~\bfnm{P.}\binits{P.}},
  \bauthor{\bsnm{Ayres},~\bfnm{D.~L.}\binits{D.~L.}},
  \bauthor{\bsnm{Darling},~\bfnm{A.}\binits{A.}},
  \bauthor{\bsnm{H\"ohna},~\bfnm{S.}\binits{S.}},
  \bauthor{\bsnm{Larget},~\bfnm{B.}\binits{B.}},
  \bauthor{\bsnm{Liu},~\bfnm{L.}\binits{L.}},
  \bauthor{\bsnm{Suchard},~\bfnm{M.~A.}\binits{M.~A.}} \AND
  \bauthor{\bsnm{Huelsenbeck},~\bfnm{J.~P.}\binits{J.~P.}}
(\byear{2012}).
\btitle{{MrBayes} 3.2: Efficient {Bayesian} Phylogenetic Inference and Model
  Choice Across a Large Model Space}.
\bjournal{Systematic Biology}
\bvolume{61}
\bpages{539--542}.
\end{barticle}
\endbibitem

\bibitem[\protect\citeauthoryear{Rosenthal}{1995}]{Rosenthal:1995}
\begin{barticle}[author]
\bauthor{\bsnm{Rosenthal},~\bfnm{J.~S.}\binits{J.~S.}}
(\byear{1995}).
\btitle{Minorization Conditions and Convergence Rates for {M}arkov Chain
  {M}onte {C}arlo}.
\bjournal{Journal of the American Statistical Association}
\bvolume{90}
\bpages{558--566}.
\end{barticle}
\endbibitem

\bibitem[\protect\citeauthoryear{Sanderson}{1997}]{Sanderson:1997}
\begin{barticle}[author]
\bauthor{\bsnm{Sanderson},~\bfnm{MJ.}\binits{M.}}
(\byear{1997}).
\btitle{A Nonparametric Approach to Estimating Divergence Times in the Absence
  of Rate Constancy}.
\bjournal{Molecular Biology and Evolution}
\bvolume{14}
\bpages{1218}.
\end{barticle}
\endbibitem

\bibitem[\protect\citeauthoryear{Siepel and Haussler}{2004}]{Siepel:2004}
\begin{barticle}[author]
\bauthor{\bsnm{Siepel},~\bfnm{A.}\binits{A.}} \AND
  \bauthor{\bsnm{Haussler},~\bfnm{D.}\binits{D.}}
(\byear{2004}).
\btitle{Phylogenetic Estimation of Context-Dependent Substitution Rates by
  Maximum Likelihood}.
\bjournal{Molecular Biology and Evolution}
\bvolume{21}
\bpages{468--488}.
\end{barticle}
\endbibitem

\bibitem[\protect\citeauthoryear{Tavaré}{1986}]{Tavare:1986}
\begin{barticle}[author]
\bauthor{\bsnm{Tavaré},~\bfnm{Simon}\binits{S.}}
(\byear{1986}).
\btitle{Some Probabilistic and Statistical Problems in the Analysis of {DNA}
  Sequences}.
\bjournal{Lectures on Mathematics in the Life Sciences}
\bvolume{17}
\bpages{57--86}.
\end{barticle}
\endbibitem

\bibitem[\protect\citeauthoryear{Thorne, Kishino and
  Painter.}{1998}]{Thorne:1998}
\begin{barticle}[author]
\bauthor{\bsnm{Thorne},~\bfnm{J.~L.}\binits{J.~L.}},
  \bauthor{\bsnm{Kishino},~\bfnm{H.}\binits{H.}} \AND
  \bauthor{\bsnm{Painter.},~\bfnm{I.~S.}\binits{I.~S.}}
(\byear{1998}).
\btitle{Estimating the rate of evolution of the rate of molecular evolution}.
\bjournal{Molecular Biology and Evolution}
\bvolume{15}
\bpages{1647-1657}.
\end{barticle}
\endbibitem

\bibitem[\protect\citeauthoryear{VanDerwerken and
  Schmidler}{2013}]{Vandewerken:2013}
\begin{barticle}[author]
\bauthor{\bsnm{VanDerwerken},~\bfnm{D.}\binits{D.}} \AND
  \bauthor{\bsnm{Schmidler},~\bfnm{S.~C.}\binits{S.~C.}}
(\byear{2013}).
\btitle{Parallel {M}arkov Chain {M}onte {C}arlo}.
\bjournal{arXiv preprint}.
\end{barticle}
\endbibitem

\bibitem[\protect\citeauthoryear{VanDerwerken and
  Schmidler}{2017}]{Vandewerken:2017}
\begin{barticle}[author]
\bauthor{\bsnm{VanDerwerken},~\bfnm{D.}\binits{D.}} \AND
  \bauthor{\bsnm{Schmidler},~\bfnm{S.~C.}\binits{S.~C.}}
(\byear{2017}).
\btitle{Monitoring Joint Convergence of {MCMC} Samplers}.
\bjournal{Journal of Computational and Graphical Statistics}
\bvolume{26}
\bpages{558--568}.
\end{barticle}
\endbibitem

\bibitem[\protect\citeauthoryear{Vempala}{2005}]{Vempala:2005}
\begin{barticle}[author]
\bauthor{\bsnm{Vempala},~\bfnm{S.}\binits{S.}}
(\byear{2005}).
\btitle{Geometric Random Walks: A Survey}.
\bjournal{Combinatorial and Computational Geometry}
\bvolume{52}
\bpages{573--612}.
\end{barticle}
\endbibitem

\bibitem[\protect\citeauthoryear{von Haeseler and
  Schöniger}{1998}]{VonHaeseler:1998}
\begin{barticle}[author]
\bauthor{\bparticle{von} \bsnm{Haeseler},~\bfnm{A.}\binits{A.}} \AND
  \bauthor{\bsnm{Schöniger},~\bfnm{M.}\binits{M.}}
(\byear{1998}).
\btitle{Evolution of {DNA} or amino acid sequences with dependent sites}.
\bjournal{Journal of Computational Biology}
\bvolume{5}
\bpages{149-163}.
\end{barticle}
\endbibitem

\bibitem[\protect\citeauthoryear{Wiehe et~al.}{2018}]{Wiehe:2018}
\begin{barticle}[author]
\bauthor{\bsnm{Wiehe},~\bfnm{K.}\binits{K.}},
  \bauthor{\bsnm{Bradley},~\bfnm{T.}\binits{T.}},
  \bauthor{\bsnm{Meyerhoff},~\bfnm{RR.}\binits{R.}},
  \bauthor{\bsnm{Hart},~\bfnm{C.}\binits{C.}},
  \bauthor{\bsnm{Williams},~\bfnm{WB.}\binits{W.}},
  \bauthor{\bsnm{Easterhoff},~\bfnm{D.}\binits{D.}},
  \bauthor{\bsnm{Faison},~\bfnm{WJ.}\binits{W.}},
  \bauthor{\bsnm{Kepler},~\bfnm{TB.}\binits{T.}},
  \bauthor{\bsnm{Saunders},~\bfnm{KO.}\binits{K.}},
  \bauthor{\bsnm{Alam},~\bfnm{SM.}\binits{S.}},
  \bauthor{\bsnm{Bonsignori},~\bfnm{M.}\binits{M.}} \AND
  \bauthor{\bsnm{Haynes},~\bfnm{BF.}\binits{B.}}
(\byear{2018}).
\btitle{Functional Relevance of Improbable Antibody Mutations for {HIV} Broadly
  Neutralizing Antibody Development}.
\bjournal{Cell Host Microbe}
\bvolume{23}
\bpages{759--765}.
\end{barticle}
\endbibitem

\bibitem[\protect\citeauthoryear{Wolpert and Schmidler}{2012}]{Wolpert:2012}
\begin{barticle}[author]
\bauthor{\bsnm{Wolpert},~\bfnm{R.~L.}\binits{R.~L.}} \AND
  \bauthor{\bsnm{Schmidler},~\bfnm{S.~C.}\binits{S.~C.}}
(\byear{2012}).
\btitle{$\alpha$-Stable Limit Laws for Harmonic Mean Estimators of Marginal
  Likelihoods}.
\bjournal{Statistica Sinica}
\bvolume{22}
\bpages{1233-1251}.
\end{barticle}
\endbibitem

\bibitem[\protect\citeauthoryear{Yaari et~al.}{2013}]{Yaari:2013}
\begin{barticle}[author]
\bauthor{\bsnm{Yaari},~\bfnm{G.}\binits{G.}},
  \bauthor{\bsnm{Vander~Heiden},~\bfnm{J.~A.}\binits{J.~A.}},
  \bauthor{\bsnm{Uduman},~\bfnm{M.}\binits{M.}},
  \bauthor{\bsnm{Gadala-Maria},~\bfnm{D.}\binits{D.}},
  \bauthor{\bsnm{Gupta},~\bfnm{N.}\binits{N.}},
  \bauthor{\bsnm{Stern},~\bfnm{J.~N.}\binits{J.~N.}},
  \bauthor{\bsnm{O'Connor},~\bfnm{K.~C.}\binits{K.~C.}},
  \bauthor{\bsnm{Hafler},~\bfnm{D.~A.}\binits{D.~A.}},
  \bauthor{\bsnm{Laserson},~\bfnm{U.}\binits{U.}},
  \bauthor{\bsnm{Vigneault},~\bfnm{F.}\binits{F.}} \AND
  \bauthor{\bsnm{Kleinstein},~\bfnm{S.~H.}\binits{S.~H.}}
(\byear{2013}).
\btitle{Models of Somatic Hypermutation Targeting and Substitution Based on
  Synonymous Mutations from High-Throughput Immunoglobulin Sequencing Data}.
\bjournal{Frontiers in Immunology}
\bvolume{4}
\bpages{358}.
\end{barticle}
\endbibitem

\bibitem[\protect\citeauthoryear{Yang}{1994}]{Yang:1994}
\begin{barticle}[author]
\bauthor{\bsnm{Yang},~\bfnm{Z.}\binits{Z.}}
(\byear{1994}).
\btitle{Maximum likelihood phylogenetic estimation from DNA sequences with
  variable rates over sites: Approximate methods}.
\bjournal{Journal of Molecular Evolution}
\bvolume{39}
\bpages{306--314}.
\end{barticle}
\endbibitem

\bibitem[\protect\citeauthoryear{Yang, Kumar and Nei}{1995}]{Yang:1995}
\begin{barticle}[author]
\bauthor{\bsnm{Yang},~\bfnm{Z.}\binits{Z.}},
  \bauthor{\bsnm{Kumar},~\bfnm{S.}\binits{S.}} \AND
  \bauthor{\bsnm{Nei},~\bfnm{M.}\binits{M.}}
(\byear{1995}).
\btitle{A new method of inference of ancestral nucleotide and amino acid
  sequences}.
\bjournal{Genetics}
\bvolume{141}
\bpages{1641-1650}.
\end{barticle}
\endbibitem

\bibitem[\protect\citeauthoryear{Yang and Nielsen}{2008}]{Yang:2008}
\begin{barticle}[author]
\bauthor{\bsnm{Yang},~\bfnm{Z.}\binits{Z.}} \AND
  \bauthor{\bsnm{Nielsen},~\bfnm{R.}\binits{R.}}
(\byear{2008}).
\btitle{Mutation-selection models of codon substitution and their use to
  estimate selective strengths on codon usage}.
\bjournal{Molecular Biology and Evolution}
\bvolume{25}
\bpages{568--579}.
\end{barticle}
\endbibitem

\end{thebibliography}
